\documentclass[10pt,a4paper,twoside]{article}
\usepackage{bbm}
\usepackage{color}
\usepackage{mathtools}
\usepackage[dvipsnames,svgnames,table]{xcolor}
\usepackage{amsmath,amssymb,mathrsfs,amsthm}
\usepackage{algorithm}
\usepackage{algpseudocode}
\usepackage{lipsum}
\usepackage[margin=2.5cm]{geometry}
\usepackage{lineno}
\usepackage{natbib}
\usepackage{hyperref}
\usepackage{caption}

\DeclareMathOperator*{\argmin}{arg\,min}
\newtheorem{definition}{Definition}[section] 
\newtheorem{assumption}[definition]{Assumption}
\newtheorem{condition}[definition]{Condition}
\newtheorem{theorem}{Theorem}[section]
\newtheorem{proposition}[theorem]{Proposition}
\newtheorem{lemma}[theorem]{Lemma}
\newtheorem{corollary}[theorem]{Corollary}

\usepackage{float}
\usepackage{authblk}
\usepackage{graphicx}
\usepackage{setspace}
\usepackage{hyperref}
\usepackage{blindtext}
\usepackage{parskip}
\newcommand{\bbR}{\mathbb{R}}
\newcommand{\bbE}{\mathbb{E}}
\newcommand{\al}{\alpha}
\newcommand{\rmd}{\,{\rm d}\,}

\renewcommand{\argmin}{{\rm arg}\,\min}

\newcommand{\T}{\mathrm{\scriptscriptstyle T} }

\newcommand{\indp}{\!\perp\!\!\!\perp}
\newcommand{\nindp}{\not\!\perp\!\!\!\perp}

\newcommand\restr[2]{{
  \left.\kern-\nulldelimiterspace 
  #1
  \littletaller 
  \right|_{#2}
  }}
\newcommand{\littletaller}{\mathchoice{\vphantom{\big|}}{}{}{}}
\usepackage{tcolorbox}
\allowdisplaybreaks[4]
\usepackage[english]{babel}
\usepackage[nottoc]{tocbibind}

\title{\Large A Generalized Framework for Approximate Co-Sufficient Sampling}

\author{Jie Xie}
\author{Dongming Huang}
\affil{Department of Statistics and Data Science, National University of Singapore}
\date{\today}
\begin{document}
\maketitle

\begin{abstract}
Approximate co-sufficient sampling (aCSS) offers a principled route to hypothesis testing when null distributions are unknown, yet current implementations are confined to maximum likelihood estimators with smooth or linear regularization and provide little theoretical insight into power. We present a generalized framework that widens the scope of the aCSS method to embrace nonlinear regularization, such as group lasso and nonconvex penalties, as well as robust and nonparametric estimators. Moreover, we introduce a weighted sampling scheme for enhanced flexibility and propose a generalized aCSS framework that unifies existing conditional sampling methods. Our theoretical analysis rigorously establishes validity and, for the first time, characterizes the power optimality of aCSS procedures in certain high-dimensional settings.
\end{abstract}

\noindent{\bf Keywords}: 
Goodness-of-fit test; Approximate sufficiency; Conditional randomization test; High-dimensional inference; Model-X.

\section{Introduction}\label{section 1}
In statistics and machine learning, many inference problems can be cast as the following hypothesis testing for data \(X\) drawn from some space \(\mathcal{X} \subseteq \mathbb{R}^n\):
\begin{equation}\label{eq: parametric model}
H_0: X \sim P_{\theta} \text{ for some } \theta \in \Theta \subseteq \mathbb{R}^d,
\end{equation}
where \(\{P_{\theta} : \theta \in \Theta\}\) represents a parametric family of distributions on $\mathcal{X}$. 
This setup is typical for Goodness-of-Fit (GoF) testing,  which aims to determine whether a hypothesized model accurately captures the underlying distribution of data. 
Beyond GoF, this setup also facilitates diverse applications such as constructing confidence intervals, model selection, and model-X conditional independence testing, which arise across diverse fields, from genomics to social sciences; see the discussion in \citet{barber2022testing}. 

To address such problems, a test statistic \(T = T(X)\) is often employed, where larger values of \(T\) indicate stronger evidence against the null hypothesis \(H_0\). 
If the (asymptotic) distribution of \(T\) under \(H_0\) is known, hypothesis tests can be formalized directly \citep{shao2003mathematical, van2000asymptotic}. 
However, deriving these null distributions is feasible only in limited cases. 

Resampling methods, such as bootstrap \citep{efron1992bootstrap} and permutation tests \citep{good2013permutation}, address this by generating copies \(\widetilde{X}^{(1)}, \ldots, \widetilde{X}^{(M)}\) and computing the p-value as 
\begin{equation}\label{pval}
    \text{pval} = \text{pval}_T(X, \widetilde{X}^{(1)}, \ldots, \widetilde{X}^{(M)}) = \frac{1}{M+1}\left(1 + \sum_{m=1}^M {1}\{T(\widetilde{X}^{(m)}) \geq T(X)\}\right).
\end{equation}
If the sample \(X\) and its copies \(\widetilde{X}^{(1)}, \ldots, \widetilde{X}^{(M)}\) are exchangeable under the null hypothesis \(H_0\), the p-value defined in \eqref{pval} is valid, meaning that \(\mathbb{P}_{H_0}(\text{pval} \leq \alpha) \leq \alpha\) for any significance level \(\alpha \in (0,1)\). 

These approaches reformulate the testing problem as a sampling problem for exchangeable copies. 
In the special case when $\Theta=\{\theta_0\}$ is a singleton, the copies \(\widetilde{X}^{(m)}\) could be sampled as independent and identically distributed (i.i.d.)  copies from \(P_{\theta_0}\).

When $\Theta$ is not a singleton, co-sufficient sampling (CSS) \citep{stephens2012goodness} removes 
the dependence on $\theta$ by leveraging sufficient statistics. 
Specifically, if \(S(X)\) is sufficient for \(\theta\in \Theta\), then the conditional distribution of \(X \mid S(X)\) does not depend on \(\theta\) and i.i.d. samples from this conditional distribution will be exchangeable with $X$. 
Although CSS is theoretically appealing, it is often impractical because the conditional distribution may collapse to a point mass at \(X\). 
To address these issues, \cite{barber2022testing} introduced approximate co-sufficient sampling (aCSS), which replaces exact sufficiency with approximate sufficiency by conditioning on perturbed maximum likelihood estimators (MLEs) instead of sufficient statistics. \cite{zhu2023approximate} extended the aCSS framework to incorporate linear regularization, such as the lasso estimator \citep{tibshirani1996regression}, thereby broadening its applicability to high-dimensional settings.

Though novel and useful, 
aCSS in its original and linearly-regularized forms remains limited in scope.
The existing aCSS framework is largely confined to regularized MLEs with simple penalties (either linear or continuously twice differentiable). 
This excludes the vast spectrum of regularization techniques such as group lasso \citep{yuan2006model} and nonconvex penalties \citep{fan2001variable}, which are often used to perform efficient high-dimensional regression. 
The existing aCSS framework also prohibits the use of other estimators, such as robust regression techniques against extreme observations and sieve-type estimators in nonparametric additive models \citep{hastie2017generalized}. 
Furthermore, it require the penalty parameter to be independent of \(X\), which deviates from common practice. 
Lastly, although the approximate validity of the aCSS methods has been established, the lack of theoretical investigation into their power remains a critical limitation, as validity alone does not ensure efficient testing. 

\textit{Our contributions.}
This paper addresses the aforementioned limitations by fundamentally enhancing the aCSS framework through a series of achievements: 
(1) We extend the aCSS framework to incorporate more complex penalties like group lasso, SCAD, and MCP, which are important regularization methods in high-dimensional statistics. 
Extensions to general smooth constrained and $\ell_p$-norm constrained MLEs are also developed. 
(2) We extend the framework to encompass estimators beyond MLEs, including maximum trimmed likelihood estimators, quantile estimators, and sieve-type estimators in nonparametric additive models. 
(3) We integrate aCSS with the Model-X Conditional Randomization Test (CRT) for testing conditional independence and establish a theory on power optimality, which provides theoretical guidance for balancing the Type-I error and power. 
Moreover, we proposes a generalization of the aCSS method, referred to as the \textit{generalized aCSS method}, which unifies CSS methods, aCSS methods, and other related methods. 

Our work is closely related to \cite{zhu2023approximate}, which incorporates linear regularization into the aCSS methods. 
We emphasize that the extension from linear to nonlinear regularization is highly nontrivial and requires the development of new techniques. 
Take the constrained case as an example: under linear constraints, the boundary of the feasible set is formed by hyperplanes (or the intersections of multiple hyperplanes), which are simple to analyze, whereas for general smooth constraints, the boundary forms a manifold, which introduces significant challenges. 
To overcome these challenges, we employ the rank constant theorem \citep{rudin1976principles} to characterize such a manifold and identify conditions for characterizing local optimas of the constrained optimization.  

Our power analysis is, to our knowledge, the first theoretical investigation on the efficiency of the aCSS method. 
It reveals several noteworthy insights: (1) We can leverage the unlabeled data to improve the Type-I error control of the aCSS methods. (2) In settings where the sample size and model dimension grow proportionally, the aCSS method becomes asymptotically equivalent to the CSS methods, thereby achieving parametric efficiency. (3) In the ultra-high dimensional regime, we introduce a novel test statistic under the aCSS CRT framework that attains the minimax optimal rate. 

We use the following notation throughout this paper. For an integer $p \geq 1$, let $[p] = \{1,\ldots,p\}$. 
For a vector $\nu \in \mathbb{R}^k$, let $\text{supp}(\nu) = \{i \in [k] : \nu_i \neq 0\}$ denote its support. 
The Euclidean norm is denoted by $\|\nu\|$, the  $\ell_p$-norm is $\|\nu\|_p$, and $\|\nu\|_0$ counts the nonzero entries.
For a subset $S \subseteq [k]$, $\nu_S$ is the subvector indexed by $S$. 
For a matrix $M$ and two index sets $I$ and $J$, $M_{I,J}$ is the submatrix formed by rows in $I$ and columns in $J$. Also define $M_J := M_{J,J}$. Let $\lambda_{\max}(M)$ denotes its largest eigenvalue in the positive direction. 
\({\rm d}_{TV}\) denotes the total variation distance. 
Let \(\Phi\) represent the cumulative distribution function of the standard normal distribution \({N}(0, 1)\) and $z_{\alpha}$ satisfies $\Phi(z_{\alpha}) = \alpha$. $\mathbb{E}_\theta$ denotes the expectation under the distribution $P_\theta$.
Whenever there is no confusion, we denote by \(\theta_0 \in \Theta\) the true parameter value under $H_0$.
Proofs of the theoretical results are provided in the Supplementary Material. 

\section{aCSS with general regularized MLE}\label{section 2}
This section first reviews the aCSS methods developed in \cite{barber2022testing} and \cite{zhu2023approximate}, and then extends the aCSS framework to incorporate nonlinear regularization. 
\subsection{Review: Original aCSS method}
Throughout the paper, we make the following assumption about the family \(\{P_\theta : \theta \in \Theta\}\): \(\Theta\) is a convex, open subset of \(\mathbb{R}^d\), and for each \(\theta \in \Theta\), \(P_\theta\) has a positive density \(f(x; \theta)\) with respect to a \(\sigma\)-finite measure \(\nu_{\mathcal{X}}\). Additionally, \(f(x; \theta)\) is twice continuously differentiable in \(\theta\) for any fixed \(x\). 

The aCSS method defines the perturbed MLE $\hat{\theta}$ as the minimizer of the objective function 
\begin{equation}\label{random MLE}
   {\cal L}(\theta; X, W) :=\mathcal{L}(\theta;X) + \sigma W^T \theta,
\end{equation} 
where $\mathcal{L}(\theta;X) := -\log f(X; \theta) + \mathcal{R}(\theta)$ is the negative log likelihood with an optional twice-differentiable penalty function \(\mathcal{R}(\theta)\), \(W\sim {N}(0,{I}_d/d)\) is a noise vector, and \(\sigma > 0\) controls the magnitude of perturbation. 
Instead of requiring an exact minimizer of the objective function, we assume that the solver used to obtain $\hat{\theta}$ is designed to locate \textit{strict second-order stationary points} (SSOSPs) of \(\mathcal{L}(\theta; x, w)\). Specifically, a parameter \(\theta\) is an SSOSP of \(\mathcal{L}(\theta; x, w)\) if it satisfies: (i) the first-order condition: \(\nabla_\theta \mathcal{L}(\theta; x, w) = \nabla_\theta \mathcal{L}(\theta; x) + \sigma w = 0\), and (ii) the strict second-order condition: \(\nabla^2_\theta \mathcal{L}(\theta; x, w)=\nabla^2_\theta \mathcal{L}(\theta; x, w) \succ 0\).
We write this perturbed MLE as \(\hat{\theta}(X,W): \mathcal{X} \times \mathbb{R}^d \to \Theta\), which is a random variable.

For \(\theta \in \Theta\), let \(P_{\theta}(\cdot \mid \hat{\theta})\) denote the conditional distribution of \(X \mid \hat{\theta}\) if \(X \sim P_{\theta}\). 
Under \(H_0\), we have \(X \mid \hat{\theta} \sim P_{\theta_0}(\cdot \mid \hat{\theta})\).
By replacing \(\theta_0\) with \(\hat{\theta}\), we generate samples \(\widetilde{X}^{(m)} \sim P_{\hat{\theta}}(\cdot \mid \hat{\theta})\) for \(m=1,\ldots,M\).
While \(X\) and \(\widetilde{X}^{(m)}\) may not be exchangeable, \cite{barber2022testing} introduce the following measurement to quantify approximate exchangeability. 

\begin{definition}[Distance to Exchangeability]\label{Distance to Exchangeability} For r.v.s $\{A_i\}_{i=1}^k$, define
\[
{\rm d}_{\rm exch}(A_1, \ldots, A_k) := \inf\{{\rm d}_{TV}((A_1, \ldots, A_k), (B_1, \ldots, B_k)) : B_1, \ldots, B_k \text{ are exchangeable}\}.
\]
\end{definition} 
The p-value defined in \eqref{pval} satisfies that
\begin{equation}\label{approximate valid}
\mathbb{P}\bigl(\text{pval}_{T}(X, \widetilde{X}^{(1)}, \ldots, \widetilde{X}^{(M)}) \le \alpha\bigr) \le \alpha + \mathrm{d}_{\mathrm{exch}}(X, \widetilde{X}^{(1)}, \ldots, \widetilde{X}^{(M)}),
\end{equation}
and thus is an approximately valid p-value if $\mathrm{d}_{\mathrm{exch}}(X, \widetilde{X}^{(1)}, \ldots, \widetilde{X}^{(M)})$ is small. 
\cite{barber2022testing} establish an upper bound on the quantity \(\mathrm{d}_{\mathrm{exch}}\) under the assumption that the penalty function \(\mathcal{R}(\theta)\) is twice differentiable. 
However, this assumption is not met by many common penalty functions, such as the \(\ell_1\)-penalty in the (perturbed) lasso estimator (\cite{tibshirani1996regression}) given by 
\begin{equation}\label{lasso}
 \hat{\theta} = \hat{\theta}_\lambda(X,W) = \argmin_{\theta \in \Theta} \mathcal{L}(\theta; X, W)  + \lambda \|\theta\|_1, 
\end{equation}
since $\|\theta\|_1$ is not everywhere differentiable. 
For the estimator defined in \eqref{lasso}, \cite{zhu2023approximate} propose to not only condition on the estimator \(\hat{\theta} = \hat{\theta}(X,W)\), but also condition on the gradient 
\begin{equation}\label{gradient lasso}
    \hat{g} = \hat{g}(X,W) := \nabla_\theta \mathcal{L}(\hat{\theta}; X, W) = \nabla_\theta \mathcal{L}(\hat{\theta};X) + \sigma W. 
\end{equation}
The analysis for conditioning on $(\hat{\theta}, \hat{g})$ becomes possible if we partition the parameter space \( \Theta = \mathbb{R}^d \) according to all possible supports of \( \theta \). More concretely, within each partition, the lasso penalty is locally linear, and therefore one can derive the conditional distribution \( P_{\theta_0}(\cdot \mid \hat{\theta}, \hat{g}) \). 
Replacing \( \theta_0 \) with \( \hat{\theta} \), the copies \( \widetilde{X}^{(m)}\) are then drawn from $P_{\hat{\theta}}(\cdot \mid \hat{\theta}, \hat{g})$. 
Based on these insights, they establish the aCSS framework with linear regularization.

In summary, the general recipe for the aCSS method is outlined below:
\begin{algorithm}
   \caption{General Recipe for aCSS Method.}
   
   \indent (i) Adaptively define SSOSP suited to the specific setting.
   
   \indent (ii) Generate a noise vector \(W\), compute the statistic \(\hat{\theta} = \hat{\theta}(X, W)\) and $\hat{g}=\hat{g}(X,W)$, and determine the true conditional density \(p_{\theta_0}(\cdot \mid \hat{\theta},\hat{g})\).
   
    \indent (iii) If $\hat{\theta}$ is an SSOSP, then replace \(\theta_0\) with \(\hat{\theta}\), and sample copies \(\widetilde{X}^{(m)}\) from \(p_{\hat{\theta}}(\cdot \mid \hat{\theta},\hat{g})\); otherwise, set \(\widetilde{X}^{(m)}=X\). 
\end{algorithm}

The original aCSS method \citep{barber2022testing} can be regarded as a special case of this general receipe with \(\hat{g}(X,W)\equiv 0\), where \(p_{\theta_0}(\cdot \mid \hat{\theta},\hat{g})\) simplifies to \(p_{\theta_0}(\cdot \mid \hat{\theta})\). 
\subsection{aCSS with Nonlinear Regularization}\label{penalty section}
We extend the aCSS framework to incorporate nonlinear regularization, including both the penalized and constrained forms. 
The main challenges are redefining SSOSPs and deriving conditional densities \(p_{\theta_0}(\cdot \mid \hat{\theta}, \hat{g})\). 

For nonlinear penalties, we consider the following flexible formulation:  
Let \(G = \{G_1, \ldots, G_J\}\) be a partition of \([d]\), where each \(G_j\) represents a group of variables. 
\(G\) can represent either predefined groups (e.g., for group lasso) or individual variables (by setting \(G_j = \{j\}\) for all $j\in [d]$).  
We formulate the perturbed MLE with the group-wise penalty as 
\begin{equation}\label{grouppenalty}
\hat{\theta} = \argmin_{\theta \in \Theta} \mathcal{L}(\theta; X, W) + \sum_{j=1}^J \rho_j\bigl(\|\theta_{G_j}\|\bigr),
\end{equation}
where \(\mathcal{L}(\theta; X, W)\) is defined as in \eqref{random MLE} and $\rho_j$ are the penalty functions imposed on $j$th group.
This group-wise penalty formulation encompasses a wide range of penalties, such as lasso and group-lasso. It reduces to a standard coordinate-wise penalty when each \(G_j\) is a singleton. We impose the following mild standard conditions on each penalty function \(\rho_j\). 
\begin{assumption}[Penalty Functions]\label{assump:4}
(i) \(\rho_j\) is non-decreasing on \([0,\infty)\);
(ii) \(\rho_j\) is continuously differentiable with derivative \(\rho'_j\), and the limit \(\rho'_j(0) = \lim_{t \to 0^+} {\rho_j(t)}/{t}\) exists; 
(iii) \(\rho'_j\) is continuously differentiable on \((0,\infty)\) almost everywhere with derivative \(\rho''_j\) where it exists.
\end{assumption}

Unlike \cite{zhu2023approximate}, which only considers aCSS with linear regularization, Assumption \ref{assump:4} accommodates a much broader class of penalty functions, including group lasso \citep{yuan2006model}, SCAD \citep{fan2001variable}, and MCP \citep{zhang2010nearly}. 
Though the optimization is often more complex, 
many empirical studies have demonstrated that the resulting estimators have smaller estimation errors compared to those obtained using the lasso penalty \citep{10.1214/009053607000000802,10.1214/10-AOAS388, simon2013sparse}.

To deal with the non-differentiability of $\rho_j$ at 0, we introduce the following notations:
For \(\theta \in \mathbb{R}^d\), denote by \(\mathcal{A}(\theta) = \{ j \in [J] : \|\theta_{G_j}\| \neq 0 \}\) the corresponding active groups, and by \(\mathcal{S}(\theta) = \cup_{j \in \mathcal{A}(\theta)} G_j\) the corresponding active coordinates. For convenience, let \({I}_{d,j}\) be the \(d \times d_{G_j}\) submatrix of the \(d \times d\) identity matrix obtained by selecting those columns associated with the indices in \(G_j\).
Now we define the SSOSP for the penalized problem 
 \eqref{grouppenalty}. 
\begin{definition}[SSOSP for penalized estimation]\label{definition for penalty}
    A parameter $\theta\in\Theta$ is a strict second-order stationary point (SSOSP) of the penalized estimation \eqref{grouppenalty} if it satisfies the followings:
    
    \indent (i) {Regularity:} For all $j\in [J]$ either $\|\theta_{G_j}\| = 0$ or $\rho''(\|\theta_{G_j}\|)$ exists.
    
    \indent (ii) {First-Order Condition:}
        $$\nabla_\theta {\cal L}(\theta;X,W)+\sum_{j\in [J]}{I}_{d,j}s_j(\theta_{G_j}) =0,$$
       \indent where $s_j(\theta)\in \bbR^{d_{G_j}}$ satisfying that
       $$\left\{
        \begin{array}{cc}
     s_j(\theta_{G_j})=\rho_j^\prime(\theta_{G_j})\frac{\theta_{G_j}}{\left\|\theta_{G_j}\right\|},   ~~& \text{if}\,\left\|\theta_{G_j}\right\|\neq {0}\\
          \left\|s_j(\theta_{G_j})\right\|\leq \rho_j^\prime(0),    ~~&  \text{if}\,\left\|\theta_{G_j}\right\|={0}
        \end{array}\right.
        $$
      
      \indent (iii) {Second-Order Condition:}
        $$\left(\nabla_\theta^2 \mathcal{L}(\theta ;X)+\sum_{j\in {\cal A}(\theta)}{I}_{d,j}s_j^\prime (\theta_{G_j}){I}_{d,j}^\T\right)_{{\cal S}(\theta)}\succ 0.
      $$
\end{definition}

When all the penalty functions $\rho_j$ are twice differentiable, conditions (ii) and (iii) in Definition \ref{definition for penalty} simplify to the first- and second-order conditions for local optimality. 
Therefore, these conditions can be interpreted as an adaptation that accounts for the non-differentiability of the penalty functions $\rho_j$ at the point $0$. For $\theta$ being an SSOSP of \eqref{grouppenalty}, define a positive function on $\bbR^d\times {\cal X}$ as follows:
\begin{equation}\label{F_pen}
         {F}_{\rm pen}(\theta;X):=  \det \left(\nabla_\theta^2 \mathcal{L}(\theta ;X)+\sum_{j\in {\cal A}(\theta)}{I}_{d,j}s_j^\prime (\theta_{G_j}){I}_{d,j}^\T\right)_{{\cal S}(\theta)},
   \end{equation}
which will be used in expressing the conditional density $p_{\theta_0}(\cdot\mid \hat{\theta},\hat{g})$ in the following lemma.

\begin{lemma}[Conditional density for penalized estimation]\label{lemma penalized}\label{lemma4}
Suppose Assumption \ref{assump:4} holds. 
Fix any $\theta_0 \in \Theta$ and let $(X, W, \hat{\theta}, \hat{g})$ be drawn from the following joint model
    $$
    \left\{\begin{array}{l}
    (X,W) \sim P_{\theta_0}\times N(0,I_d/d), \\
    \hat{\theta}=\hat{\theta}(X, W), \hat{g}=\hat{g}(X, W)=\nabla_\theta {\cal L}(\hat{\theta};X,W). 
    \end{array}\right.
    $$
    Suppose the event that $\hat{\theta}(X, W)$ is an SSOSP of \eqref{grouppenalty} has positive probability. 
Then, when this event happens,
    the density of the conditional distribution of $X \mid \hat{\theta}, \hat{g}$ w.r.t. $\nu_{\mathcal{X}}$ is 
     \begin{equation}\label{density3}
      \begin{aligned}
         p_{\theta_0}(\cdot\mid \hat{\theta},\hat{g})\propto f\left(x ; \theta_0\right) \cdot\exp\left(-\frac{\left\|\hat{g}-\nabla_\theta \mathcal{L}(\hat{\theta};x)\right\|^2}{2\sigma^2/d}\right)
         \cdot F_{{\rm pen}}(\hat{\theta};x) \cdot {1}_{x \in {\mathcal{X}}_{\hat{\theta}, \hat{g}}},
      \end{aligned}
     \end{equation}
 where ${F}_{\rm pen}(\theta;X)$ is defined in \eqref{F_pen}
   and $${\cal X}_{{\theta},{g}}=\left\{x\in{\cal X}: \textrm{ for some } w\in\bbR^d,\theta=\hat{\theta}(x,w)\textrm { is an SSOSP of \eqref{grouppenalty}, and } g=\hat{g}(x,w)\right\}.$$
   \end{lemma}

Replacing $\theta_0$ in the expression \eqref{density3} with $\hat{\theta}$, we obtain the density for sampling copies \(\widetilde{X}^{(m)}\):
\begin{equation}\label{density4}
      \begin{aligned}
         p_{\hat{\theta}}(\cdot\mid \hat{\theta},\hat{g})\propto f(x ; \hat{\theta}) \cdot\exp\left(-\frac{\left\|\hat{g}-\nabla_\theta \mathcal{L}(\hat{\theta};x)\right\|^2}{2\sigma^2/d}\right)\cdot
         F_{\rm pen}(\hat{\theta};x)\cdot {1}_{x \in {\mathcal{X}}_{\hat{\theta}, \hat{g}}}
      \end{aligned}.
     \end{equation}

To confirm that \eqref{density4} represents a valid density, we verify that the integral of the expression on the right-hand side is both positive and finite, with details provided in Appendix~\ref{verify} 
of the Supplementary Material.

For space considerations, we defer the extensions to nonlinear constraints, including smooth constraints and \(\ell_p\)-norm constraints (\(\forall p \geq 0\)), to the Supplementary Material 
(Appendices~\ref{sec:smooth-constrained} and \ref{sec:lp-constrained}).

Compared with the penalized case described here, the extension to nonlinear constraints is more challenging, because 
the influence of the constraints on the optimization is difficult to characterize. 
The key challenge is to identify conditions that ensure local optimality for points situated on the boundary of the feasible set.

\section{Weighted aCSS method and approximate validity}\label{section 3}
This section presents the implementation of the aCSS method using \(p_{\hat{\theta}}(\cdot \mid \hat{\theta}, \hat{g})\) as derived in the last section, and then establishes its approximate validity.  

\subsection{Weighted Sampling for copies}
For ease of notation, let $P(\cdot; \,\hat{\theta}, \hat{g})$ denote the target distribution \(p_{\hat{\theta}}(\cdot \mid \hat{\theta},\hat{g})\) where we generate  copies \(\widetilde{X}^{(m)}\). 
For any given $P(\cdot; \,\theta, g)$,  a direct sampling may be challenging in general. 
Instead, \cite{barber2022testing} proposed to draw the copies from
\(\widetilde{X}^{(m)}\sim \widetilde{P}_M(\cdot;X,\hat{\theta},\hat{g})\),
where the sampler $\widetilde{P}_M(\cdot;X,\theta, g)$ is required to satisfy: 
    \begin{equation}\label{sampling}
        \centering
         \begin{aligned}
             \text{If } X \sim P(\cdot; \,\theta,g) \text{ and } (\widetilde{X}^{(1)},\ldots,\widetilde{X}^{(M)})\mid X \sim \widetilde{P}_M(\cdot;X,\theta,g), \text{ then}\\
             \text{the random vector } (X, \widetilde{X}^{(1)},\ldots,\widetilde{X}^{(M)}) \text{ is exchangeable.}
         \end{aligned} 
         \end{equation}
Apart from i.i.d. samplers for $P(\cdot \mid \theta, g)$, some MCMC approaches also satisfy \eqref{sampling} (see \citet[Section 2.2.3]{barber2022testing}). 
To generalize this sampling strategy, we follow \cite{harrison2012conservative} and introduce the following weighted version of aCSS sampling. 

 Suppose there exists another distribution \(Q(\cdot; \theta, g)\) such that \(P(\cdot ; \theta, g) \ll Q(\cdot; \theta, g)\), 
 and let \(\mu(\cdot; \theta, g)\) be the Radon-Nikodym derivative \(\rmd P(\cdot; \theta, g)/\rmd Q(\cdot; \theta, g)\).
Further suppose there is a sampler \(\widetilde{Q}_M(\cdot; X, \theta, g)\) such that
\begin{gather} \label{weighted sampling}
    \begin{aligned}
    \text{If } X \sim {Q}(\cdot\,;\theta,g) \text{ and } (\widetilde{X}^{(1)},\cdots,\widetilde{X}^{(M)})|X \sim \widetilde{Q}_M(\cdot\,;X,\theta,g), \text{ then}\\
    \text{the random vector } (X, \widetilde{X}^{(1)},\cdots,\widetilde{X}^{(M)}) \text{ is exchangeable.}
    \end{aligned}
\end{gather}

Instead of sampling copies \(\widetilde{X}^{(m)}\) from \(p_{\hat{\theta}}(\cdot \mid \hat{\theta},\hat{g})\), we sample $\widetilde{X}^{(m)}\sim  \widetilde{Q}_M(\cdot; X, \hat{\theta}, \hat{g})$ and define
the weighted p-value as 
\begin{equation}\label{weighted p-values}
        {\rm pval}_w(X,\tilde{X}^{(1)},\dots,\tilde{X}^{(M)}=\frac{\mu(X;\hat{\theta}, \hat{g})+\sum_{i=1}^M \mu(\widetilde{X}^{(i)};\hat{\theta}, \hat{g}){1}\left\{ T(\widetilde{X}^{(i)}) \geq T(X) \right\}}{\mu(X;\hat{\theta}, \hat{g})+\sum_{i=1}^M \mu(\widetilde{X}^{(i)};\hat{\theta}, \hat{g})}.
    \end{equation} 
Since this weighted p-value is invariant to scaling of \(\mu\), we can replace \(\mu(\cdot; \hat{\theta}, \hat{g})\) with its scaled version in \eqref{weighted p-values} when exact computation is infeasible.
The weighted aCSS method is summarized as follows. 
\begin{algorithm}
\caption{Implementation of Weighted aCSS Method}
\label{alg: aCSS implementation}
     \indent (i) Given \(X\), draw a randomization \(W \sim {N}(0,{I}_d/d)\). Compute the (randomized) statistic \((\hat{\theta}, \hat{g})\) and derive the sampling distribution \(p_{\hat{\theta}}(\cdot \mid \hat{\theta}, \hat{g})\).\\
     \indent (ii) If \(\hat{\theta}\) is not as an SSOSP, set \(\widetilde{X}^{(1)} = \cdots = \widetilde{X}^{(M)} = X\). Otherwise, sample copies from \((\widetilde{X}^{(1)}, \ldots, \widetilde{X}^{(M)}) | (X,\hat{\theta},\hat{g}) \sim \widetilde{Q}_M(\cdot\,; X, \hat{\theta}, \hat{g})\), where \(\widetilde{Q}_M\) is selected to satisfy \eqref{weighted sampling} with the conditional density \(p_{\hat{\theta}}(\cdot \mid \hat{\theta}, \hat{g})\) as computed.\\
     \indent (iii) Compute the weighted p-value defined in \eqref{weighted p-values} with some test statistic \(T\).
    \end{algorithm} 

    The weighted aCSS method is more general and allows flexible choices of \(Q(\cdot; \theta, g)\). In particular, setting \(Q(\cdot; \theta, g) = P(\cdot; \theta, g)\) yields \(\mu \equiv 1\), thereby recovering the unweighted approach of \cite{barber2022testing}. Consequently, their sampling procedures extend naturally to the weighted aCSS. 
Moreover, when each population \(P_\theta\) of the model permits direct i.i.d. sampling, choosing \(Q(\cdot; \theta, g) = P_\theta\) amounts to i.i.d. sampling from \(P_{\hat{\theta}}\). 
In this case, the additional term in the conditional density \(p_{\hat{\theta}}(\cdot \mid \hat{\theta}, \hat{g})\) apart from $f(x;\hat{\theta})$ acts as a weighting function. 
From this perspective, the aCSS method is analogous to the parametric bootstrap \citep{efron1992bootstrap}, but with weights to achieve a better approximation.

The performance of the weighted method is sensitive to the choice of \(Q(\cdot; \theta, g)\) and \(\widetilde{Q}_M(\cdot; X, \theta, g)\). 
At one extreme, if the sampler \(\widetilde{Q}_M(\cdot; X, \theta, g)\) is chosen such that the generated copies \(\widetilde{X}^{(m)}\) are identical to \(X\), the weighted p-value returned by Algorithm \ref{alg: aCSS implementation} will always be 1. 
At the other extreme, if \(\widetilde{X}^{(m)}\) are too dissimilar to \(X\), which often leads to very small \(\mu(\widetilde{X}^{(i)}; \theta, g)\), then the resulting weighted p-value may also be close to 1. 
In either case, the weighted aCSS method is powerless. 
In Appendix~\ref{app: BF}, we illustrate how to avoid such extremes in a numerical experiment, but we leave a thorough study for future work. 

\subsection{Approximate validity}\label{section: approximate valid}

To establish the approximate validity of the weighted p-value in \eqref{weighted p-values}, we impose the following assumptions on the estimator $\hat{\theta} : \mathcal{X} \times \mathbb{R}^d \to \Theta$ and the family $\{P_\theta: \theta\in \Theta\}$. 

\begin{assumption}\label{assumption:2}
For $(X, W) \sim P_{\theta_0} \times {\cal N}(0,{I}_d/d)$ with $\theta_0 \in \Theta$, with probability at least $1-\delta(\theta_0)$, 
$\hat{\theta}(X, W)$ is an SSOSP and satisfies that $\|\hat{\theta}(X, W) - \theta_0\|\leq {e(\theta_0)}$.  
\end{assumption}

Assumption~\ref{assumption:2} is compatible with well-studied convergence rates of MLEs and penalized MLEs. 
By controlling the magnitude of perturbation, the estimator \( \hat{\theta} \) often retains the same rate as the unperturbed one; see \cite{barber2022testing} for a discussion on this assumption in low-dimensional setting and \cite{zhu2023approximate} for the high-dimensional case.

\begin{assumption}\label{assumption:3}Define $H(\theta; x) = -\nabla^2_\theta \log f(\theta; x)$.
    For any $\theta_0 \in \Theta$, the expectation $H(\theta) := \mathbb{E}_{\theta_0} [H(\theta; x)]$ is finite for all $\theta \in B(\theta_0, e(\theta_0)) \cap \Theta$, where \( e(\theta_0) \) is the same constant as in Assumption \ref{assumption:2}. 
    Furthermore, the following two inequalities hold: 
\[
\begin{aligned}
    \mathbb{E}_{\theta_0}\left\{\sup_{\theta \in B(\theta_0, e(\theta_0)) \cap \Theta} e(\theta_0)^2 \left(\lambda_{\max} (H(\theta_0) - H(\theta; X))\right)\right\} &\leq \varepsilon(\theta_0),\\
    \log \mathbb{E}_{\theta_0}\left[\exp\left\{\sup_{\theta \in B(\theta_0, e(\theta_0)) \cap \Theta} e(\theta_0)^2 \cdot \left(\lambda_{\max}(H(\theta; X) - H(\theta_0))\right)\right\}\right] &\leq \varepsilon(\theta_0).
\end{aligned}
\]
\end{assumption}
Assumption~\ref{assumption:3} is a regularity condition regarding the Hessian of the log likelihood. 
As shown in \cite{barber2022testing}, for canonical generalized linear models, we can set \( \varepsilon(\theta_0) = 0 \).

The following theorem guarantees the approximate validity of the weighted aCSS method. 
\begin{theorem}\label{Theorem 1}
    Suppose Assumptions \ref{assumption:2} and \ref{assumption:3} hold. If $X\sim P_{\theta_0}$ for some $\theta_0\in \Theta$, then for any test statistic \( T
    \) and rejection threshold \( \alpha \in [0,1] \), the weighted p-value returned by Algorithm \ref{alg: aCSS implementation} satisfies
    \[
        pr\left( \mathrm{pval} \leq \alpha \right) \leq \alpha + 3\sigma e(\theta_0) + \varepsilon(\theta_0) + \delta(\theta_0),
    \]
    where \( e(\theta_0) \), \( \varepsilon(\theta_0) \), and \( \delta(\theta_0) \) are defined in Assumptions \ref{assumption:2} and \ref{assumption:3}.
\end{theorem}

As discussed in \cite{barber2022testing}, in classical settings with fixed $\theta_0$, the terms \(e(\theta_0)\), \(\delta(\theta_0)\), and \(\varepsilon(\theta_0)\) all vanish as  the sample size $n$ grows. 
Moreover, if \(e(\theta_0) = \mathcal{O}(n^{-\kappa})\) for some $\kappa>0$, we can obtain asymptotic Type-I error control, i.e,
\(
pr\left(\mathrm{pval} \leq \alpha\right) = \alpha + o(1), 
\) by setting \(\sigma \asymp n^{\tilde{\kappa}}\) for some $\tilde{\kappa} < \kappa$. 
 
Theorem~\ref{Theorem 1} may appear similar to Theorem 1 in \citet{barber2022testing}, but it introduces several important improvements. 
First, it adopts the weighted aCSS framework, which accommodates a wider variety of sampling schemes. 
Second, the guarantee for Algorithm~\ref{alg: aCSS implementation} permits the use of MLEs with general forms of regularization, which can yield a much smaller $e(\theta_0)$. 

In high-dimensional settings, Theorem 2 in \citet{zhu2023approximate} improves upon Theorem 1 of \citet{barber2022testing} by leveraging sparsity assumptions, and we provide an analogous result in the Supplementary Material 
(Appendix~\ref{sec:improved-validity}).
Furthermore, we demonstrate in Appendix~\ref{sec:misspecification} of the Supplementary Material that aCSS methods are robust to model misspecification.

\section{aCSS with other estimators beyond Regularized MLE}\label{section 4}

Thus far, the aCSS method has exclusively employed the regularized MLE. 
However, the regularized MLE may not always be the ideal estimator, especially when it is sensitive to extreme observations or is hard to optimize. 
This section illustrates how the aCSS framework can be extended to incorporate alternative estimators, with necessary modifications to ensure compatibility.

\subsection{Maximum Trimmed Likelihood Estimator}\label{trimmed MLE section}
Consider the linear model
$
X_i = Z_i^\top \theta + \varepsilon_i$ for $i = 1, \ldots, n$,
where \(X_i \in \mathbb{R}\) is the response variable, \(Z_i \in \mathbb{R}^d\) is a vector of fixed covariates, \(\varepsilon_i \in \mathbb{R}\) is the i.i.d. noise, and \(\theta \in \mathbb{R}^d\) is the parameter of interest with its true value denoted by \(\theta_0\). 
We assume that the noise distribution has a twice continuously differentiable density \(f_0\) with respect to the Lebesgue measure \(\nu_0\) on \(\mathbb{R}\). Consequently, the density of \(X = (X_1, \ldots, X_n)\) with respect to the Lebesgue measure
\(\nu_{\cal X} = \nu_0^{\otimes n}\) on \({\cal X} = \mathbb{R}^n\) is $
f(X; \theta) = \prod_{i=1}^n f_0(X_i-Z_i^\T\theta)$ .

To increase the robustness of the estimator used in the aCSS method, we replace the MLE with the Maximum Trimmed Likelihood Estimator (MTLE) \citep{hadi1997maximum}. 
For simplicity, we focus on the unregularized case, but it is possible to incorporate the regularization that has been considered before. 

Suppose \(h < n\) is a pre-specified integer. 
Define the negative log trimmed likelihood as
\[
\mathcal{L}_{\rm trim}(\theta; X) = \sum_{i=1}^h -\log f_0(X_{(i)} - Z_{(i)}^\T \theta),
\]
where the observations \(X_i\)'s are ordered such that:
\[
f_0(X_{(1)} - Z_{(1)}^\T \theta) \geq \cdots \geq f_0(X_{(n)} - Z_{(n)}^\T \theta).
\]
This ordering depends on \(\theta\), but we omit the notation for clarity. 
With \(\sigma > 0\) and \(W \in \mathbb{R}^d\) as previously defined, the perturbed MTLE is defined as 
\begin{equation}\label{MTLE}
    \hat{\theta}(X, W) = \argmin_{\theta \in \Theta} \mathcal{L}_{\rm trim}(\theta; X) + \sigma W^\T \theta,
\end{equation}

The MTLE operates by selecting the top \(h\) ``good'' subsamples \((X_{(1)}, \ldots, X_{(h)})\) and then computing the MLE based on them. 
Thus, the MTLE gains robustness at the cost of efficiency compared to the MLE. 
We denote the selected subsamples as 
\begin{equation}\label{MTLE-Selection}
{\cal J}(X; \theta) = \left\{ i \in [n] : f_0(X_i - Z_i^\top \theta) \geq f_0(X_{(h)} - Z_{(h)}^\top \theta) \right\}
\end{equation}
If the selection is strictly ordered, i.e., \(f_0(X_{(h)} - Z_{(h)}^\top \theta) > f_0(X_{(h+1)} - Z_{(h+1)}^\top \theta)\) and \(J = {\cal J}(X; \theta)\), then the gradient and Hessian of \({\cal L}_{\rm trim}\) at \(\theta\) are
\begin{align}
G_{{\rm trim}, J}(\theta; X_J) &= \sum_{i=1}^h \nabla \log f_{0}(X_{J_i}-Z_{J_i}^\T\theta) Z_{J_i}, \\
H_{{\rm trim}, J}(\theta; X_J) &= -\sum_{i=1}^h \nabla^2 \log f_{0}(X_{J_i}-Z_{J_i}^\T\theta) Z_{J_i}Z^\T_{J_i}.
\end{align}

For shorthand, we write \(x = [a, b]_J\) for the vector $x$ such that $x_J = a $ and $ x_{-J} = b$ where \(a \in \mathbb{R}^h\) and \(b \in \mathbb{R}^{n-h}\). We can now define a new version of SSOSP tailored for \eqref{MTLE}. 

\begin{definition}[SSOSP for MTLE]\label{SSOSP MTLE}
    A parameter \(\theta \in \bbR^d\) is a strict second-order stationary point (SSOSP) of the optimization problem \eqref{MTLE} if it satisfies all the following conditions:
    
     \indent (i) \(f_{0}(X_{(h)}-Z_{(h)}^\T\theta) > f_{0}(X_{(h+1)}-Z_{(h+1)}^\T\theta)\). 
     
     \indent (ii) {First-Order Condition (KKT):}  $G_{{\rm trim}, J}(\theta; X_J) + \sigma W =0$ with $J={\cal J}(X;\theta)$.
     
     \indent (iii) {Second-Order Condition:} $
        H_{{\rm trim}, J}(\theta; X_J)\succ 0$ with $J={\cal J}(X;\theta)$.
\end{definition}

The next lemma provides 
the conditional density for MTLE given both $\hat{\theta}$ and ${\cal J}(X;\theta)$. 
Note that $\hat{g}$ is not needed since we do not include any regularization in \eqref{MTLE}. 

\begin{lemma}[Conditional density for the MTLE]\label{lemma conditional density trimmed}  
Fix any $\theta_0\in\bbR^d$ and let $(X,W,\hat{\theta})$ be drawn from the joint model: 
    $$\left\{
        \begin{array}{l}
            (X,W)\sim P_{\theta_0}\times N(0,{I}_d/d)\\
             \hat{\theta} =\hat{\theta}(X,W)
        \end{array}
    \right.$$
    Fix any ${J}\subseteq [n]$. Suppose the event that $\hat{\theta}=\hat{\theta}(X, W)$ is an SSOSP of \eqref{MTLE} and ${\cal J}(X;\hat{\theta})=J$ has positive probability. Then, when this event happens, the conditional distribution of the subvector $X_{J}\mid \hat{\theta},X_{-J}$ has density  
    \begin{equation}\label{true density trimmed}
        \begin{aligned}
            p_{\theta_0}\left(x \mid \hat{\theta},X_{-J}\right)\propto \prod_{i\in J} f_{0}(x_i;\theta_0)\cdot 
            \exp\left(-\frac{\left\|G_{{\rm trim}, J}(x; \hat{\theta})\right\|^2}{2\sigma^2/d}\right)\cdot \det \left(H_{{\rm trim}, J}(x;\hat{\theta})\right)\cdot {1}_{x\in {\cal X}^{\rm trim}_{\hat{\theta},X_{-J},{J}}}
        \end{aligned}
    \end{equation}
    with respect to the measure $\nu_{\cal X}^\prime=\nu_0^{\otimes h}$ defined on ${\bbR}^h$, where
    \begin{align*}
        {\cal X}^{\rm trim}_{\theta,y,J}=\left\{x\in \bbR^h: \text{for some }w\in\bbR^d,\theta=\hat{\theta}\left(\left[x,y\right]_{J},w\right)\text{ is an SSOSP of }\eqref{MTLE} \text{ and } J={\cal J}([x,y]_J;\theta)\right\}.
    \end{align*}
\end{lemma}

Given $J={\cal J}([x,y]_J;\theta)$, 
we will sample sub-copies \(\widetilde{X}_{\rm sub}\) independently from the density that replaces $\theta_0$ in \eqref{true density trimmed} by $\hat{\theta}$: 
    \begin{equation}\label{sample density trimmed}
    p_{\hat{\theta}}\left(x \mid \hat{\theta},X_{-J}\right)\propto \prod_{i\in J} f_{0}(x_i;\hat{\theta})\cdot 
            \exp\left(-\frac{\left\|G_{{\rm trim}, J}(x; \hat{\theta})\right\|^2}{2\sigma^2/d}\right)\cdot \det \left(H_{{\rm trim}, J}(x; \hat{\theta}\right)\cdot {1}_{x\in {\cal X}^{\rm trim}_{\hat{\theta},X_{-J},{J}}},
\end{equation}
and the full copies $\widetilde{X}$ will be constructed as $\left[\widetilde{X}_{\rm sub},X_{-{J}}\right]_{J}$. 

The formal aCSS method with MTLE is summarized as follows:
\medskip
\begin{algorithm}
    \caption{Formal aCSS Method with MTLE}
    
    \indent (i) Given \(X\), draw \(W \sim N(0,{I}_d/d)\). Compute \(\hat{\theta}=\hat{\theta}(X,W)\) in \eqref{MTLE} and $J={\cal J}(X;\hat{\theta})$ in \eqref{MTLE-Selection}.
        
    \indent (ii) If $\hat{\theta}$ is not an SSOSP, then set $\widetilde{X}^{(1)}=\ldots=\widetilde{X}^{(M)}=X$ and return p-value as 1.
        
    \indent (iii) Otherwise, sample sub-copies $\widetilde{X}^{(1)}_{\rm sub},\ldots, \widetilde{X}^{(M)}_{\rm sub}$ from the sampler with respect to the density $p_{\hat{\theta}}(\cdot\mid \hat{\theta},X_{-J})$ given in \eqref{sample density trimmed}, and construct the copies $\widetilde{X}^{(i)}_{\rm full}=\left[\widetilde{X}^{(i)}_{\rm sub},X_{-J}\right]_{J}$.
        
    \indent (iv) Compute the weighted p-value defined in \eqref{weighted p-values} with some test statistic $T$.
    \end{algorithm} 
   
    \medskip

To establish the theoretical guarantee of Type-I error control, Assumption \ref{assumption:3} and Theorem~\ref{Theorem 1} are adapted to MTLE as follows. 
\begin{assumption}\label{assumption:MTLE}
    For any $\theta_0 \in \Theta$ and $J\subseteq [n]$ with $|J|=h$, the expectation $H_{{\rm trim},J}(\theta)=\bbE_{\theta_0} (H_{{\rm trim},J}(\theta; X_J) )$ exists for all $\theta \in B(\theta_0, e(\theta_0)) \cap \Theta$, and furthermore the following inequalities hold: 
\[
\begin{aligned}
    \mathbb{E}_{\theta_0}\left\{\sup_{\theta \in B(\theta_0, e(\theta_0)) \cap \Theta} e(\theta_0)^2 \left(\lambda_{\max} (H_{{\rm trim},J}(\theta_0) - H_{{\rm trim},J}(\theta; X))\right)\right\} &\leq \varepsilon(\theta_0),\\
    \log \mathbb{E}_{\theta_0}\left[\exp\left\{\sup_{\theta \in B(\theta_0, e(\theta_0)) \cap \Theta} e(\theta_0)^2 \cdot \left(\lambda_{\max}(H_{{\rm trim},J}(\theta; X) - H_{{\rm trim},J}(\theta_0))\right)\right\}\right] &\leq \varepsilon(\theta_0).
\end{aligned}
\]
Here $e(\theta_0)$ is the same constant as that appears in Assumption \ref{assumption:2}.
\end{assumption}

\begin{theorem}\label{theorem: MTLE}
Under Assumptions \ref{assumption:2}, the result of Theorem~\ref{Theorem 1} continues to hold for the aCSS method with MTLE, provided that Assumption \ref{assumption:3} is replaced with Assumption \ref{assumption:MTLE}.
\end{theorem}

Since the sampled copy \(\widetilde{X}\) differs from the original \(X\) only in the subvector \(\widetilde{X}_J\), the aCSS method using the MTLE conditions on more information than using the MLE. 
This additional conditioning reduces randomization and often lowers statistical efficiency, which is a trade-off for improved robustness. 
We use the trimmed estimator to mitigate the influence of extreme or contaminated samples, which can otherwise degrade the MLE.
By discarding outliers and focusing on the reliable samples, the trimmed estimator achieves greater robustness. 
We have also extended the aCSS method to accommodate quantile estimators, as detailed in Appendix~\ref{sec:quantile} of the Supplementary Material.

\subsection{Nonparametric Additive Models with Gaussian Noises}\label{gaussian model section}
Let the covariate matrix \(Z = (Z_1, \ldots, Z_n)^\T \in \mathbb{R}^{n \times d}\) be fixed, where \(z_{ij}\) denotes the \(j\)th element of \(Z_i\). 
The response
$X=(X_1,\ldots,X_n)\in \bbR^n$ has independent normal entries modelled as
\begin{equation}\label{gaussian model}
    X_i\sim N(\mu_i,\nu^2),\quad  \mu_i = \sum_{j=1}^d h_j(z_{ij}), \quad i\in [n], 
\end{equation}
where $h_j$'s are unknown nonparametric functions for each coordinate, and \(\nu^2\) is the known variance. 
This model is parametrized by the mean vector \(\mu = (\mu_1, \ldots, \mu_n)^\T \in \mathbb{R}^n\), whose true value is denoted by \(\mu_0 = (\mu_1^0, \ldots, \mu_n^0)^\T \in \mathbb{R}^n\). 
The additive structure in \eqref{gaussian model} provides desirable interpretability by isolating the effect of each predictor, but estimating the component functions $h_j$'s remains challenging due to their infinite-dimensional nature.

In the literature on nonparametric additive models, sieve-type estimators are commonly used. 
Specifically, assume each \(h_j\) belongs to a function class \(\mathcal{F}_j = \mathrm{span}\{\psi_{j1},\psi_{j2}, \ldots\}\), where $\{\psi_{jl}\}$ are basis functions such as polynomials or splines. 
With a sample of size $n$, this function class is truncated to a finite-dimensional subset $
\mathcal{F}_j^n = \mathrm{span}\{\psi_{j1}, \ldots, \psi_{jp_j}\}$, where \(p_j\) is a truncation parameter chosen based on $n$ and $\mathcal{F}_j$. 

Let $D=\sum_{j=1}^d p_j$. 
Using the selected basis functions, we construct an \(n \times D\) design matrix as 
\[
B = \left[ B^1 \mid \ldots \mid B^d \right], \text{ where } B^j_{il} = \psi_{jl}(z_{ij}),  j\in [d],  l\in [p_j], i\in [n],
\]

and estimate $\mu$ by \(B\theta\) for some \(\theta \in \mathbb{R}^D\), thereby approximating the nonparametric additive model with a normal linear model. 
This framework encompasses much of the existing literature on nonparametric additive models, e.g., \cite{10.1214/09-AOS692, ravikumar2009sparse}.

Our aCSS method for this model will make use of the following perturbed estimator:
\begin{equation}\label{finite gaussian}
    \hat{\theta}(X, W) = \argmin_{\theta \in \mathbb{R}^D} \mathcal{L}(\theta; X, W) + \sum_{j=1}^d \rho_j(\theta_{G_j}),
\end{equation}
where $G_j$ denote the indices of $\theta$ associated with $B^j$ so that \(\theta_{G_j} \in \mathbb{R}^{p_j}\) is the coefficients for \(h_j\), \(\rho_j(\cdot)\) are functions to impose structure-induced penalties on \(h_j\), and 
\begin{equation}\label{loss gaussian}
        \mathcal{L}(\theta; X, W) = \frac{1}{2\nu^2} \|X - B\theta\|^2 + {\cal R}(\theta) + \sigma W^\T (B\theta).
\end{equation}
Again, \(\mathcal{R}(\theta)\) is an optional twice-differentiable penalty function
and $W\sim N(0, {I}_n/n)$.
Note that the noise vector \(W\) is of dimension $n$ rather than $D$, so it may seem different from the previous formulation in \eqref{random MLE}. 
However, since $X$ is modelled by \(\mu \approx B\theta\), treating \(B \theta\) (rather than $\theta$ alone) as the parameter aligns this formulation with the previous one in \eqref{random MLE}. 
Consequently, Definition \ref{definition for penalty} of SSOSP applies to \eqref{finite gaussian},  provided that the penalty functions \(\rho_j(\cdot)\) satisfy Assumption \ref{assump:4}. 

Our aCSS method for the additive models does not condition on $\nabla_\theta \mathcal{L}(\hat{\theta} ; X, W)$; instead, we define \begin{equation}\label{gradient gaussian}
    \hat{g} = \hat{g}(X,W) = \frac{1}{\nu^2}(B\hat{\theta}(X,W)-X)+\sigma W.
\end{equation} 
Then we can express $\nabla_\theta \mathcal{L}(\hat{\theta}; X, W)$ as $B^\T \hat{g}+\nabla_\theta {\cal R}(\hat{\theta})$, which is a function of $(\hat{\theta}, \hat{g})$.

Assume each \( \rho_j \) is convex (which implies Assumption \ref{assump:4}) and that the loss function \( \mathcal{L}(\theta; X, W) \) has a positive definite Hessian (for example, by taking \( \mathcal{R}(\theta) = \tau \|\theta\|^2 \) for some small \( \tau > 0 \)). This ensures that the optimization problem \eqref{finite gaussian} has a unique minimizer, which is also an SSOSP of \eqref{finite gaussian}. 
Therefore, the conditional distribution of \(X \mid \hat{\theta}, \hat{g}\) can be simplified as
\begin{equation}\label{eq:additive-Gaussian-true}
     p_{\mu_0}(\cdot \mid \hat{\theta}, \hat{g})\propto{N}\left(B\hat{\theta}+\left(1+\frac{n}{\sigma^2\nu^2}\right)^{-1}\left\{\left(\mu_0-B\hat{\theta}\right)-\frac{n}{\sigma^2}\hat{g}\right\},\nu^2\left(1+\frac{n}{\sigma^2\nu^2}\right)^{-1}{I}_n\right).
\end{equation}
Replacing $\mu_0$ with $\hat{\mu}=B\hat{\theta}$, we directly draw i.i.d. copies $\widetilde{X}^{(m)}$ from 
\begin{equation}\label{sampling gaussian}
 {N}\left(B\hat{\theta}-\left(1+\frac{n}{\sigma^2\nu^2}\right)^{-1}\frac{n}{\sigma^2}\hat{g},\nu^2\left(1+\frac{n}{\sigma^2\nu^2}\right)^{-1}{I}_n\right). 
\end{equation}
The following theorem provides theoretical guarantees on the resulting aCSS testing. 

\begin{theorem}\label{theorem: linear additive}
    Consider the model \eqref{gaussian model}, and assume each $\rho_j$ is convex and \(\mathcal{L}\) has positive definite Hessian. Suppose that the estimator $\hat{\theta}(X,W):\bbR^n\times \bbR^n\to\bbR^D$ satisfies 
    $\|B\hat{\theta}(X,W)-\mu_0\|\leq e(\mu_0)$
    with probability at least $1-\delta(\mu_0)$, where the probability is taken with respect to the distribution $(X,W)\sim {N}\left(\mu_0,\nu^2\right)\times {N}(0,{I}_n/n)$. Then the i.i.d. copies $\widetilde{X}^{(1)},\ldots, \widetilde{X}^{(M)}$ generated from \eqref{sampling gaussian} are approximately exchangeable with $X$, which means that
    $${\rm d}_{\rm exch}\left(X,\widetilde{X}^{(1)},\ldots,\widetilde{X}^{(M)}\right)\leq 
    \frac{\sigma}{2\sqrt{n}} e(\mu_0)+\delta(\mu_0).$$
    In particular, for any predefined test statistic $T : \mathcal{X} \to \mathbb{R}$ and rejection threshold $\alpha \in [0,1]$, the p-value defined in \eqref{pval} satisfies
    \[
    \mathbb{P}\left( {\rm pval}_T(X,\widetilde{X}^{(1)},\ldots,\widetilde{X}^{(M)})\leq \al\right)\leq \alpha +\frac{\sigma}{2\sqrt{n}} e(\mu_0)+\delta(\mu_0).
    \]
\end{theorem}

Note that the optimization problem \eqref{finite gaussian} can be rewritten equivalently as
\[
\hat{\theta}(X, W) = \argmin_{\theta \in \mathbb{R}^D} \left( -\frac{1}{2\nu^2} \left\| (X - \sigma \nu^2 \cdot W) - B\theta \right\|^2 + R(\theta) + \sum_{j=1}^p \sum_{j=1}^d \rho_j(\theta_{G_j})\right),
\]
and that \(X - \sigma \nu^2 \cdot W \sim N(\mu_0,(\nu^2+\sigma^2/n)I_n)\). Therefore, if the perturbation magnitude is set as $\sigma=O(\sqrt{n})$, the perturbed estimator can achieve the same convergence rate as in the case without perturbation. 
For example, in $\al$-smooth Sobolev spaces, we can establish $e(\mu_0)$ with a rate of $O(\sqrt{s(\log d/n )^{2\alpha/(2\alpha+1)}})$ using the results in \cite{10.1214/09-AOS692}.

In Appendix~\ref{sec:gam} of the Supplementary Material, we extend this aCSS method to non-Gaussian generalized additive models \citep{hastie2017generalized}, but the theoretical upper bound on the error inflation converges at a slower rate in $n$ than that in Theorem \ref{theorem: linear additive}. 

While the aCSS framework was originally developed for parametric models, our successful extension to nonparametric additive models indicates its potential for use in other nonparametric or semiparametric settings. 

\section{Gaussian aCSS and Generalized aCSS}\label{section 5}

Previous sections extended the aCSS method to accommodate more regularized estimators, but some key limitations remain: the penalty parameter must be chosen independently of the data, and the choices of estimators are still restrictive. 
To overcome these limitations, this section introduces the generalized aCSS method.

To start with, we revisit the Gaussian model (\ref{gaussian model}) in Section~\ref{gaussian model section}. Instead of conditioning on an estimator as in the ordinary aCSS framework, we condition on a perturbed observation defined as $
X_{\rm noise} = X + \sigma U$,
where \(U \sim N(0, I_n)\) is independent of \(X\). 
We then have \(X_{\rm noise} \sim N(\mu_0, (\nu^2 + \sigma^2) I_n)\) and 
\begin{equation}\label{eq: conditional gaussian aCSS}
X \mid X_{\rm noise} \sim N\left(\frac{\sigma^2}{\sigma^2 + \nu^2} \mu_0 + \frac{\nu^2}{\sigma^2 + \nu^2} X_{\rm noise}, \frac{\sigma^2 \nu^2}{\sigma^2 + \nu^2} I_n\right).
\end{equation}

Let \(\hat{\mu} = \hat{\mu}(X_{\rm noise})\) be any estimator of $\mu_0$ based on \(X_{\rm noise}\). 
Replacing \(\mu_0\) in the above conditional distribution with \(\hat{\mu}\), we propose to draw i.i.d. copies $\widetilde{X}^{(m)}$ from:
\begin{equation}\label{sampling gaussian aCSS}
    N\left(\frac{\sigma^2}{\sigma^2 + \nu^2} \hat{\mu} + \frac{\nu^2}{\sigma^2 + \nu^2} X_{\rm noise}, \frac{\sigma^2 \nu^2}{\sigma^2 + \nu^2} I_n\right).
\end{equation}
Since this method is different from the ordinary aCSS method and is specifically designed for Gaussian distributions, we name it \textit{Gaussian aCSS}. 
Its approximate validity is guaranteed by the following theorem. 

\begin{theorem}\label{theorem: gacss}
    Consider the Gaussian model (\ref{gaussian model}). Assume that $\text{pr}\left\{\|\hat{\mu}(X_{\rm noise})-\mu_0\|\leq e(\mu_0)\right\}\geq 1-\delta(\mu_0)$,
    where the probability is taken with respect to $X_{\rm noise}\sim N(\mu_0,(\nu^2+\sigma^2){I}_n)$. 
    Then the i.i.d. copies $\widetilde{X}^{(1)},\cdots, \widetilde{X}^{(M)}$ generated from 
    \eqref{sampling gaussian aCSS}
    are approximately exchangeable with $X$, satisfying 
    $${\rm d}_{\rm exch}\left(X,\widetilde{X}^{(1)},\cdots,\widetilde{X}^{(M)}\right)\leq 
    \frac{\sigma}{2\nu^2} e(\mu_0)+\delta(\mu_0).$$
    In particular, for any $T : \mathcal{X} \to \mathbb{R}$ and $\alpha \in [0,1]$, the p-value defined in (\ref{pval}) satisfies
    \[
    \mathbb{P}\left( {\rm pval}_T(X,\widetilde{X}^{(1)},\cdots,\widetilde{X}^{(M)})\leq \al\right)\leq \alpha +\frac{\sigma}{2\nu^2} e(\mu_0)+\delta(\mu_0).
    \]
\end{theorem}

The result of Theorem \ref{theorem: gacss} is equivalent to that of Theorem \ref{theorem: linear additive}, after accounting for scaling differences. However, the interpretations are fundamentally different, as detailed below.

In the ordinary aCSS framework, the sampling density involves an indicator function of the event that the estimator is an SSOSP. 
To eliminate the indicator of this event, Theorem \ref{theorem: linear additive} requires the loss function to be strictly convex so that the event holds with probability 1. 
For an estimator under non-convex regularization, an ordinary aCSS method must verify it is an SSOSP for each generated copy, which incurs significantly computational costs. 
In contrast, Gaussian aCSS does not require the definition of SSOSP. 
Furthermore, the sampling distribution of copies is simply Gaussian.  
Consequently, if the test statistic is a linear functional of the data, 
then Gaussian aCSS becomes resampling-free, as the distribution of \( T(\widetilde{X}) \) can be determined directly.

Moreover, the ordinary aCSS method is limited to non-overlapping group-type regularized MLEs (as we have developed in Section~\ref{penalty section}) and requires the penalty parameter to be set independently of \(X\).  
For instance, these limitations exclude the common practice of using cross-validation to tune the penalty parameter, as well as the use of overlapping regularization.
In contrast, the Gaussian aCSS method is compatible with any estimator that depends on \(X_{\rm noise}\), including the MLE with sparse group lasso \citep{simon2013sparse}, which involves overlapping penalties, and the Dantzig selector \citep{10.1214/009053606000001523}, which does not maximize the likelihood. 

Note that if \( \nu \) is bounded, then the estimator \( \mu(X_{\text{noise}}) \) can achieve the same convergence rate as the estimator based on the original data \( \mu(X) \). Consequently, we can generally attain good estimation accuracy and asymptotic validity. 

The choice of \( \sigma \) represents a trade-off between test validity and power, because if \( \sigma = 0 \), the generated copies \( \widetilde{X} \) are identical to the original data \( X \), rendering the Gaussian aCSS powerless.

Extending the Gaussian aCSS method to other parametric models $\{P_\theta\}$ leads to the \textit{generalized aCSS method} outlined below: 

\begin{algorithm}
    \caption{General Recipe for generalized aCSS Method}\label{alg: gaCSS}
        \indent (i) Compute the noisy statistics \(Z = H(X, W)\), where \(W\) is a random perturbation, and obtain an estimator \(\hat{\theta}\) based on \(Z\).\\
        \indent (ii) Derive the conditional density \(p_{\theta_0}(\cdot \mid Z)\) of \(X \mid Z\), where \(\theta_0\) is the true parameter. \\
        \indent (iii) Replace \(\theta_0\) with \(\hat{\theta}\) and generate copies \(\widetilde{X}\) from \(p_{\hat{\theta}}(\cdot \mid Z)\).
\end{algorithm}

Algorithm~\ref{alg: gaCSS} is general and it subsumes several existing methods as special cases. 
In particular, setting \(Z=\left(\hat{\theta}(X, W),\hat{g}(X, W)\right)\) recovers the ordinary aCSS method, while setting $Z=X_{\text{noise}}$ yields the Gaussian aCSS method. 
 Moreover, if \(Z\) is chosen as an unperturbed sufficient statistic, Algorithm~\ref{alg: gaCSS} reduces to the classical CSS method \citep{stephens2012goodness}. 
 
In Algorithm~\ref{alg: gaCSS}, \(W\) need not be a random vector; it can be any random element whose specific form depends on the application. For example, \(W\) could represent a random data split, dividing \(X\) into \(X_{(1)}\) and \(X_{(2)}\), where \(X_{(1)}\) is used for estimation and \(X_{(2)}\) for inference. 
\section{Application to CRT with power analysis}\label{section 6}

We apply the aCSS method to test conditional independence and study the power optimality in a high-dimensional setting, where the parameter dimension \(d=d_n\) can grow with the sample size \(n\).

\subsection{Preliminaries: model-X CRT and conditional CRT}
\cite{candes2018panning} introduced the conditional randomization test (CRT) to test conditional independence. 
For random variables \(({\cal X}, {\cal Y}, {\cal Z})\), the goal is to test the null hypothesis 
$H_0: {\cal Y}\indp {\cal X}\mid {\cal Z}$ against the alternative $H_1: {\cal Y}\nindp {\cal X}\mid {\cal Z}$.
Let \({X} \in \mathbb{R}^n\), \({Y} \in \mathbb{R}^n\), and \({Z} \in \mathbb{R}^{n \times d}\) denote the observed $n$ realizations of ${\cal X}, {\cal Y}$, and $ {\cal Z}$, respectively. 
If the conditional distribution \(\mathcal{L}(X \mid Z)\) is known, the null hypothesis implies that  
\(
\mathcal{L}({X} \mid  {Y}, {Z})=\mathcal{L}({X} \mid {Z})
\).
This property is the foundation of CRT: we sample copies \(\tilde{{X}}^{(1)}, \cdots, \tilde{{X}}^{(M)}\) from \(\mathcal{L}({X} \mid {Z})\) and compute the p-value as
\begin{equation}\label{pval CRT}
    {\rm pval} = \frac{1}{M+1}\left(1 + \sum_{m=1}^M \mathbbm{1}\left\{T(\tilde{{X}}^{(m)}, {Y}, {Z}) \geq T({X}, {Y}, {Z})\right\}\right),
\end{equation}
where the test statistic \(T({X}, {Y}, {Z})\) is chosen such that larger values provide evidence against $H_0$. 

The requirement of a known distribution \(\mathcal{L}(X \mid Z)\) limits the original CRT. 
To relax this, we can use the CSS method to generate copies for CRT, which we refer to as the conditional CRT.
If \(\mathcal{L}(X \mid Z)\) lies in a model \(\{P_{\theta}(\cdot \mid Z) : \theta \in \Theta\}\), then testing $H_0: X \indp Y \mid Z$ amounts to testing 
\begin{equation}\label{null hypothesis aCSS CRT}
H_0: X \mid (Z, Y) \sim P_{\theta}(\cdot \mid Z) \text{  for some } \theta \in \Theta. 
\end{equation} 
The original CRT assumes that the true parameter value \(\theta_0\) is known and generates copies from \(P_{\theta_0}(\cdot \mid Z)\). In contrast, the conditional CRT generates copies from \(P_\theta(\cdot \mid Z, U)\), where \(U = U(X, Z)\) is a sufficient statistic so \(P_\theta(\cdot \mid Z, U)\) does not depend on the parameter \(\theta\).

As discussed in Section~\ref{section 1}, the CSS-based conditional CRT is ineffective when \(P_\theta(\cdot \mid Z, U)\) collapses to a point mass at $X$, which is often the case if \(d \gg n\). 
To leverage structural assumptions (e.g. sparsity) commonly imposed in high-dimensional settings, one can test \eqref{null hypothesis aCSS CRT} using the aCSS method with regularized MLEs. 
The Gaussian aCSS method in Section \ref{section 5} can be applied similarly.  
We collectively refer to the CRT using copies generated by these methods as the aCSS CRT. 
Since the aCSS CRT relies on the aCSS framework to generate copies, it inherits the robustness of aCSS methods with respect to the model specification of \( X \mid Z \).

\subsection{Power analysis of aCSS CRT}
Unlike the original CRT and conditional CRT, the aCSS CRT does not inherently control Type-I error. 
In fact, it requires careful tuning of the perturbation magnitude \(\sigma\) to balance between controlling Type-I error and achieving high power. 
Therefore, our analysis should evaluate validity and power simultaneously. 
Previous validity guarantees based on bounding the exchangeability distance apply to all test statistics but are often too loose for a specific statistic. 
Instead, we focus on particular test statistics to derive sharper characterizations of the validity-power trade-off. Consequently, the resulting Type-I error guarantees are specific to the chosen statistics and may not hold universally. 

For tractable analysis, we follow \cite{wang2022high} in considering the linear model 
\begin{equation}\label{Setting}
     {Y}={X}\beta+{Z}\xi+\varepsilon,\quad \varepsilon\sim {N}_n(0,\nu_Y^2I),
\end{equation}
    where $\nu_Y^2$ is a constant, ${X} = (X_1,\ldots, X_n)^\T\in\bbR^n$, and ${Z}\in\bbR^{n\times d}$ with rows $Z_i$  such that the following holds: 
    $$Z_i\stackrel{\rm i.i.d.}{\sim} {N}_d(0,\Sigma_Z),\quad X_i\mid Z_i\sim {N}(Z_i^\T\theta_0,1),\quad ({X},{Z})\perp \varepsilon.$$
Moreover, we assume that the eigenvalues of \(\Sigma_Z\) are bounded away from zero and infinity, and \(\|\xi\|\), \(\|\theta_0\|\), and \(\beta\) are bounded. 
For \eqref{Setting}, testing \(H_0: Y \perp X \mid Z\) is equivalent to testing $H_0:\beta = 0$.
The assumption of known variance for the regression of \( X \) on \( Z \) is made for analytical simplicity in studying aCSS methods.  

\textbf{\textit{Analysis with unlabeled data.}}
\cite{wang2022high} considered using potentially available unlabeled data \((X_{n+i}, Z_{n+i})_{i=1}^m\), which are samples of \((\mathcal{X},\mathcal{Z})\) without corresponding responses \(Y\), to enhance the power of the conditional CRT. 
Define the total number of sample as \(n_* = n + m\). Construct the augmented vector \(X_*\in \mathbb{R}^{n_*}\) and augmented matrix \(Z_*\in\mathbb{R}^{n_*\times d}\) by stacking the labeled data on top of the unlabeled ones. 
In conditional CRT, copies \(\widetilde{X}_{*}^{(\rm css)}\) are generated based on the sufficient statistic \(Z_*^\T X_*\), and the first \(n\) rows of the copies are used as \(\widetilde{X}^{(\rm css)}\) to compute the p-value. 
Since this construction is consistent with the case where there is no unlabeled data, our analysis accommodates $m=0$. 

We assume \(n_* > d\) and $Z_*$ has full column rank, so that the conditional CRT is applicable. 
We consider the aCSS CRT under the same setting: 
with $\sigma>0$ and ${W}\sim {N}(0,{I}_d/d)$ independent of the other random variables, we can derive the perturbed MLE as  
$$\hat{\theta}_{\rm OLS}({X}_*,{W})=\argmin_\theta \frac{1}{2}\left\|{X}_*-{Z}_*\theta\right\|^2+\sigma {W}^T\theta=\left({Z}_*^T{Z}_*\right)^{-1}\left({Z}_*^T{X}_*+\sigma {W}\right).$$
Since the loss function is strictly convex, $\hat{\theta}_{\rm OLS}({X}_*,{W})$ is unique, which simplifies the sampling distribution of copies as 
\begin{equation}\label{aCSS unre CRT}
  \tilde{{X}}_*^{(\rm acss)}\mid Z_*,\hat{\theta}_{\rm OLS}\sim {N}\left({Z}_*\hat{\theta}_{\rm OLS},({I}_{n_*}+\frac{d}{\sigma^2}{Z}_*{Z}_*^T)^{-1}\right). 
\end{equation}
As in conditional CRT, we use the first \(n\) rows of $\tilde{{X}}_*^{(\rm acss)}$ as \(\widetilde{X}^{(\rm acss)}\) to perform aCSS CRT. 

The following theorem shows that 
for an appropriate choice of \(\sigma\) and test statistic $T(X; Y, Z)$, the above aCSS CRT is asymptotically equivalent to the conditional CRT. 

\begin{theorem}\label{connection}
   Consider the linear model \eqref{Setting} with  the following asymptotics:
   $$n_*=n+m,\quad  \lim_{n\to\infty}\frac{n}{n_*}=\kappa^*,\quad \Sigma_Z=I_d,\quad \lim_{n\to\infty}\frac{d}{n}=\kappa,\quad \kappa\kappa^*\in (0,1).$$
   
   If we set \(\sigma = o\left(\sqrt{{d n_*}/{n}}\right)\) and use the test statistic \(T(X; Y, Z) = n^{-1} Y^\T X\), then aCSS CRT is asymptotically equivalent to conditional CRT.
\end{theorem}
The proof of Theorem~\ref{connection} builds on the coupling technique introduced by \cite{fan2023ark}, and the technical tools developed in the Supplementary Material may be of independent interest for power analysis. 
We discuss several implications of Theorem~\ref{connection} in the following.

(1) \textit{Asymptotic efficiency}.  \cite{wang2022high} prove that the conditional CRT has non-trivial asymptotic power under local alternatives $H_1 : \beta = h/\sqrt{n}$ for a fixed $h > 0$. Combined with their result, Theorem~\ref{connection} (with $m=0$) implies the parametric efficiency of aCSS CRT. 

(2) \textit{Benefit of unlabeled data}.
In the regime where $d\asymp n$, the estimation error of OLS cannot vanish without additional assumptions such as sparsity \citep{raskutti2011minimax}, so \(\sigma = o(\sqrt{d})\) is needed in order to ensure asymptotic validity using Theorem 3 of \citet{zhu2023approximate}. 
Theorem \ref{connection} relaxes this requirement to \(\sigma=o(\sqrt{d} \sqrt{n_*/n})\) by leveraging the unlabeled data.

(3) \textit{Connection between CSS and aCSS}.
Theorem \ref{connection} reveals an interesting phenomenon: in our setting, the aCSS method with an appropriate \(\sigma\) is asymptotically equivalent to the CSS method. 
In particular, letting \(\sigma \to 0\) causes the distribution of \(\tilde{{X}}^*_{\rm acss}\), as in \eqref{aCSS unre CRT}, to become a degenerate Gaussian distribution identical to that of \(\tilde{{X}}^*_{\rm css}\). 
This result deepens our understanding of the connection between aCSS and CSS methods, and it may extend to broader settings.

\textbf{\textit{Analysis without unlabeled data.}}
The conditional CRT fails to apply without unlabeled data when \(n \ll d\). 
In this regime, we apply the Gaussian aCSS method introduced in Section~\ref{section 5}: compute \(X_{\rm noise} = X + \sigma W\), where \(W \sim N(0, I_n)\), and obtain an estimator \(\hat{\theta}\) based on \(X_{\rm noise}\). 
Based on \eqref{sampling gaussian aCSS} (with \( \nu = 1 \)), the copies \( \widetilde{X}_{\mathrm{acss}} \) are generated from
\[
N\left(\frac{\sigma^2}{1 + \sigma^2} Z\hat{\theta} + \frac{1}{1 + \sigma^2} X_{\rm noise}, \frac{\sigma^2}{1 + \sigma^2} I_n\right).
\]

In this high dimensional setting, we require the sparsity of $\theta_0$ to ensure consistent estimation. 
\begin{assumption}[Sparsity of $\theta_0$]\label{sparsity}
   \(\theta_0\) is sparse with \(s\) nonzero entries and for some nonnegative sequence \(\{\gamma_{n}\}\),  we have 
   ${s\log d}/{n^{1/2 + \gamma_{n}}} \to 0$ and $ n^{\gamma_{n} - 1/2} \to 0$.
\end{assumption}
As long as \(n\gg s\log d\), the existence of the sequence \(\{\gamma_{n}\}\) in Assumption~\ref{sparsity} is guaranteed; for example, by taking any $\gamma_{n}\in (1/2-\log\{n/(s\log(d))\}/\log n, 1/2$). 
The requirement that $n\gg s \log d$ is mild since it corresponds to the optimal sample size for consistent estimation \citep{raskutti2011minimax}. 
Furthermore, we assume the $\ell_1$-error of the perturbed estimator $\hat{\theta}$ can be bounded as follows. 

\begin{assumption}[Estimator error]\label{estimator}
  $\|\hat{\theta} - \theta_0\|_1 = O_p\left(s\sqrt{{\log d}/{n}}\right)$,   where $s = \|\theta_0\|_0$. 
\end{assumption}
Since \(X_{\text{noise}} \sim N(Z\theta_0, (\sigma^2 + 1){I}_n)\), if $\sigma$ is bounded and Assumption~\ref{sparsity} holds, various estimators can satisfy Assumption~\ref{estimator}, including the Dantzig selector, MLE with lasso or some nonconvex penalties \citep{10.1214/08-AOS620, loh2015regularized}. Moreover, if \(\log(d/s) \asymp \log d\), the rate of the $\ell_1$-error in Assumption~\ref{estimator} is known to be optimal \citep{ye2010rate}.

Assumptions~\ref{sparsity} and \ref{estimator} regard the estimation of the model for $X\mid Z$. To study the power of aCSS CRT, we need to specify the test statistic. The naive choice of $n^{-1} X^\T Y$ fails to remove the effect of $Z$. Instead, we consider \(T(X; Y, Z) = n^{-1} X^\T \widetilde{Y}\), where $\widetilde{Y}$ satisfies the following assumption. 

\begin{assumption}[Test Statistic]\label{Test statistics}
     \(\widetilde{Y}\) is a function of \((Y, Z)\) and there are two positive constants $C_1$ and $C_2$ such that with probability converging to 1, it holds that 
    \begin{equation}\label{power condition on statistic}
            \frac{\|Z^\T \widetilde{Y}\|_\infty}{\|\widetilde{Y}\|} \leq C_1 \sqrt{\log d}, \quad \frac{(Y - Z\xi - Z\beta \theta_0)^\T \widetilde{Y}}{\|\widetilde{Y}\|} \geq C_2 \sqrt{n}.
        \end{equation}
\end{assumption}

The two inequalities in \eqref{power condition on statistic} aim to ensure the validity and power of the aCSS CRT, respectively. A discussion on the existence and construction of $\widetilde{Y}$ is provided in Supplementary Material 
(Appendix~\ref{sec:test-stat}).  
With the preparatory work, we can state the following theorem. 
\begin{theorem}\label{power acss}
     Suppose Assumptions \ref{sparsity}, \ref{estimator}, and \ref{Test statistics} hold. 
     Consider the aCSS CRT with \(\sigma = C_3 n^{-\gamma_n}\) for any constant \(C_3 > 0\). 
     For any significance level $\alpha\in (0,1)$, this aCSS CRT is asymptotically valid. Moreover, under local alternatives \(\beta = h / n^{1/2 - \gamma_n}\), this aCSS CRT has asymptotic power no less than by \(\Phi(C_2 C_3 h/\nu_Y^2 - z_{1-\alpha})\).
\end{theorem}

By Corollary 5 in \cite{bradic2022testability}, we conclude that our method achieves the optimal detection rate. Compared to the method in \cite{bradic2022testability}, our approach does not require data splitting. 
In the moderately sparse regime, where \(\sqrt{n}/\log d < s \ll n/(\log d)\), the choice of $\sigma$ in Theorem~\ref{power acss} requires knowledge of \(s\) since the value of \(\gamma_{n}\) in Assumption~\ref{sparsity} depends on \(s\). 
By contrast, in the ultra-sparse regime, where \(s \log d / \sqrt{n} \to 0\), we can set \(\gamma_{n} = 0\) in Assumption \ref{sparsity} and perform the aCSS CRT with a constant $\sigma$ that does not rely on \(s\).  
This transition phenomenon in the required knowledge of sparsity aligns with the adaptivity of optimal inference discussed in \cite{bradic2022testability}.

\section{Simulations}\label{section 7}
In this section, we summarize the numerical performance of our extended aCSS methods in two simulated examples. 
Details are deferred to Appendix~\ref{app: simulate} of the Supplementary Material to save space. 

\subsection{Behrens--Fisher problem with contamination}
\begin{figure}[ht]
\centering
\includegraphics[width=0.8\textwidth]{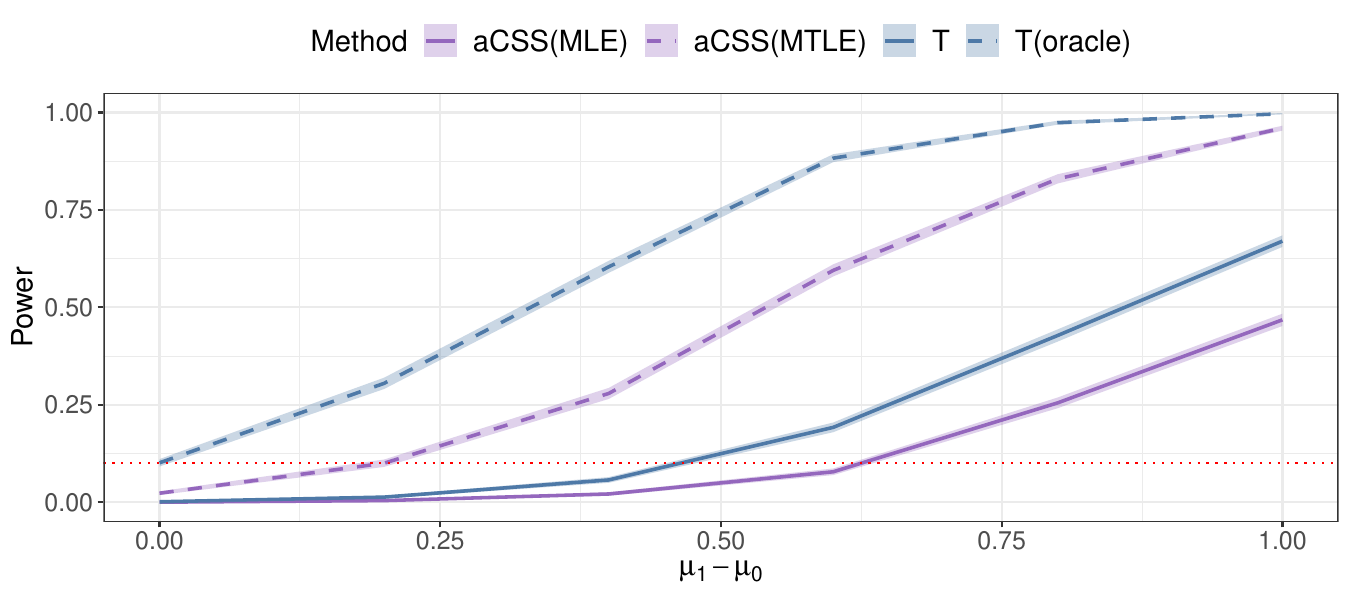}
\caption{Comparison of aCSS Test Power with Alternative Methods in the Behrens-Fisher Problem with Contamination (Nominal Level 0.1, Averaged Over 1000 Replications)}
\label{fig1}
\end{figure}

We begin with the classical Behrens--Fisher problem. Two independent samples are draw from normal distributions with potentially different variances: 
\[
X_1^{(0)}, \ldots, X_{n^{(0)}}^{(0)} \stackrel{\text{i.i.d.}}{\sim}N(\mu^{(0)}, \gamma^{(0)}), \quad 
X_1^{(1)}, \ldots, X_{n^{(1)}}^{(1)} \stackrel{\text{i.i.d.}}{\sim} {N}(\mu^{(1)}, \gamma^{(1)}).  
\]
The goal is to test \(H_0: \mu^{(0)} = \mu^{(1)}\) versus \(H_1: \mu^{(0)} < \mu^{(1)}\). Under the null hypothesis, the model can be parameterized as
\[
\theta = (\mu, \gamma^{(0)}, \gamma^{(1)}) \in \Theta = \mathbb{R} \times \mathbb{R}_+ \times \mathbb{R}_+ \subseteq \mathbb{R}^3.
\]
This problem has been studied by \cite{barber2022testing} in the context of aCSS. 
Here, we consider a more challenging scenario where the first dataset \(X^{(0)}\) is partially contaminated: $m=5$ out of $n^{(0)}=50$ observations are perturbed to have a higher mean. 
This contamination will impact the power of standard inference. 
In this example, we compare the following methods:

\indent (i) aCSS with MLE and MTLE:  The test statistic is the difference in sample means, and the copies are generated using the weighted sampling method. For MTLE, we set \(h = 45\), i.e., the exact number of uncontaminated observations. 
Additional results for other choices of \(m\) and \(h\) are provided in the Supplementary Material.

\indent (ii) t-test and its oracle version: t-test
is known as the most powerful test for the two sample test and the oracle version requires the information of the contaminated part of the dataset.

As shown in Figure \ref{fig1}, the oracle version of the t-test is the most powerful among the four tests, while the standard t-test suffers due to the contamination. 
Comparing the two aCSS methods, the one using MTLE significantly outperforms the one using MLE.

\subsection{Conditional independence test with aCSS CRT}
\begin{figure}[ht]
\centering
\includegraphics[width=0.8\textwidth]{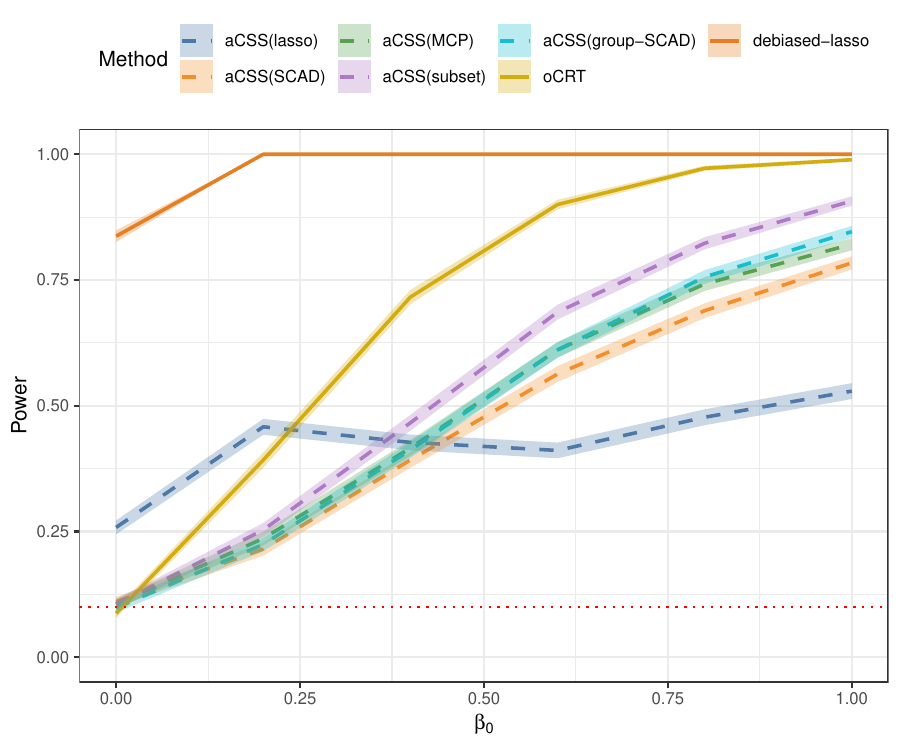}
\caption{Power of aCSS-CRT Compared to Alternative Methods for Conditional Independence Testing (Nominal Level 0.1, Averaged Over 1000 Replications)}
\label{fig2}
\end{figure}
Consider the setting in \eqref{Setting}, where we aim to test \(H_0: \beta = 0\) versus \(H_1: \beta > 0\). This setting was studied by \cite{zhu2023approximate}, where their aCSS method was limited to a lasso estimator of $\theta$. 
Here we use the following methods to perform the inference:

\indent (i) aCSS method with five estimators of $\theta$: lasso, SCAD, MCP, Group-SCAD, and best subset selection estimators. 
Following \cite{liu2022fast}, the CRT with an estimator $\hat{\theta}$ is implemented using the distilled statistic $T(X;Y,Z)= (Y-Z\hat{\xi})^T(X-Z\hat{\theta})$,
where $\hat{\xi}$ is an estimator of $\xi$. 
    Since the sampling distribution of copies \(\widetilde{X}\) is Gaussian, \(\widetilde{T} = T(\widetilde{X})\) is also normally distributed, which makes the CRT resampling-free.
For comparison, we also perform the oracle CRT using the true value of \( \theta_0 \), denoted by the legend ``oCRT'' in Figure \ref{fig2}.

\indent (ii) debiased lasso: This is a well-known method for high dimensional inference; see \cite{javanmard2014confidence, zhang2014confidence, 10.1214/14-AOS1221} for details. 

Figure \ref{fig1} illustrates the validity and performance of each method.
Specifically, the debiased-lasso and the aCSS method using the lasso estimator fail to control the Type-I error. 
Among the valid methods, aCSS with best subset selection estimator achieves the highest power, closely approaching the oracle CRT. 
The aCSS methods using group-SCAD and MCP estimators perform similarly, while the one using SCAD is slightly worse. 
This example demonstrates the necessity and benefit of extending the aCSS method to incorporate nonlinear and group penalty for high-dimensional problems.

\section{Discussion}
This paper extends the aCSS methodology by incorporating nonlinear regularization and exploring its applicability to a broader class of estimators beyond regularized maximum likelihood estimation. 
Furthermore, we propose a generalized aCSS framework that unifies several existing approaches in the literature. In addition, we conduct a detailed power analysis of aCSS procedures and demonstrate their optimality across a range of settings.

Our work opens several directions for future research.
A promising avenue is to extend aCSS methods beyond the Model-X CRT. For instance, selective inference with randomization \citep{10.1214/17-AOS1564} shares similarities with Gaussian aCSS in that both perform inference after adding Gaussian noises. Furthermore, aCSS techniques may be adapted to generate approximately exchangeable copies for use in Model-X knockoff procedures. Recent developments, such as those in \cite{fan2023ark}, may offer useful tools for achieving false discovery rate (FDR) control.

\section*{Acknowledgement}
D. Huang was partially supported by NUS Start-up Grant A-0004824-00-0. 
\section*{Supplementary material}
\label{SM}
The Supplementary Material includes auxiliary results, proofs, and simulation details.

\bibliographystyle{plainnat}
\bibliography{ref.bib}

\newpage
\appendix
\noindent\textbf{\LARGE Appendix}
\smallskip

Appendix~\ref{app: additional results} provides additional results beyond the main text, including the constrained estimator and the quantile estimator. 
Appendix~\ref{app: power analysis} discusses several technical tools that are essential for the power analysis in Section~\ref{section 6}.
Appendix~\ref{app: proof} contains the proofs of the main theorems. Appendix~\ref{app: conditional density proof} presents the proofs of the lemmas related to the derivation of the conditional density \( p_{\theta_0}(\cdot\mid\hat{\theta},\hat{g}) \). Appendix~\ref{app: add proof} provides additional proofs of some technical lemmas. Appendix~\ref{app: simulate} offers more details and results for the simulation studies. Additional notations used in the appendix are introduced as follows:
For any vector $\nu\in \mathbb{R}^d$ and $\epsilon > 0$, define $\mathcal{B}(\nu; \epsilon) = \{x \in \mathbb{R}^d : \|x - \nu\| < \epsilon\}$ and $\mathring{\mathcal{B}}(\nu; \epsilon) = \{x \in \mathbb{R}^d : 0 < \|x - \nu\| < \epsilon\}$. Let \( a, b \in \mathbb{R}^d \) be two vectors. We denote their Hadamard (element-wise) product by
\[
a \odot b = (a_1 b_1, a_2 b_2, \ldots, a_d b_d)^\intercal.
\]

\section{Some additional results}\label{app: additional results}

\subsection{aCSS with General Smooth Constrained MLE}\label{sec:smooth-constrained}
In this section, we consider aCSS method with the MLE under general constraints of the following form: 
\begin{equation}\label{MLE}
    \hat{\theta}=\hat{\theta}(X,W)=\argmin_{\theta\in \Theta}\left\{{\cal L}(\theta;X,W)\mid G_i(\theta)\leq 0,\forall i=1,\cdots,r\right\},
\end{equation}
where the loss function is defined as
\begin{equation*}\label{loss function}
    {\cal L}(\theta;X,W)=-\log f(X;\theta)+R(\theta)+\sigma W^\T\theta,
\end{equation*} and $G_i, i = 1, \dots, r$ are the constraint functions. 

Before proceeding, let us introduce some regularity conditions and definitions:
\begin{assumption}[Constrained Space]\label{assump:constraint}
    Define the constrained parameter space \(\Theta_0\) as
    \[
    \Theta_0 = \left\{ \theta \in \Theta \subseteq \mathbb{R}^d \mid G_i(\theta) \leq 0, \; i = 1, \ldots, r \right\},
    \]
    where each \(G_i\) is twice continuously differentiable, and \(\Theta_0\) contains no isolated points.
\end{assumption}
\begin{definition}[Active constraints]\label{active}
    Consider a point $\theta\in\Theta_0$. The constraint $G_i(\theta)\leq 0$ is called \textit{active} at $\theta$ if, for every $\varepsilon>0$, there exists another point $\theta^\prime\in \Theta_0\cap \mathring{\cal B}(\theta,\varepsilon)$ such that $G_i(\theta^\prime)=0$. We denote the set of active constraint indices at $\theta$ by $${\cal I}(\theta)=\{i\in [r] \mid G_i(\theta)\text{ is active at }\theta\}.$$
\end{definition}

\begin{definition}[Regular point]\label{regular}
    A point \(\theta \in \Theta_0\) is called regular if, for some \(\varepsilon > 0\), the matrix composed of the gradients \(\nabla_\theta G_i(\theta^\prime), i \in \mathcal{I}(\theta)\) maintains a constant rank for \(\theta^\prime\in \mathcal{B}(\theta, \varepsilon)\). The constant rank is denoted by \(r(\theta)\).
    \end{definition}

    This regularity condition comes from the rank constant theorem \citep{rudin1976principles}. Now we can propose the definition of SSOSP for the general constrained problem.

\begin{definition}[SSOSP for the general constrained problem]\label{SSOSP}
    A parameter $\theta$ is a strict second-order stationary point (SSOSP) of the optimization problem (\ref{MLE}) if it satisfies all of the following:

    \indent (i) Regularity: $\theta\in \Theta_0$ is a regular point.
    
    \indent (ii) First-order necessary conditions, i.e, Karush-Kuhn-Tucker (KKT) conditions:
    \begin{equation}\label{KKT}
        \nabla_\theta {\cal L}(\theta;X,W)+\sum_{i=1}^r\lambda_i \nabla_\theta G_i(\theta)=0,
    \end{equation}
    where $\lambda_i\geq 0$ for all $i$, and $\lambda_i=0$ for all $i\in [r]\backslash {\cal I}(\theta)$. 

    \indent (iii) Second-order sufficient condition: 
    $$U_{\theta}^\T\left[\nabla_\theta^2\mathcal{L}(\theta ;X)+\sum_{i=1}^r\lambda_i \nabla_\theta^2 G_i(\theta)\right]U_{\theta}\succ 0.$$
    Here, $\lambda=(\lambda_1,\cdots,\lambda_r)^\T$ is the solution that satisfies the KKT conditions (\ref{KKT}), and $U_{\theta}$ denotes a matrix which forms an orthonormal basis for subspace orthogonal to span $\left\{\nabla_\theta G_i(\theta),i\in {\cal I}(\theta)\right\}$, that is,
    \begin{equation}\label{proj}
        U_\theta\in \bbR^{d\times (d-r(\theta))},U_\theta U_\theta^\T={\cal P}^{\perp}_{\text{\rm span}\{\nabla_\theta G_i(\theta)\}_{i\in{\cal I}(\theta)}}. 
    \end{equation}
\end{definition}

The matrix $U_{\theta}^\T\bigl[\nabla_\theta^2\mathcal{L}(\theta ;X)+\sum_{i=1}^r\lambda_i \nabla_\theta^2 G_i(\theta)\bigr]U_{\theta}$ is orthogonally similar for different $U_\theta$ under (\ref{proj}), which means that the definition of SSOSP is invariant for $U_{\theta}$. Despite the Lagrange multipliers $\lambda_i$ not being uniquely determined by the KKT conditions (\ref{KKT}), we will show in Appendix \ref{invariance-to-multipliers} that the matrix $$U_{\theta}^\T\left[\nabla_\theta^2\mathcal{L}(\theta ;X)+\sum_{i=1}^r\lambda_i \nabla_\theta^2 G_i(\theta)\right]U_{\theta}$$
    remains invariant for all sets of $\lambda_{1:r}$ that satisfy the KKT conditions (\ref{KKT}). Consequently, for $\theta$ being  SSOSP of (\ref{MLE}), we can define a function from$\bbR^d\times {\cal X}\times \bbR^d\to \bbR$:
    \begin{equation}\label{matrix}
        {F}_{\rm con}(\theta;X,W)=
                \det\left(U_{\theta}^\T\bigl[\nabla_\theta^2\mathcal{L}(\theta ;X)+\sum_{i=1}^r\lambda_i \nabla_\theta^2 G_i(\theta)\bigr]U_{\theta}\right)
    \end{equation}

\begin{lemma}\label{lemma1}(Conditional density for the constrained problem) Suppose Assumption \ref{assump:constraint} hold. Fix any $\theta_0\in\Theta$ and let $(X,W,\hat{\theta},\hat{g})$ be drawn from the joint model: 
    $$\left\{
        \begin{array}{l}
            (X,W)\sim P_{\theta_0}\times N(0,I_d/d)\\
             \hat{\theta}=\hat{\theta}(X,W),
             \hat{g}=\hat{g}(X,W)=\nabla_\theta {\cal L}(\hat{\theta};X,W)
        \end{array}
    \right.$$
    
    Fix any ${\cal I}\subseteq [r]$. 
    Let $\mathcal{E}$ be the event that $\hat{\theta}(X,W)$ is an SSOSP of (\ref{MLE}) with ${\cal I}(\hat{\theta}(X,W))={\cal I}$. Note that $\mathcal{E}\in \sigma(\hat{\theta})$. 
    Assume that $\mathcal{E}$ happens with positive probability. 
    Then,when $\mathcal{E}$ happens, the conditional distribution of $X\mid \hat{\theta},\hat{g}$ has density  
    \begin{equation}\label{density1}
        \begin{aligned}
            p_{\theta_0}(\cdot \mid \hat{\theta},\hat{g})\propto f(x;\theta_0)\cdot \exp\left\{-\frac{\left\|\hat{g}-\nabla_\theta \mathcal{L}(\hat{\theta};x)\right\|^2}{2\sigma^2/d}\right\}\cdot {F}_{\rm con}\left(\hat{\theta};x,\frac{1}{\sigma}\left(\hat{g}-\nabla_\theta \mathcal{L}(\hat{\theta};x)\right)\right)\cdot {1}_{x\in \mathcal{X}_{\hat{\theta},\hat{g}}}
        \end{aligned}
    \end{equation}
    with respect to the base measure $\nu_{\mathcal{X}}$, where $U_\theta$ is defined in (\ref{proj}), ${F}_{\rm con}$ is defined in (\ref{matrix}) and 
    $$\mathcal{X}_{\theta,g}=\left\{x\in\mathcal{X}:\textrm{ for some } w\in \mathbb{R}^d, \theta=\hat{\theta}(x,w) \textrm{is an SSOSP of (\ref*{MLE}), and } g=\hat{g}(\theta;x,w)\right\}$$
\end{lemma}
The proof of Lemma \ref{lemma1} is given in Section~\ref{proof:Lemma1}. 

By direct calculation and correspondence with some symbols, it is clear that the presented conditional density here is consistent with the formula in \cite[Lemma 1]{zhu2023approximate} when only linear constraints are involved. Replacing $\theta_0$ with $\hat{\theta}$, now the density of the sampling distribution is given by
\begin{equation}\label{density2}
    \begin{aligned}
        p_{\hat{\theta}}(\cdot \mid \hat{\theta},\hat{g})\propto f(x;\hat{\theta})\cdot \exp\left\{-\frac{\left\|\hat{g}-\nabla_\theta \mathcal{L}(\hat{\theta};x)\right\|^2}{2\sigma^2/d}\right\}\cdot{F}_{\rm con}\left(\hat{\theta};x,\frac{1}{\sigma}\left(\hat{g}-\nabla_\theta \mathcal{L}(\hat{\theta};x)\right)\right)\cdot {1}_{x\in \mathcal{X}_{\hat{\theta},\hat{g}}}
    \end{aligned}
\end{equation}

\subsection{aCSS with $\ell_{p}$ constrained MLE}\label{sec:lp-constrained}
In Section~\ref{sec:smooth-constrained}, the constraint functions were required to be twice continuously differentiable. 
Here we consider a special case where the MLE is subject to an \(\ell_p\)-norm constraint. By permitting \(0 \leq p < \infty\), we significantly extend the applicability of aCSS method. 

For any $p\geq 0$, denote $\|\theta\|_p^p := \sum_{j=1}^d |\theta_j|^p$, where we adopt the convention that $0^0=0$ and $x^0=1$ for any $x>0$. 
We consider the following constrained estimator:  
\begin{equation}\label{l_0 problem}
    \hat{\theta}(X,W) = \argmin_{\|\theta\|_p^p \leq R_p} \mathcal{L}(\theta; X, W) = \argmin_{\|\theta\|_p^p \leq R_p} -\log f(X; \theta) + \sigma W^\T \theta,
\end{equation}
where \(R_p > 0\) is a prespecified number. Here, the estimator is constrained within an \(\ell_p\)-ball defined as
\[
\mathcal{B}_p(R_p) = \left\{\theta \in \mathbb{R}^d \mid 
\|\theta\|_p^p  \leq R_p\right\}, \quad 0 \leq p \leq \infty. 
\]
When \( p = 0 \), the \( \ell_0 \)-ball refers to the set of vectors in \( \mathbb{R}^p \) with at most \( R_0 \) nonzero elements, where \( R_0 \) is a given positive integer.
When \(p=1\), the \(\ell_1\)-ball constraint simplifies to a set of linear constraints. 
For $0<p<1$, the $\ell_p$-ball is also used to induce sparsity in high-dimensional estimation \citep{raskutti2011minimax,wainwright2019high}.

For the $\ell_p$-norm constrained estimator in \eqref{l_0 problem}, SSOSP is defined as follows. 
\begin{definition}[SSOSP for the $\ell_p$ constrained problem]\label{SSOSP l_0}
    A parameter $\theta\in \bbR^d$ is a strict second-order stationary point (SSOSP) of the optimization problem \eqref{l_0 problem} if it satisfies all of the following:
 
   \indent (i) Feasibility: $\|\theta\|_p^p\leq R_p$.

   \indent (ii) First-order necessary conditions (i.e, KKT conditions):
    \begin{equation}\label{l_p KKT}
        \nabla_\theta {\cal L}(\theta;X,W)+\lambda s=0,
    \end{equation}
    where $\lambda\geq 0$ and  $\lambda=0$ if $\|\theta\|_p^p< R_p$ and 
    $s\in \bbR^d$ satisfies that $s_j=p\cdot {\rm sign}(\theta_j)\cdot|\theta_j|^{p-1}$ if $j\in \text{\rm supp}(\theta)$ and for $ j\notin \text{\rm supp}(\theta)$, it holds that
    $$\left\{
        \begin{array}{cc}
            s_j\in (-\infty,\infty),&0\leq p<1\\
            s_j\in \left[-1,1\right], &p=1\\
            s_j=0&p>1.
        \end{array}
    \right.$$

   \indent (iii) Second-order sufficient condition: 
    $$Z_\theta^\T\left(\nabla_\theta^2\mathcal{L}(\theta ;X)+\lambda \Lambda(\theta)\right)_{{\cal T}(\theta)}Z_\theta\succ 0,$$
    where $\lambda$ is the solution to (\ref{l_p KKT}), $\Lambda(\theta)={\rm diag}\left\{p(p-1)|\theta|^{p-2}\right\}$
    and ${{{\cal T}(\theta)}}$ is defined via 
    $${{\cal T}(\theta)}=\left\{
        \begin{array}{lc}
            \text{\rm supp}(\theta)&\|\theta\|_p^p= R_p\\
            \,[d]&\|\theta\|_p^p< R_p
        \end{array}
    \right.$$
    For the matrix \( Z_\theta \), we have \( Z_\theta = I_d \) if \( \|\theta\|_p^p < R_p \).
Otherwise, let \( n_s(\theta) = |{\rm supp}(\theta)| \) denote the cardinality of the support of \( \theta \). Then:
\begin{itemize}
    \item For \( p = 0 \) or \(n_s(\theta)=1\), \( Z_\theta = I_{n_s(\theta)} \in \mathbb{R}^{n_s(\theta) \times n_s(\theta)} \).
    \item For \( p > 0 \), \( Z_\theta \in \mathbb{R}^{n_s(\theta) \times (n_s(\theta)-1)} \) satisfies
        \begin{eqnarray}\label{proj matrix}
        Z_\theta^\T Z_\theta={I}_{n_s(\theta)-1},    \qquad \text{ and }Z_\theta^\T\left({\rm sign}(\theta)\odot \theta^{p-1}\right)_{{\rm supp}(\theta)} ={0}.
    \end{eqnarray}
\end{itemize} 

\end{definition}

This scenario closely resembles the previous case involving smooth constraints, with the primary distinction arising from the \(\ell_p\) norm's non-differentiability at the point \(0\). 
We consider the cases with $p>0$ and $p=0$ separately. 

For $p>0$, our strategy involves dividing the \(\ell_p\) surface into finite regions according to the support of \(\theta\). Moreover, points inside the \(\ell_p\) ball can be handled using the ordinary aCSS method proposed by \cite{barber2022testing}, leveraging the fact that the ball excluding its boundary is an open set since $p>0$. 
To see why the second order condition is well defined, we point out that the parameter $\lambda$ is uniquely decided by the KKT conditions in \eqref{l_p KKT}: if all $\theta_j=0$, then $\lambda=0$; otherwise, there is at least a nonzero $s_j$, and the corresponding equation in \eqref{l_p KKT} determines the value of $\lambda$.

The case \(p=0\) is uniquely distinct. For \(p>0\), after partitioning the boundary of the \(\ell_p\) ball into finite regions as previously described, the location of a point is at least locally determined by \(n_s(\theta) - 1\) parameters, owing to the necessity of adhering to the boundary constraint function. 
In contrast, this condition does not apply when \(p=0\). For \(p=0\), the boundary constraint function is automatically met if the support of \(\theta\) is chosen appropriately (i.e., \(n_s(\theta) \leq R_p\)), thereby allowing the degrees of freedom for \(\theta\) to be \(n_s(\theta)\).

The following lemma summarizes the conditional density given $\hat{g}$ and $\hat{\theta}$ for \eqref{l_0 problem}. 
\begin{lemma}[Conditional density for the $\ell_p$ constrained problem]\label{lem:constrain} Fix any $\theta_0 \in \Theta$ and let $(X, W, \hat{\theta}, \hat{g})$ be drawn from the joint model
    $$
    \left\{\begin{array}{l}
    (X,W) \sim P_{\theta_0}\times N(0,{I}_d/d), \\
    \hat{\theta}=\hat{\theta}(X, W),
    \hat{g}=\hat{g}(X, W)=\nabla_\theta {\cal L}(\hat{\theta};X,W)
    \end{array}\right.
    $$
    
    Let ${\cal T} \subseteq[d]$. 
Define an event 
\[
\mathcal{E} = \left\{ \hat{\theta}(X, W) \text{ is an SSOSP of \eqref{l_0 problem} with } {\cal T}(\hat{\theta}(X, W)) = {\cal T}, \text{ and } \|\hat{\theta}(X,W)\|_p^p = R_p \text{ if } p > 0 \text{ holds} \right\}.
\] 
Note that \( \mathcal{E} \) is in the sigma-field generated by \( (\hat{\theta}(X, W), \hat{g}(X,W)) \).
Assume $\mathcal{E}$ happens with positive probability. 
    Then, when $\mathcal{E}$ happens, 
    the conditional distribution of $X \mid \hat{\theta}, \hat{g}$ has density
     \begin{equation}\label{true density l_0}
      \begin{aligned}
         p_{\theta_0}(\cdot\mid \hat{\theta},\hat{g})\propto f\left(x ; \theta_0\right) \cdot\exp\left\{-\frac{\left\|\hat{g}-\nabla_\theta \mathcal{L}(\hat{\theta};x)\right\|^2}{2\sigma^2/d}\right\}
         \cdot \det
            \left(Z_\theta^\T\left[\nabla_\theta^2\mathcal{L}(\theta ;X)+\lambda \Lambda(\theta)\right]_{{\cal T}(\theta)}Z_\theta\right)\cdot {1}_{x \in {\mathcal{X}}_{\hat{\theta}, \hat{g}}},
      \end{aligned}
     \end{equation}
   where $${\cal X}_{{\theta},{g}}=\left\{x\in{\cal X}: \textrm{ for some } w\in\bbR^d,\theta=\hat{\theta}(x,w)\textrm { is an SSOSP of \eqref{l_0 problem}, and } g=\hat{g}(x,w)\right\}.$$
 \end{lemma}
The proof of Lemma \ref{lem:constrain} is given in Section~\ref{pf:lem:constrain}.
 
Based on Lemma \ref{lem:constrain}, we can derive the sampling density $p_{\hat{\theta}}(\cdot \mid \hat{\theta},\hat{g})$ for aCSS as follows:
 \begin{equation}
       \begin{aligned}
          p_{\hat{\theta}}(\cdot\mid \hat{\theta},\hat{g})\propto f\left(x ; \hat{\theta}\right) \cdot\exp\left\{-\frac{\left\|\hat{g}-\nabla_\theta \mathcal{L}(\hat{\theta};x)\right\|^2}{2\sigma^2/d}\right\}
          \cdot \left(Z_\theta^\T\left[\nabla_\theta^2\mathcal{L}(\theta ;X)+\lambda \Lambda(\theta)\right]_{{\cal T}(\theta)}Z_\theta\right)\cdot {1}_{x \in {\mathcal{X}}_{\hat{\theta}, \hat{g}}}
       \end{aligned}.
      \end{equation}

\subsection{Improved version of Theorem \ref{Theorem 1}}\label{sec:improved-validity}
Following \cite[Theorem 2]{zhu2023approximate}, we assume that the error \( (\hat{\theta} - \theta_0) \) is approximately sparse with respect to a certain basis, which enables us to present an improved version of Theorem~\ref{Theorem 1}. For illustrative examples of such settings, we refer the reader to \cite{zhu2023approximate}.

For a given choice of $d^\prime$ vectors $\{v_1,\cdots,v_{d^\prime} \}\subset \mathbb{R}^d$, we define the following norm for any $w \in \mathbb{R}^d$: 
\[
\|w\|_{v,0} = \begin{cases}
\min \{|\mathcal{S}| : \mathcal{S} \subseteq [d^\prime], w \in \text{span}(\{v_i\}_{i\in\mathcal{S}})\}, & w \in \text{span}(\{v_i\}_{i\in[d^\prime]}), \\
+\infty, & \text{otherwise}.
\end{cases}
\]
In other words, \( \|w\|_{v,0} \) represents the minimal number of vectors \( v_i \) required such that \( w \) lies in their span. Specifically, if we take \( d' = d \) and \( v_i = e_i \) for \( i \in [d] \), then \( \|w\|_{v,0} = \|w\|_0 \). Next, for any subset $S\subseteq [d^\prime]$, denote by ${\cal P}_{\nu_S}$ the projection to ${\rm span}\left(\{\nu_i\}_{i\in S}\right)$. 
For each $k=0,\cdots,d$, we define 
$$h_\nu(k)= \log \left[\bbE_{Z\sim {\cal N}(0,{I}_d)}\max_{S\subseteq [d^\prime],|S|\leq k}\exp(\left\|{\cal P}_{\nu_S}Z\right\|^2)\right].$$

\begin{theorem}\label{Theorem 2}
Under the notations and assumptions of Theorem \ref{Theorem 1}, suppose it also holds that
\[
\mathbb{P}\left[\|\hat{\theta}(X,W) - \theta_0\|_{v,0} \leq k(\theta_0)\right] \geq 1 - \tilde{\delta}(\theta_0),
\]
for a fixed set of vectors $\nu_1, \ldots, \nu_p \in \mathbb{R}^d$ and some $k(\theta_0)\leq d$. Then we have
\[
    \bbE_{Q_{\theta_0}^*}\left[d_{\rm TV}(p_{\theta_0}(\cdot\mid\hat{\theta},\hat{g}),p_{\hat{\theta}}(\cdot\mid\hat{\theta},\hat{g}))\right] \leq 3\sigma e(\theta_0)\sqrt{\frac{h_\nu(k(\theta_0))}{d}} + {\varepsilon}(\theta_0) + {\delta}(\theta_0)+\tilde{\delta}(\theta_0).
\]
In particular, this implies that for any predefined test statistic $T : \mathcal{X} \to \mathbb{R}$ and rejection threshold $\alpha \in [0,1]$, the p-value output by the aCSS method satisfies
    \[
    \mathbb{P}\left( {\rm pval}\leq \al\right)\leq \alpha +3\sigma e(\theta_0)\sqrt{\frac{h_\nu(k(\theta_0))}{d}} + {\varepsilon}(\theta_0) + {\delta}(\theta_0)+\tilde{\delta}(\theta_0).
    \]
\end{theorem}

Theorem \ref{Theorem 2} is an extension of Theorem \ref{Theorem 1}: if we consider the special case where $k(\theta_0) = d$ and $\tilde{\delta}(\theta_0) = 0$, we can take $h_{\nu}(d)=d$, and consequently the bound in Theorem \ref{Theorem 2} is simplified to match that of Theorem \ref{Theorem 1}. 

\cite{zhu2023approximate} introduces the notion of $\|w\|_{\nu,0}$ in the context of high-dimensional estimation since such estimation is often sparse. 
Furthermore, \cite[Lemma 2]{zhu2023approximate} shows that if \(k(\theta_0) \ll d\), then we have $h_{\nu}(\theta_0)=o(d)$; in this case, Theorem \ref{Theorem 2} gives a sharper bound compared with Theorem \ref{Theorem 1}.

\subsection{Robustness of aCSS method to Model misspecification}\label{sec:misspecification}
When implementing the aCSS methods for testing \eqref{eq:g-in-lambda-theta},  we assume that, under the null hypothesis, \( X \) is distributed according to \( P_{\theta} \) with density \( f(\cdot; \theta) \) relative to a base measure \( \nu_{\mathcal{X}} \), for some \( \theta \in \Theta \subseteq \mathbb{R}^d \). However, in practice, model misspecification may occur—namely, the true distribution of \( X \) is \( P_0 \), with density \( f_0(\cdot) \) with respect to the same base measure \( \nu_{\mathcal{X}} \), but \( f_0(\cdot) \) does not belong to the alternative hypothesis \( H_1 \) we seek to detect.

For instance, consider the problem of testing whether the mean of \( X \) is zero. A common approach is to model \( X \) as Gaussian and perform hypothesis testing on the mean. Yet, in practice, \( X \) may not be perfectly Gaussian, while still having zero mean. This discrepancy motivates an investigation into the potential inflation of Type-I error in the aCSS method resulting from model misspecification.

To start with, we establish a weighted version of  \eqref{approximate valid} to bound the Type-I error inflation. 
\begin{lemma}\label{lemma_weighted}
    Let $P_{\theta_0}^*$ be the distribution of $(X,W)\sim P_{\theta_0}\times {N}(0,{I}_d/d)$ conditional on the event that $\hat{\theta}(X,W)$ is an SSOSP, and $Q_{\theta_0}^*$ be the distribution of $(\hat{\theta}(X,W),\hat{g}(X,W))$ under $(X,W)\sim P_{\theta_0}^*$. For sample \(X\), copies \((\widetilde{X}^{(1)}, \ldots, \widetilde{X}^{(M)})\) generated by the weighted aCSS algorithm, any test statistic \(T: \mathcal{X} \to \mathbb{R}\), and any rejection threshold \(\alpha \in [0,1]\),
 the p-value defined in \eqref{weighted p-values} satisfies
    \begin{equation}\label{weighted approximate valid}
         {P}\left( {\rm pval}(X,\widetilde{X}^{(1)},\cdots,\widetilde{X}^{(M)}) \leq \alpha \right) \leq \alpha + \bbE_{Q_{\theta_0}^*}\left[d_{\rm TV}(p_{\theta_0}(\cdot\mid\hat{\theta}, \hat{g}),p_{\hat{\theta}}(\cdot\mid\hat{\theta}, \hat{g}))\right],
    \end{equation}
    where the probability is taken over the original sample $X$ and copies $\widetilde{X}^{(1)},\cdots,\widetilde{X}^{(M)}$.
\end{lemma}

The proof of Lemma \ref{lemma_weighted} is given in Section~\ref{pf:lemma_weighted}.

Lemma~\ref{lemma_weighted} provides an upper bound on the Type-I error inflation for the weighted aCSS method, which is central to establishing the approximate validity of the weighted aCSS method. 
For example, it will be used in our proof of Theorem \ref{theorem: linear additive} and Theorem \ref{theorem: gacss} in Section~\ref{app:pf-gacss}. 

Now, we consider the robustness of aCSS method under model misspecification. The following analysis applies to any estimator \( \hat{\theta} \) compatible with the aCSS method.

Under the null hypothesis, we assume that \(X\) is distributed according to \(P_{\theta}\) with density \(f(\cdot; \theta)\) relative to a base measure \(\nu_{\cal X}\), for some \(\theta \in \Theta \subseteq \mathbb{R}^d\). 
In cases of model misspecification, where \( X \) follows \( P_0 \) with the true density \( f_0(\cdot) \) with respect to the same base measure \( \nu_{\cal X} \), we continue to apply the aCSS method within the parametric family \( P_\theta \), despite the unknown misspecification. It is essential to quantify the potential increase in Type-I error, as this assessment is crucial for evaluating the robustness of the aCSS method under model misspecification.

\textbf{Implications of the model.}
Suppose that if \( X \sim P_{\theta_0} \) for some $\theta_0\in \Theta$ (i.e.,  without model misspecification), the conditional density of $X$ given $(\hat{\theta},\hat{g})$ is expressed as 
\[
p_{\theta_0}(x \mid \hat{\theta},\hat{g}) \propto f(x; \theta_0) \cdot H(x; \hat{\theta},\hat{g}),
\]
where \( H(x; {\theta},{g}) \) represents the function that accounts for the influence on the distribution of \( X \) conditioning on \( {\theta} = \hat{\theta}(X,W), {g} = \hat{g}(X,W) \).  
This expression comes from the Bayes' theorem: 
given $X=x$, the only randomness comes from $W$, and 
the factor \( H(x; \hat{\theta},\hat{g}) \) is proportional to the conditional density of the random variable $(\hat{\theta}(X,W),\hat{g}(X,W))$. The normalizing constant is hidden in the proportional notation. 

\textbf{Practice of aCSS}. 
We will perform aCSS method using the sampling distribution with the following density
\begin{equation}\label{model misspecification}
    p_{\hat{\theta}}(\cdot\mid \hat{\theta},\hat{g})\propto f(x;\hat{\theta})\cdot H(x;\hat{\theta},\hat{g}). 
\end{equation}

\textbf{True conditional density under model misspecification.}
If \( X \sim P_0 \), the actual conditional density of $X$ given $(\hat{\theta},\hat{g})$ is 
\begin{equation}\label{true density}
    p_0(\cdot \mid \hat{\theta},\hat{g}) \propto f_0(x) \cdot H(x; \hat{\theta},\hat{g}).
\end{equation}

Lemma \ref{lemma_weighted} can be modified to provide a new bound on the Type-I error inflation under model misspecification as follows. 

\begin{lemma}\label{lem:size-inflation-misspecified}
    Let \(P_0^*\) denote the distribution of \((X,W)\) under \(P_{0} \times {\cal N}(0,\frac{1}{d}{I}_d)\) conditional on the event that \(\hat{\theta}(X,W)\) is an SSOSP. Let \(Q_0^*\) be the distribution of \(\hat{\theta}(X,W),\hat{g}(X,W)\) under \((X,W)\sim P_0^*\). Given any predefined test statistic \(T : \mathcal{X} \to \mathbb{R}\) and a rejection threshold \(\alpha \in [0,1]\), the p-value generated by the aCSS method is bounded by:
\[\text{pr}\left(\text{\rm pval} \leq \alpha\right) \leq \alpha + \mathbb{E}_{Q_0^*}\left[d_{\text{TV}}(p_{0}(\cdot \mid \hat{\theta}, \hat{g}), p_{\hat{\theta}}(\cdot \mid \hat{\theta}, \hat{g}))\right].
\]
\end{lemma}

The proof is analogous to that of Lemma \ref{lemma_weighted} and is omitted here.

In Section \ref{section: approximate valid}, we have previously established a bound for $\bbE_{Q_{\theta_0}^*}\left[d_{\rm TV}(p_{\theta_0}(\cdot\mid\hat{\theta},\hat{g}),p_{\hat{\theta}}(\cdot\mid\hat{\theta},\hat{g}))\right]$ under some mild additional assumptions. Now we introduce the following theorem to bound $\bbE_{Q_{0}^*}\left[d_{\rm TV}(p_{0}(\cdot\mid\hat{\theta},\hat{g}),p_{\hat{\theta}}(\cdot\mid\hat{\theta},\hat{g}))\right]$ using $\bbE_{Q_{\theta_0}^*}\left[d_{\rm TV}(p_{\theta_0}(\cdot\mid\hat{\theta},\hat{g}),p_{\hat{\theta}}(\cdot\mid\hat{\theta},\hat{g}))\right]$ for general $\theta_0$.
\begin{theorem}\label{theorem:robust} Fixed $\theta_0\in \Theta$, we have 
    $$\bbE_{Q_{0}^*}\left[d_{\rm TV}(p_{0}(\cdot\mid\hat{\theta},\hat{g}),p_{{\hat{\theta}}}(\cdot\mid\hat{\theta},\hat{g}))\right]\leq 2d_{\rm TV}(P_0^*,P_{\theta_0}^*)+\bbE_{Q_{\theta_0}^*}\left[d_{\rm TV}(p_{\theta_0}(\cdot\mid\hat{\theta},\hat{g}),p_{\hat{\theta}}(\cdot\mid\hat{\theta},\hat{g}))\right],$$
    where $Q_{\theta_0^*}$ is defined in Lemma \ref{lemma_weighted} and $Q_{0}^*$ is defined in Lemma~\ref{lem:size-inflation-misspecified}.
\end{theorem}

Theorem~\ref{theorem:robust} demonstrates that the difference between $$
\mathbb{E}_{Q_{\theta_0}^*}\left[d_{\mathrm{TV}}(p_{\theta_0}(\cdot \mid \hat{\theta}, \hat{g}), p_{\hat{\theta}}(\cdot \mid \hat{\theta}, \hat{g}))\right]
$$and 
$$\mathbb{E}_{Q_{0}^*}\left[d_{\mathrm{TV}}(p_{0}(\cdot \mid \hat{\theta}, \hat{g}), p_{\hat{\theta}}(\cdot \mid \hat{\theta}, \hat{g}))\right]$$ 
can be bounded by twice the total variance between \(P_0^*\) and \(P_{\theta_0}^*\).

Using Theorem~\ref{theorem:robust} and Lemma~\ref{lem:size-inflation-misspecified} together can bound the error inflation if we can bound $d_{\mathrm{TV}}(p_{0}(\cdot \mid \hat{\theta}, \hat{g}), p_{\hat{\theta}}(\cdot \mid \hat{\theta}, \hat{g}))$. 
This is very different from the well-specified case, where direct applications of Theorem \ref{Theorem 1} and Theorem \ref{Theorem 2} involves using Assumption \ref{assumption:2} on the estimation error. 

For a misspecified model, since the true distribution $P_0$ does not lie in the model, we can no longer use Assumption \ref{assumption:2}. 
Instead, we follow \cite{white1982maximum} and select \(\theta_0\) that minimizes the Kullback-Leibler Information Criterion (KLIC): 
\[
\theta_0 = \min_{\theta} D_{KL}(P_0 \| P_{\theta}). 
\]
Classical statistical theory guarantees that under some mild conditions, the MLE will converge to \({\theta_0}\). For a perturbed MLE \(\hat{\theta}(X,W)\), similar convergence is expected since we can control $\sigma$ such that the perturbation is sufficiently small and can be omitted.

By assuming that the estimator always returns a strict second-order stationary point (SSOSP), the following corollary provides a simplification of Theorem~\ref{theorem:robust}. 

\begin{corollary}
    If $\hat{\theta}(X,W)$ is an SSOSP of the optimization problem with probability $1$ (e.g. ${\cal L}(\theta;X,W)$ is strictly convex), then we have 
    $$P_0^*\stackrel{{\rm d}}{=}P_0\times {\cal N}(0,\frac{1}{d}{I}_d),P_{\theta_0}^*\stackrel{{\rm d}}{=}P_{\theta_0}\times {\cal N}(0,\frac{1}{d}{I}_d)$$
    Combing with Theorem~\ref{theorem:robust}, we have
     $$\bbE_{Q_{0}^*}\left[d_{\rm TV}\left(p_{0}(\cdot\mid\hat{\theta},\hat{g}),p_{{\hat{\theta}}}(\cdot\mid\hat{\theta},\hat{g})\right)\right]\leq 2d_{\rm TV}(P_0,P_{\theta_0})+\bbE_{Q_{\theta_0}^*}\left[d_{\rm TV}\left(p_{\theta_0}(\cdot\mid\hat{\theta},\hat{g}),p_{\hat{\theta}}(\cdot\mid\hat{\theta},\hat{g})\right)\right]$$
\end{corollary}

\subsection{aCSS with quantile estimator}\label{sec:quantile}

This section explores the application of the aCSS method combined with quantile regression. 
We again consider the linear model defined in Section \ref{trimmed MLE section}, but assume the noise terms $\varepsilon_i$ satisfy $\text{pr}(\varepsilon_i\leq 0)=\tau$ for some prespecified $0<\tau<1$. 
We consider the following penalized quantile regression:
\begin{equation}\label{quantile estimator}
    \hat{\theta}=\hat{\theta}_{\rm quantile}(X)=\argmin_{\theta\in\Theta} F_{\rm quantile}(\theta;X)
\end{equation}
 where $$F_{\rm quantile}(\theta;X)=\frac{1}{n}\sum_{i=1}^n\ell_\tau(X_i-Z_i^\T\theta)+p(\theta),$$ 
 with $\ell_\tau (t) = t \left\{\tau- {1} (t < 0)\right\}$ being the quantile loss function, and $p(\cdot)$ an optional penalty function. 
 Note that, unlike in other aCSS methods that rely on perturbed estimators, no perturbation is introduced here. 
 This is because the conditional distribution of $X$ given the quantile estimator can be explicitly analyzed as in the following lemma, whereas such analysis in the case of a general MLE is typically intractable.
 \begin{lemma}\label{quantile lemma}
    For $X\in \bbR^n,\theta\in\bbR^d$ and $\varepsilon\geq 0$, we define
\begin{equation}
        s_{\varepsilon}(X;\theta)=T_\varepsilon(X-Z\theta),\qquad J_{\varepsilon}(X;\theta)=\left\{i\in [n]:\left|X_i-Z_i^\T\theta\right|> \varepsilon\right\},
\end{equation}
where $T_{\varepsilon}(t)={\rm sign}(t)\cdot 1\left(|t|>\varepsilon\right)$. If $\hat{\theta}$ is a local minimizer of $F_{\rm quantile}(\theta;X)$, then for all $\widetilde{X}\in \bbR^n$ such that 
$ s_{\varepsilon}(X;\hat{\theta})=s_{\varepsilon}(\widetilde{X};\hat{\theta})$ and $ X_{-J_{\varepsilon}(X;\hat{\theta})}=\widetilde{X}_{-J_{\varepsilon}(\widetilde{X};\hat{\theta})}$,
 $\hat{\theta}$ is also a local minimizer of $F_{\rm quantile}(\theta;\widetilde{X})$.
 \end{lemma}
The proof of Lemma \ref{quantile lemma} in Section~\ref{pf:quantile lemma}.

The essence of Lemma \ref{quantile lemma} lies in the fact that the local optimality of \( F_{\text{quantile}}(\theta; X) \) at \( \hat{\theta} \) can be determined by the sign of \( X - Z\hat{\theta} \). 

Lemma \ref{quantile lemma} characterizes a set of samples that retains the estimator $\hat\theta$. 
The parameter \(\varepsilon\) in
Lemma \ref{quantile lemma} is introduced to facilitate practical computation, which can be treated as a fixed parameter and denoted by \(\varepsilon_0\) henceforth. 

Next, we provide the new definition of SSOSP. 
\begin{definition}[SSOSP for quantile regression]\label{SSOSP quantile}
Given $\varepsilon_0\geq 0$, 
    a parameter $\theta$ is a strict second-order stationary point (SSOSP) of the optimization problem (\ref{quantile estimator}) if $\theta$ is a local minimizer of $F_{\rm quantile}(\theta;X)$ with $J_{\varepsilon_0}(X;\theta)\neq \emptyset$ and $$-\sum_{i\in J_{\varepsilon_0}(X;\theta)}\nabla_\theta^2\log f_i(X_i;\theta)\succ 0.$$
\end{definition}
Similar to the aCSS method with MLE, we compute 
\begin{equation}\label{quantile gradient}
    \hat{g}=\hat{g}_{\rm quantile}(X,W)=\sum_{i\in J^c_{\varepsilon_0}(X;\hat{\theta})}\nabla_\theta \log f_i(X;\hat{\theta})+\sigma W,
\end{equation}
where $W\sim {N}(0,{I}_d/d)$ and $\sigma>0$ are defined as before.
\begin{lemma}\label{lemma conditional density quantile}(Conditional density for quantile regression) Fix any $\theta_0\in\bbR^d$ and let $(X,W,\hat{\theta},\hat{g})$ be drawn from the joint model: 
    $$\left\{
        \begin{array}{l}
            (X,W)\sim P_{\theta_0}\times N(0,I_d/d)\\
             \hat{\theta}=\hat{\theta}_{\rm quantile},(X),\hat{g}=\hat{g}_{\rm quantile}(X,W).
        \end{array}
    \right.$$
    Let ${\cal J}\subseteq [n]$ and let $h=|\mathcal{J}|$. 
    Suppose the event that $\hat{\theta}$ is an SSOSP of (\ref{quantile estimator}) with $J_{\varepsilon_0}(X;\hat{\theta})={\cal J}$ has positive probability.
    When this event happens, 
    the conditional distribution of $X_{{\cal J}}\mid \hat{\theta},\hat{g},X_{-{\cal J}}$ has density  
    \begin{equation}\label{true density quantile}
        \begin{aligned}
            p_{\theta_0}\left(x \mid \hat{\theta},\hat{g}, X_{-\cal J}\right)\propto &\prod_{i=1}^h f_0(x_i-Z_{{\cal J}_i}^\T\theta_0)\cdot 
            \exp\left\{-\frac{\left\|\sum_{i=1}^h f_0^\prime(x_i-Z_{{\cal J}_i}^\T\hat{\theta})Z_{{\cal J}_i}\right\|^2}{2\sigma^2/d}\right\}\\
            &\qquad \qquad \quad \cdot \det\left(\sum_{i=1}^h f_0''(x_i-Z_{{\cal J}_i}^\T\hat{\theta})Z_{{\cal J}_i}Z_{{\cal J}_i}^\T\right)\cdot {1}_{x\in {\cal X}^{\rm quantile}_{\hat{\theta},\hat{g},X_{-\cal J},{\cal J}}}
        \end{aligned}
    \end{equation}
    with respect to the measure $\nu_{\cal X}^\prime=\nu_0^{\otimes h}$ defined on ${\bbR}^h$, where
    \begin{align*}
        {\cal X}^{\rm quantile}_{\theta,g, y,\cal J}& :=\left\{x\in \bbR^h: \theta=\hat{\theta}_{\rm quantile}\left(\left[x,y\right]_{\cal J}\right)\text{ is an SSOSP of (\ref{quantile estimator})},\right.\\
        &\qquad \qquad\qquad \qquad \qquad \quad \left.\text{for some }w\in\bbR^d,\text{ and }{g}=\hat{g}(\left[x,y\right]_{\cal J},w)\right\}.
    \end{align*}
\end{lemma}

Given ${\cal J}=J_{\varepsilon_0}(X;\hat{\theta})$, 
we will condition on 
$(\hat{\theta}, \hat{g})$ as well as  $X_{-{\cal J}}$
and generate a sub copy $\widetilde{X}^{\rm sub}$ for the units in ${\cal J}$ so that the full copy $\widetilde{X}$ is constructed as $\left[\widetilde{X}_{\rm sub},X_{-{\cal J}}\right]_{\cal J}$.
Here we have used the shorthand notation $[a,b]_{J}$ introduced in Section~\ref{trimmed MLE section}. 
Based on the conditional density (\ref{true density quantile}), the sampling distribution of $\widetilde{X}^{\rm sub}$ is given by 
\begin{equation}\label{sampling density quantile}
    \begin{aligned}
 \forall x\in \mathbb{R}^{|{\cal J}|}, ~~       p_{\hat{\theta}}\left(x \mid \hat{\theta},\hat{g}, X_{-\cal J}\right)\propto &\prod_{i=1}^h f_0(x_i-Z_{{\cal J}_i}^\T\hat{\theta})\cdot 
        \exp\left\{-\frac{\left\|\sum_{i=1}^h f_0^\prime(x_i-Z_{{\cal J}_i}^\T\hat{\theta})Z_{{\cal J}_i}\right\|^2}{2\sigma^2/d}\right\}\\
        &\qquad \qquad \quad \cdot \det\left(\sum_{i=1}^h f_0''(x_i-Z_{{\cal J}_i}^\T\hat{\theta})Z_{{\cal J}_i}Z_{{\cal J}_i}^\T\right)\cdot {1}_{x\in {\cal X}^{\rm quantile}_{\hat{\theta},\hat{g},X_{-\cal J},{\cal J}}}. 
    \end{aligned}
\end{equation}
The formal aCSS method with quantile estimator is summarized in Algorithm~\ref{alg:acss-quantile}. 
\begin{algorithm}
    \caption{Formal aCSS Method with quantile estimator}
    \label{alg:acss-quantile}
    \indent (i) Given \(X\), draw a randomization \(W \sim {N}(0,I_d/d)\). Compute the (randomized) statistic \(\hat{\theta}=\hat{\theta}_{\rm quantile}(X),\hat{J}={J}(X;\hat{\theta})\) and $\hat{g}=g(X,W)$.
    
    \indent (ii) If $\hat{\theta}$ is not an SSOSP, then set $\widetilde{X}^{(1)}=\cdots=\widetilde{X}^{(M)}=X$ and return p-value as 1.
    
    \indent (iii) Otherwise, sample sub copies $\widetilde{X}^{(1)}_{\rm sub},\cdots, \widetilde{X}^{(M)}_{\rm sub}$ from the sampler with respect to the density $ p_{\hat{\theta}}\left(x \mid \hat{\theta},\hat{g}, X_{-\cal J}\right)$, and construct the full copies $\widetilde{X}^{(i)}=\left[\widetilde{X}^{(i)}_{\rm sub},X_{-\hat{ J}}\right]_{\hat{J}}$.
    
    \indent (iv) Compute the p-value in \eqref{pval} with some test statistic $T$.
   \end{algorithm}
This resampling method modifies the original \(X\) only for samples whose indices fall within \(J_{\varepsilon_0}(X; \hat{\theta})\), while treating the remaining samples as fixed. Although subsampling may reduce power, we can set a small threshold \(\varepsilon_0\) so that few samples being treated as fixed.

Now we need to modify Assumption \ref{assumption:3} for the quantile estimator. 
    For a set $J\subseteq [n]$, we define 
    \begin{equation}\label{likelihood quantile}
        H_{{\rm quantile},J}(\theta; x) = -\sum_{i\in J} \nabla_\theta^2 \log f_i(x_i;\theta), \text{ and } \overline{H}_{{\rm quantile},J}(\theta) = \mathbb{E}_{\theta_0} [H_{{\rm quantile},J}(\theta; x)].
    \end{equation}
    \begin{assumption}\label{assumption:quantile}
        For any $\theta_0 \in \Theta$ and $J \subseteq [n]$ with $|E|=h$, the expectation $H_{{\rm quantile},J}(\theta)$ exists for all $\theta \in B(\theta_0, e(\theta_0)) \cap \Theta$, and
    \[
    \begin{aligned}
        \mathbb{E}_{\theta_0}\left[\sup_{\theta \in B(\theta_0, e(\theta_0)) \cap \Theta} e(\theta_0)^2 \left(\lambda_{\max} (H_{{\rm quantile},E}(\theta_0) - H_{{\rm quantile},E}(\theta; X))\right)\right] &\leq \varepsilon(\theta_0)\\
        \log \mathbb{E}_{\theta_0}\left[\exp\left(\sup_{\theta \in B(\theta_0, e(\theta_0)) \cap \Theta} e(\theta_0)^2 \cdot \left(\lambda_{\max}(H_{{\rm quantile},E}(\theta; X) - H_{{\rm quantile},E}(\theta_0))\right)\right)\right] &\leq \varepsilon(\theta_0).
    \end{aligned}
    \]
Here $e(\theta_0)$ is the same constant as that appears in Assumption \ref{assumption:2}.
\end{assumption}

The theoretical guarantee of Type-I error control is provided by the following theorem:
\begin{theorem}\label{Theorem: quantile}
If we replace Assumption~\ref{assumption:3} by Assumption \ref{assumption:quantile}, then Theorems \ref{Theorem 1} and \ref{Theorem 2} continue to hold for the aCSS method with quantile estimators. 
\end{theorem}

\subsection{aCSS for Nonparametric Generalized Additive Models}\label{sec:gam}

\citet[Section 4.1]{barber2022testing} demonstrates the application of aCSS methods to canonical generalized linear models. Here, we extend their framework to nonparametric generalized additive models \citep{hastie2017generalized}, building on the approach we introduced in Section \ref{gaussian model section} for Gaussian additive models.

Consider responses \(X = (X_1, \cdots, X_n) \in \mathbb{R}^n\) and covariates \(Z = (Z_1^\T, \cdots, Z_n^\T) \in \mathbb{R}^{n \times d}\). 
With respect to some base measure \(\nu_{\mathcal{X}}\) on \(\mathcal{X} = \mathbb{R}^n\), the density of \(X\) is assumed to be in the following exponential family
\begin{equation}\label{GAM model}
    f(x; \mu) = \exp\left\{x^\T \mu - \sum_{i=1}^n a(\mu_i)\right\},
\end{equation}
where the function \(a\) is the partition function (assumed strictly convex), the parameter \(\mu = (\mu_1, \cdots, \mu_n) \in \mathbb{R}^n\) is linked to $Z$. 
In particular, we assume $\mu_i$ depends on $Z_i$ through the relationship that 
\[
\mu_i = \sum_{j=1}^d h_j(z_{ij}), \quad i = 1, \cdots, n,
\]
where the component functions \(h_1, \cdots, h_d\) are unknown and need to be estimated. For simplicity, we treat \(Z\) as fixed and denote the true unknown value of \(\mu\) as \(\mu_0 = (\mu_1^0, \cdots, \mu_n^0)^\T \in \mathbb{R}^n\).

Following the sieve-type approach introduced in Section \ref{gaussian model section}, we approximate \(\mu_0\) by \(B\theta\), where \(\theta \in \mathbb{R}^D\) and \(B \in \mathbb{R}^{n \times D}\) is the corresponding design matrix. 
The parameter of this approximation is estimated using the following estimator with group-type regularization:
\begin{equation}\label{finite GAM}
    \hat{\theta}(X) = \argmin_{{\theta} \in \mathbb{R}^{D}} \mathcal{L}({\theta}; X, W) + \sum_{j=1}^d \rho_j(\theta_{G_j}),
\end{equation}
where the loss function is defined as
\begin{equation}\label{loss GAM}
    \mathcal{L}({\theta}; X, W) = -\log f\left(X; B\theta\right) + \sigma W^\T (B\theta) 
    = -X^\T (B\theta) - \mathbf{1}_{n}^\T a(B\theta) + \sigma W^\T (B\theta)
\end{equation}
and \( G = \{G_1, \ldots, G_J\} \) represent the group structure. \( B^j \in \mathbb{R}^{n \times p_j} \) and \( \theta_{G_j} \in \mathbb{R}^{p_j} \) denote the submatrix of \( B \) and the subvector of \( \theta \), respectively, associated with the \( j \)-th group \( G_j \).

This estimation problem can be interpreted as an MLE with a group-type penalty, enabling the application of SSOSP as defined in Section \ref{penalty section}.
Similarly, the active groups are defined as \(\mathcal{A}(\theta) = \{ j \in [J] : \|\theta_{G_j}\| \neq 0 \}\), and the corresponding active coordinates are \(\mathcal{S}(\theta) = \cup_{j \in \mathcal{A}(\theta)} G_j\). 
Moving from the normal distribution to other exponential families, we revise the definition of $\hat{g}$ in \eqref{gradient gaussian} as follows:
\begin{equation}\label{graident GAM}
    \hat{g} = \hat{g}_{\rm GAM}(X, W) = X - \sigma W.
\end{equation}
\begin{lemma}[Conditional density for the GAMs]\label{condition density GAM}
    Suppose Assumption \ref{assump:4} holds. Fix any $\mu_0\in \bbR^n$ and let $(X, W, \hat{\theta}, \hat{g})$ be drawn from the following joint model
    $$
    \left\{\begin{array}{l}
    (X,W) \sim f(\cdot;\mu_0)\times N(0,I_n/n), \\
    \hat{\theta}=\hat{\theta}(X, W),
    \hat{g}=\hat{g}_{\rm GAM}(X,W). 
    \end{array}\right.
    $$
    Let ${\cal A} \subseteq[d]$. Assume that the event that $\hat{\theta}(X, W)$ is an SSOSP of (\ref{grouppenalty}) with active group index set $\left\{j\in [d]:\left\|\theta_{G_j}\right\|\neq 0\right\}={\cal A}$ 
    has positive probability. 
    When this event happens, the conditional distribution of $X \mid \hat{\theta}, \hat{g}$ has density
    \begin{equation}\label{density7}
        \begin{aligned}
            p_{\mu_0}(\cdot \mid \hat{\theta},\hat{g})\propto f\left(x;\mu_0\right)\cdot \exp\left\{-\frac{\left\|\hat{g}-x\right\|^2}{2\sigma^2/n}\right\}\cdot{F}_{\cal A}(x;\hat{\theta})\cdot {1}_{x\in \mathcal{X}_{\hat{\theta},\hat{g}}}
        \end{aligned}, 
    \end{equation}

    where $$
        \begin{aligned}
           F_{\cal A}(x;\theta)&=\det\biggl(
           \left(B^\T\nabla^2 \mathcal{L}(B\theta ; x)B\right)_{S(\theta)}\biggr.\\
           &\quad \biggl.+\,{\rm diag}\left\{\rho_j^\prime\left(\left\|\theta_{G_j}\right\|\right)\left[\frac{{I}_{p_j}}{\left\|\theta_{G_j}\right\|}-\frac{\theta_{G_j}\left[\theta_{G_j}\right]^\T}{\left\|\theta_{G_j}\right\|^3}\right]+\rho''_j\left(\left\|\theta_{G_j}\right\|\right)\frac{\theta_{G_j}}{\left\|\theta_{G_j}\right\|}\left[\frac{\theta_{G_j}}{\left\|\theta_{G_j}\right\|}\right]^\T,j\in {\cal A}\right\}\biggr)
        \end{aligned}$$
    with $${\cal X}_{{\theta},{g}}=\left\{x\in{\cal X}: \textrm{ for some } w\in\bbR^n,\theta=\hat{\theta}(x,w)\textrm { is an SSOSP of (\ref{finite GAM}), and } g=\hat{g}_{\rm GAM}(x,w)\right\}.$$
\end{lemma}
Based on the above lemma, we proceed to sample copies \(\widetilde{X}\) from
\begin{equation}\label{density8}
    p_{B\hat{\theta}}(\cdot \mid \hat{\theta},\hat{g})\propto f\left(x;B\hat{\theta}\right)\cdot \exp\left\{-\frac{\left\|\hat{g}-x\right\|^2}{2\sigma^2/n}\right\}\cdot {F}_{\cal A}(x;\hat{\theta})\cdot {1}_{x\in \mathcal{X}_{\hat{\theta},\hat{g}}}, 
\end{equation}
and implement the aCSS method to compute a p-value as before. 
A modified theoretical guarantee for Type-I error is presented as follows. 
\begin{theorem}\label{theorem:GAM} Given an estimator $\hat{\theta}(X,W): {\cal X}\times \bbR^n\to \bbR^n$. 
Let $P^*_0$ be the distribution of $(X,W)\sim f\left(\cdot,\mu_0\right)\times N(0,{I}_n/n)$ under the event that $\hat{\theta}(X, W )$ is an SSOSP of (\ref{finite GAM}).
Let $Q^*_0$ be the distribution of $(\hat{\theta}(X, W ), \hat{g}(X, W ))$ under $(X, W )\sim P^*_0$.  Suppose that the estimator $\hat{\theta}(X,W)$ satisfies 
    $$\left\{
        \begin{array}{l}
           \hat{\theta}(X,W) \text{ is an SSOSP of the penalized problem } (\ref{finite GAM})\\
            \left\|B\hat{\theta}(X,W)-\mu_0\right\|\leq r_n
        \end{array}
    \right.$$
    with probability at least $1-\delta_n$, where the probability is taken with respect to the distribution $(X,W)\sim f\left(\cdot,\mu_0\right)\times N(0,\frac{1}{n}{I}_n)$. Then we have
    \[
        \bbE_{Q_{0}^*}\left[d_{\rm TV}(p_{0}(\cdot\mid\hat{\theta},\hat{g}),p_{\hat{\theta}}(\cdot\mid\hat{\theta},\hat{g}))\right] \leq 3\sigma r_n(\mu_0) + \delta_n(\mu_0).
    \]
\end{theorem}

With Lemma \ref{lemma_weighted}, we can establish the approximate validity of the aCSS method for generalized additive models.

In the special case where \(\mu_0 = Z\theta_0\) holds for some \(\theta_0 \in \mathbb{R}^d\), 
the model simplifies to a generalized linear model (GLM), and our method aligns with the original aCSS method for GLM in \cite{barber2022testing}. The only difference is that the noise in the original aCSS is \(N(0,I_d/d)\) while the noise in our method is \(N(0, Z^\T Z/n)\). 

\section{Technical Tools for Power Analysis}\label{app: power analysis}
In this section, we introduce several technical tools that will be instrumental for the power analysis in Section~\ref{section 6}, particularly focusing on a coupling technique and the construction of test statistics that satisfy Assumption~\ref{Test statistics}.

\subsection{Coupling tool}
In this subsection, we focus on coupling an approximately exchangeable variable \( \widetilde{X} \) with an exactly exchangeable one \( \hat{X} \).

From now on, we denote the asymptotic behavior as \(n \to \infty\) and refer to the null and alternative hypotheses as \(H_0\) and \(H_{a,n}\), respectively. Let \(\widetilde{X}\) denote a random variable that is approximately exchangeable with \(X\) under \(H_0\), and \(\hat{X}\) denote perfectly exchangeable copies. Following the approach in \cite{fan2023ark}, we couple \(\widetilde{T} = T(\widetilde{X})\) with \(\widehat{T} = T(\hat{X})\) for a given test statistic \(T(\cdot)\) to compare the statistical performance of the corresponding tests. Analogous to the analysis in \citet[Section 2.3]{fan2023ark}, we introduce the following conditions. 

\begin{condition}[Coupling accuracy]\label{coupling accuracy}
    There exists a random variable \(\widehat{T} = T(\hat{X})\), exchangeable with \(T_0 = T(X)\) under \(H_0\), such that for some nonegative sequence \(\{b_n\}\),
    \[
    \text{pr}(|\widehat{T} - \widetilde{T}| \geq b_n) \to 0.
    \]
\end{condition}

\begin{condition}[Distribution of $\widehat{T}$]\label{dsitribution of hat W} Let $\hat{G}(\cdot)$ be the c.d.f of $\widehat{T}$, i.e, $\hat{G}_n(t)=\text{pr}(\widehat{T}\leq t)$ for any $t\in \mathbb{R}$. 
$\hat{G}_n(t)$ is a continuous function and it holds that $$\sup_{t\in(-\infty,\infty)}\hat{G}_n(t+b_n)-\hat{G}_n(t)\to 0.$$
\end{condition}
\begin{condition}[Distribution of $T$]\label{distribution of W} Let $G(t)$ be the c.d.f of $T$ under some alternative hypothesis $H_{a,n}$, i.e, $G_n(t)=\text{pr}(T<t\mid H_{a,n})$  for any $t\in \mathbb{R}$.  
$G_n(t)$ is a continuous function and it holds that $$\sup_{t\in(-\infty,\infty)} G_n(t+b_n)-G_n(t)\to 0.$$  
\end{condition}

Condition~\ref{coupling accuracy} ensures that the copy \(\widetilde{T}=T(\widetilde{X})\) generated by the aCSS method are coupled with the random variable \(\widehat{T}=T(\hat{X})\), meaning that they are sufficiently close in probability. Conditions~\ref{dsitribution of hat W} and \ref{distribution of W} concern the distributions of \(\widehat{T}\) and \(T\).

Let \(Q_{\alpha}(Z)\) denote the \(\alpha\)-th quantile of a random variable \(Z\). 
For a test with significance level \(\alpha\), consider the following one-sided test based on \(\widetilde{T}\): 
\[
\phi_T(X; \widetilde{T}) = \left\{
    \begin{array}{cc}
        1 & T(X) > Q_{1-\alpha}(\widetilde{T}), \\
        \gamma & T(X) = Q_{1-\alpha}(\widetilde{T}), \\
        0 & T(X) < Q_{1-\alpha}(\widetilde{T}),
    \end{array}
\right.
\]
where we have assumed that \(\widetilde{T}\) is a continuous r.v. so that its exact \(\alpha\)-th quantile exists and is unique.
A two-sided test can be constructed similarly. 
The following theorem establishes the asymptotic validity of the test $\phi_T(X;\widetilde{T})$ and approximates its power by that of $\phi_T(X;\widehat{T})$. 

\begin{theorem}\label{coupling}
    Under Conditions~\ref{coupling accuracy} and \ref{dsitribution of hat W}, the test $\phi_T(X;\widetilde{T})$ is asymptotically valid. Under Conditions \ref{coupling accuracy} and \ref{distribution of W}, the asymptotic power of $\phi_T(X;\widetilde{T})$ matches that of $\phi_T(X;\widehat{T})$ for the alternative hypothesis $H_{a,n}$. This holds for both one-sided and two-sided cases.
\end{theorem}

Theorem~\ref{coupling} is used in our proof of Theorem~\ref{connection}, which shows the asymptotic equivalence between aCSS CRT and conditional CRT. 
Besides, Theorem~\ref{coupling} may be of independent interest. 

The proof of Theorem~\ref{pf:coupling} is given in Section~\ref{pf:coupling}. 

\subsection{Discussion of Test Statistics in Assumption \ref{Test statistics}}\label{sec:test-stat}

In this section, we discuss the choice of test statistics that satisfy Assumption \ref{Test statistics}. 

Before introducing our construction of $\widetilde{Y}$ in Assumption \ref{Test statistics}, we first discuss a seemingly ideal but infeasible choice and two potential choices that appear to be natural. 

\textbf{Infeasible ideal choice.}
Under model~\eqref{Setting}, we have \(Y \sim N(Z(\beta\theta_0 + \xi), (\nu^2 + \beta^2){I}_n)\). 
If the parameter $\xi$ is known, the residual \(\varepsilon^\prime = Y - Z\xi - Z\beta\theta_0\) is an ideal candidate for \(\widetilde{Y}\). 
This candidate satisfies the first inequality in \eqref{power condition on statistic} of Assumption \ref{Test statistics} by using the independence between \(\varepsilon\)  and \(Z\), and it allows for the maximal value for the constant \(C_2\) in \eqref{power condition on statistic}. 
However, this candidate is usually infeasible because the true value of \(\xi\) is unknown. 

\textbf{Sparse estimation.}
\citet{zhu2018linear} assumes that \(\xi\) is also sparse and utilizes the residual \(Y - Z\hat{\xi}\). 
The first inequality in \eqref{power condition on statistic} of Assumption \ref{Test statistics} is met if \(\hat{\xi}\) is estimated via the Dantzig selector \citep{10.1214/009053606000001523}, owing to its optimization constraints. 
The validity of the second inequality in \eqref{power condition on statistic} heavily depends on the sparsity of \(\xi\). For full details, see \citet[Section 3]{zhu2018linear}.

\textbf{Direct $Y$.}
Another natural choice for \(\widetilde{Y}\) is \(Y\) itself, which clearly satisfies the second inequality in \eqref{power condition on statistic} of Assumption \ref{Test statistics}. However, it generally does not meet the first inequality in \eqref{power condition on statistic} for arbitrary \(\xi\). 
Consider the decomposition that
\[
Z^\T Y = Z^\T(Y - Z(\beta\theta_0 + \xi)) + Z^\T Z(\beta\theta_0 + \xi).
\]
The first term is known to be of order \(O(\sqrt{n \log d})\) due to the independence between \(Z\) and the residual $\varepsilon^\prime$. However, the second term cannot be bounded by \(O(\sqrt{n \log d})\). 
Specifically, if we set \(Z(\beta\theta_0 + \xi)\) as the first column of \(Z\), denoted as \(Z_1\), then \(\|Z^\T Z(\beta\theta_0 + \xi)\|_\infty \geq Z_1^\T Z_1\) is of order \(n\). Consequently, \(Y\) does not satisfy the first inequality in \eqref{power condition on statistic}  in this case.

We have thus discussed two potential choices for \(\widetilde{Y}\) that meet Assumption \ref{Test statistics} under some additional requirements. 
To relax the requirements, we consider the following optimization problem for $\widetilde{Y}$: 
\begin{equation}\label{program}
    \begin{aligned}
        \max\quad & Y^\T \widetilde{Y} \quad \text{\rm subject to}\quad & \|Z^\T \widetilde{Y}\|_\infty \leq \lambda_n,\ \|\widetilde{Y}\| \leq 1.
    \end{aligned}
\end{equation}
\eqref{program} is a convex program and therefore computationally efficient. 
Furthermore, it has a unique maximizer, if \( Y \) does not lie in the row space of \( Z \), which holds almost surely. 

The following lemma demonstrates that the maximizer \(\widetilde{Y}\) of \eqref{program} satisfies Assumption \ref{Test statistics}. 

\begin{lemma}\label{program lemma}
    Suppose \(Y^\T(Y - Z(\xi + \beta\theta_0)) > 0\) and the program \eqref{program} has a unique maximizer. If \(\lambda_n\) is chosen large enough such that 
    \begin{equation}\label{eq:program-requirement-lambda}
    \frac{\|Y\|}{Y^\T(Y - Z(\xi + \beta\theta_0))}\|Z^\T(Y - Z(\xi + \beta\theta_0))\|_\infty < \lambda_n,
    \end{equation}
    then we have
    \[
    \widetilde{Y}^\T(Y - Z(\xi + \beta\theta_0)) \geq \frac{\left[Y^\T(Y - Z(\xi + \beta\theta_0))\right]^2}{\|Y\|^2 \|Y - Z(\xi + \beta\theta_0)\|},\quad \text{and}\quad \|\widetilde{Y}\| \geq \frac{Y^\T(Y - Z(\xi + \beta\theta_0))}{\|Y\| \|Y - Z(\xi + \beta\theta_0)\|}.\]
\end{lemma}
The proof of Lemma \ref{program lemma} is given in Section~\ref{pf:program lemma}.

We explain below how Lemma~\ref{program lemma} connects program~\eqref{program} and Assumption~\ref{Test statistics}. 
Note that stochastically, both \(\|Y\|\) and \(\|Y - Z(\xi + \beta\theta_0)\|\) are of order \(\sqrt{n}\), \(\|Z^\T(Y - Z(\xi + \beta\theta_0))\|_\infty\) is of order \(\sqrt{n \log d}\), and \(Y^\T(Y - Z(\xi + \beta\theta_0))\) is of order \(n\). Thus, Lemma~\ref{program lemma} demonstrates that the output of program~\eqref{program} satisfies the conditions for the test statistics in Assumption~\ref{Test statistics}.

To see why the requirement in Lemma~\ref{program lemma} can be met, we express \(Y^\T(Y - Z(\xi + \beta\theta_0))\) as
\begin{equation}\label{eq:program-mainterm-decomposition}
    Y^\T(Y - Z(\xi + \beta\theta_0))=\left\|\varepsilon^\prime\right\|_2^2+(Y - Z(\xi + \beta\theta_0))^T\varepsilon^\prime,
\end{equation}
where the first term $\left\|\varepsilon^\prime\right\|_2^2$ is of order \( n \) and the second term $(Y - Z(\xi + \beta\theta_0))^T\varepsilon^\prime$ is of order \( \sqrt{n} \).  
Therefore, \(Y^\T(Y - Z(\xi + \beta\theta_0)) > 0\) holds with probability approaching 1. 
Regarding the choice of \( \lambda_n \) required by \eqref{eq:program-requirement-lambda} in Lemma~\ref{program lemma}, we consider the numerator and the denominator separately. 
Since \( \varepsilon' = Y - Z(\xi + \beta \theta_0) \sim N(0, \nu_Y^2 I) \) is a Gaussian vector independent of \( Z \), the distribution of \( Z^\top Y - Z^\top Z(\xi + \beta \theta_0) \mid Z \) is known and its high-probability quantile of $\|Z^\top (Y -  Z(\xi + \beta \theta_0))\|_\infty$ (given $Z$) can be characterized. 
As shown in \eqref{eq:program-mainterm-decomposition}, the term \( Y^\top(Y - Z(\xi + \beta \theta_0)) \) in the denominator can be decomposed into two components, among which \( \|\varepsilon'\|_2^2 \) dominates asymptotically. This decomposition enables us to select \( \lambda_n \) appropriately to ensure that \eqref{eq:program-requirement-lambda} is satisfied for the theoretical guarantee.

\section{Proof of main results}\label{app: proof}

\subsection{Proof of Theorems \ref{Theorem 1}, \ref{theorem: MTLE}, \ref{Theorem: quantile} and \ref{theorem:GAM}}
\begin{proof}
    Fundamentally, the proofs of Theorems~\ref{Theorem 1}, \ref{theorem: MTLE}, \ref{Theorem: quantile}, and \ref{theorem:GAM} follow a similar structure to that of \citet[Theorem 1]{zhu2023approximate}. There are mainly two steps:
\begin{itemize}
    \item \textbf{Reduction to total variation distance:} In this step, we show that the distance to exchangeability or the inflation of Type-I error can be bounded by
    \[
    \mathbb{E}_{Q_{\theta_0}^*} \left[ d_{\mathrm{TV}} \left( p_{\theta_0}(\cdot \mid \hat{\theta}, \hat{g}), p_{\hat{\theta}}(\cdot \mid \hat{\theta}, \hat{g}) \right) \right],
    \]
    where \( Q_{\theta_0}^* \) is the distribution of \( (\hat{\theta}, \hat{g}) \). This part is established in Lemma~\ref{lemma_weighted}, and differs from the approach of \cite{zhu2023approximate} as we allow more general sampling schemes.

    \item \textbf{Bounding the total variation distance:} In this step, we bound the expected total variation distance. A central idea is to demonstrate that the ratio \( f(x; \theta)/f(x; \theta_0) \) is nearly constant over \( p_{\theta_0}(\cdot \mid \hat{\theta}, \hat{g}) \). This is the part we will discuss in detail later. The remaining steps follow those of \cite{zhu2023approximate}.
\end{itemize}
   
For Theorem \ref{Theorem 1}, we take a Taylor series expansion for the function \(\theta \to \log f(x; \theta)\):
\begin{align*}
    \log f(x; \theta_0) - \log f(x; \theta) 
    &= (\theta_0 - \theta)^\T \nabla_\theta \log f(x; \theta) + \int_{t=0}^1 t (\theta - \theta_0)^\T \nabla_\theta^2 \log f(x; \theta_t) (\theta - \theta_0) dt,
\end{align*}
where we write \(\theta_t = (1 - t)\theta_0 + t\theta\). Therefore, we have
$$\begin{aligned}
        \frac{f(x; \theta)}{f(x; \theta_0)}
        &= \exp \left\{ \log f(x; \theta) - \log f(x; \theta_0) \right\} \\
        &= \exp \left\{ 
        -(\theta_0 - \theta)^\top \nabla_{\theta} \log f(x; \theta) 
        - \int_{0}^{1} t (\theta - \theta_0)^\top \nabla_{\theta}^2 \log f(x; \theta_t)(\theta - \theta_0) \rmd t
        \right\} \\
        &= \exp \left\{
        (\theta_0 - \theta)^\top (\nabla_{\theta} \mathcal{L}(x; \theta) - g) 
        + \int_{0}^{1} t (\theta - \theta_0)^\top (H(\theta_t; x) - H(\theta_t)) (\theta - \theta_0) \rmd t \right. \\
        &\quad \left.
        + (\theta_0 - \theta)^\top (g - \nabla_{\theta} R(\theta)) 
        + \int_{0}^{1} t (\theta - \theta_0)^\top H(\theta_t)(\theta - \theta_0) \rmd t
        \right\},
    \end{aligned}
    $$
where the last step holds using the fact that
$
-\nabla_\theta \log f(x; \theta) = \nabla_\theta \mathcal{L}(x; \theta) - \nabla_\theta R(\theta)$ by the definition of \(\mathcal{L}(x;\theta)\).

    For Theorems \ref{theorem: MTLE} or \ref{Theorem: quantile}, the above equality still holds. However, the key difference lies in the use of subsampling techniques by these methods, which results in a certain part of each copy being identical to the original sample. Consequently, while the aCSS with MLE relies on Assumption \ref{assumption:3} to bound \(H(\theta_t; x) - H(\theta)\), the analysis here must be adjusted to rely on Assumption \ref{assumption:MTLE} or \ref{assumption:quantile}, which states the condition based on the selected subsamples.

     For Theorem \ref{theorem:GAM}, although GAM is a nonparametric model, we model $X$ and $\mu$ through the parametric model that $\mu=B\theta$, and therefore,
    $$\frac{f(x;B\theta)}{f(x;\mu_0)}=\exp\left\{\sum_{i=1}^n \left(a(\mu_i^0)- a(B_i\theta)\right)\right\}\cdot \exp\left\{x^\T(B\theta-\mu_0)\right\}.$$
     Let $g=x-\sigma w$, we have 
    $$ \frac{f(x;B\theta)}{f(x;\mu_0)}\exp\left\{\sum_{i=1}^n\left(a(B_i\theta)-a(\mu_i^0)\right)-g^\T(B\theta-\mu_0)\right\}= \exp\left\{\sigma(B\theta-\mu_0)^\T w\right\}.$$
    Note that $\exp\left\{\sum_{i=1}^n\left(a(B_i\theta)-a(\mu_i^0)\right)-g^\T(B\theta-\mu_0)\right\}$ is fixed given $\hat{\theta},\hat{g}.$
    Therefore, 
  the ratio ${f(x;B\theta)}/{f(x;\mu_0)}$
    can be bounded by a constant multiplied by \(\exp\left\{\sigma \|B\theta - \mu_0\| \|w\|\right\}\).  

    The remaining steps follow similarly to those in Theorem \ref{Theorem 1}, and we refer to \citet[Section A.1.2]{zhu2023approximate} for details.
\end{proof}

\subsection{Proof of Theorem \ref{theorem: linear additive} and Theorem \ref{theorem: gacss}}\label{app:pf-gacss}
\begin{proof}
    Note that for any $\mu,\mu^\prime\in \bbR^n$ and any positive definite matrix $\Sigma\in \bbR^{n\times n}$,
    $$\begin{aligned}
        d_{\rm TV}\left({N}(\mu,\Sigma),{N}(\mu^\prime,\Sigma)\right)&\leq \sqrt{\frac{1}{2}d_{\rm KL}\left({N}(\mu,\Sigma),{N}(\mu^\prime,\Sigma)\right)}=\frac{1}{2}\left\|\Sigma^{-1/2}(\mu-\mu^\prime)\right\|
    \end{aligned}.$$
    
    For Theorem \ref{theorem: linear additive}, we bound the total variation between \eqref{eq:additive-Gaussian-true} and \eqref{sampling gaussian} as follows: 
    $$\begin{aligned}
        d_{\rm TV}\left(p_{\mu_0}(\cdot\mid \hat{\theta},\hat{g}),p_{B\hat{\theta}}(\cdot\mid\hat{\theta},\hat{g})\right)&\leq \frac{1}{2\nu}\left\|\left(1+\frac{n}{\sigma^2\nu^2}\right)^{1/2}\cdot \left(1+\frac{n}{\sigma^2\nu^2}\right)^{-1}(\mu_0-B\hat{\theta})\right\|\\
        &\leq \frac{1}{2\nu}\left(1+\frac{n}{\sigma^2\nu^2}\right)^{-1/2}\left\|\mu_0-B\hat{\theta}\right\|\leq \frac{\sigma}{2\sqrt{n}}\left\|\mu_0-B\hat{\theta}\right\|.
    \end{aligned}$$
    With the assumptions in Theorem \ref{theorem: linear additive}, since the total variation distance is always bounded by $1$, and we therefore have 
    $$\bbE_{Q_{\theta_0}^*}\left( d_{\rm TV}\left(p_{\mu_0}(\cdot\mid \hat{\theta},\hat{g}),p_{B\hat{\theta}}(\cdot\mid\hat{\theta},\hat{g})\right)\right)\leq \frac{\sigma}{2\sqrt{n}}r(\mu_0)+\delta(\mu_0).$$
    
    For Theorem \ref{theorem: gacss}, we bound the total variation between \eqref{eq: conditional gaussian aCSS} and \eqref{sampling gaussian aCSS} as follows: 
    $$\begin{aligned}
        d_{\rm TV}\left(p_{\mu_0}(\cdot\mid X_{\rm noise}),p_{\hat{\mu}}(\cdot\mid X_{\rm noise})\right)&\leq \frac{1}{2}\left\|\left(\frac{\sigma^2\nu^2}{\sigma^2+\nu^2}\right)^{-1/2}\cdot \left(\frac{\sigma^2}{\sigma^2+\nu^2}\right)(\mu_0-\hat{\mu})\right\|\\
        &\leq \frac{1}{2\nu}\left(\frac{\sigma^2}{\sigma^2+\nu^2}\right)^{1/2}\left\|\mu_0-B\hat{\theta}\right\|\leq \frac{\sigma}{2\nu^2}\left\|\mu_0-B\hat{\theta}\right\|.
    \end{aligned}$$
    With the assumptions in Theorem \ref{theorem: gacss}, since the total variation distance is always bounded by $1$, and we therefore have 
    $$\bbE_{Q_{\theta_0}^*}\left( d_{\rm TV}\left(p_{\mu_0}(\cdot\mid \hat{\theta},\hat{g}),p_{B\hat{\theta}}(\cdot\mid\hat{\theta},\hat{g})\right)\right)\leq \frac{\sigma}{2\nu^2}r(\mu_0)+\delta(\mu_0).$$

    Given the bounds on the total variation distances, we complete the proof using Lemma \ref{lemma_weighted}. 
\end{proof}
\subsection{Proof of Theorem \ref{connection}}\label{pf:connection}
\begin{proof}
    For conditional CRT, the copies $\widetilde{X}^*_{\rm css}$ is generated conditioning on ${Z}_*$ and $\left({Z}_*\right)^\T {X}_*$. Consider the following equivalent variables
    $$\begin{aligned}
        \widehat{T}&=\frac{Y^\T \widetilde{X}_{\rm css}}{\left\|Y\right\|},\quad
        \widetilde{T}&=\frac{Y^\T \widetilde{X}_{\rm acss}}{\left\|Y\right\|},\quad
        {T}&=\frac{Y^\T X}{\left\|Y\right\|}.
    \end{aligned}$$
    
    Recall that $X_i\mid Z_i \sim N(Z_i^\T \theta_0, 1)$. 
    Then, from the construction of copies in the CSS and aCSS methods (see Appendix~G in the supplement of \citet{wang2022high} and \eqref{aCSS unre CRT}, respectively), we have
    \begin{equation}\label{eq:acss-css-coupling}
    \left\{ \begin{aligned}
        \widetilde{X}^*_{\rm css}\mid Z_*,Z_*^\T X_*&\stackrel{\rmd}{=}Z_*\left(Z_*^\T Z_*\right)^{-1}Z_*^\T X_*+\left(I_{n_*}-Z_*(Z_*^\T Z_*)^{-1}Z_*^\T \right)\varepsilon_0;\\
        \widetilde{X}^*_{\rm acss}\mid Z_*,(Z_*)^\T X_*+\sigma W&\stackrel{\rmd}{=}  Z_*\left(Z_*^\T Z_*\right)^{-1}(Z_*^\T X_*+\sigma W)+\left(I_{n_*}+\frac{d}{\sigma^2}Z_*Z_*^\T \right)^{-1/2}\varepsilon_0,
    \end{aligned}
    \right.
    \end{equation}
    where \(\varepsilon_0 \sim N(0, I_n)\) and is independent of all other random variables. 
Therefore, we can couple
 $\widetilde{X}^*_{\rm css}$ and $\widetilde{X}^*_{\rm acss}$ through the common noise vector $\varepsilon_0$ using the right hand sides of \eqref{eq:acss-css-coupling}. 
 This coupling is used in the following. 
    
    By the proof of Theorem~9 in the supplement of \citet{wang2022high}, 
      the conditional distribution of $\widehat{T}$ given $\left({Z}_*,\left({Z}_*\right)^\T {X}_*,Y\right)$ is a normal distribution with variance bounded away from $0$ and $\infty$ under $H_0$ and $H_{1,n}$. 
      The same holds for $T$. 
    Since \(\left\|\widehat{T} - \widetilde{T}\right\| \leq \left\|\widetilde{X}_{\rm css} - \widetilde{X}_{\rm acss}\right\|\), we can apply Theorem~\ref{coupling} to obtain the asymptotic equivalence if we can show that \(\left\|\widetilde{X}_{\rm css} - \widetilde{X}_{\rm acss}\right\| \to 0\) in probability. 

    By Markov's inequality, for any $\delta>0$, we have 
$$
\begin{aligned}
  \text{pr}\left( \left\|\widetilde{X}_{\rm css} - \widetilde{X}_{\rm acss}\right\| > \delta \right) & \leq \delta^{-2} \mathbb{E} \left\|\widetilde{X}_{\rm css} - \widetilde{X}_{\rm acss}\right\|^2  \\
  & = \delta^{-2} \frac{n}{n_*}\mathbb{E} \left\|\widetilde{X}^*_{\rm css} - \widetilde{X}^*_{\rm acss}\right\|^2, 
\end{aligned}
$$
  where the last equation holds because all \(X_i\) are exchangeable.
Therefore, it suffices to show that 
    \[
    \mathbb{E} \left\|\widetilde{X}^*_{\rm css} - \widetilde{X}^*_{\rm acss}\right\|^2 = o\left(\frac{n_*}{n}\right).
    \]
Using the coupling in \eqref{eq:acss-css-coupling} and the basic inequality that $(a+b)^2\leq 2 a^2+2 b^2$, we have 

\begin{align}
& \bbE\left\|\widetilde{X}^*_{\rm css}-\widetilde{X}^*_{\rm acss}\right\|^2 \nonumber\\
=& \bbE\left\| - \sigma  Z_*\left(Z_*^\T Z_*\right)^{-1} W + \left[\left(I_{n_*}-Z_*(Z_*^\T Z_*)^{-1}Z_*^\T \right)- \left(I_{n_*}+\frac{d}{\sigma^2}Z_*Z_*^\T \right)^{-1/2} \right]\varepsilon_0   \right\|^2  \nonumber \\
\leq &  2 \sigma^2 \bbE\left\|Z_*\left(Z_*^\T Z_*\right)^{-1} W\right\|^2 + 2 \bbE\left\|\left[I_{n_*}-Z_*(Z_*^\T Z_*)^{-1}Z_*^\T -\left(I_{n_*}+\frac{d}{\sigma^2}Z_*Z_*^\T \right)^{-1/2}\right]\varepsilon_0\right\|^2. \label{two parts}
\end{align}

To proceed, we first introduce some notations. 
Let $Z_*=U\Sigma V^\T $ be the SVD, where 
    $$
\begin{aligned}
U& \in \bbR^{n_*\times n_*} \text{ s.t. }UU^\T =U^\T U={I}_{n_*}, \\
V &  \in \bbR^{d\times d},  \text{ s.t. } V^\T V=VV^\T ={I}_d, \\
\Sigma & =\begin{bmatrix}
        \Lambda \\
        \mathbf{0}_{(n_*-d)\times d}
    \end{bmatrix}
    \in \bbR^{n_*\times d},\quad \text{ with } \Lambda ={\rm diag}\left\{\lambda_1,\lambda_2,\cdots,\lambda_d\right\}\in \bbR^{d\times d}. 
\end{aligned}
$$
Since $Z_i$'s are normally distributed with non-singular covariance matrix, $\lambda_j$'s are positive a.s.

Using the SVD of $Z_*$, we can straightforwardly obtain the following expressions: 
    $$\begin{aligned}
        Z_*^\T Z_*&=V{\rm diag}\left\{\lambda_1^2,\cdots,\lambda_d^2\right\}V^\T , \\
        {I}_{n_*}-Z_*\left(Z_*^\T Z_*\right)^{-1}Z_*^\T &=U{\rm diag}\left\{{0}_{d},I_{n_*-d}\right\}U^\T ,\\
        {I}_{n_*}+\frac{d}{\sigma^2}Z_*Z_*^\T &=U\left(\begin{array}{cc}
            {\rm diag}\left\{1+\frac{d}{\sigma^2}\lambda_1^2,\cdots,1+\frac{d}{\sigma^2}\lambda_d^2\right\}&{0}\\
            {0}&{I}_{n_*-d}
        \end{array}\right)U^\T . 
    \end{aligned}$$

    For the first term in (\ref{two parts}), we have 
    $$\begin{aligned}
        \bbE\left\|Z_*\left(Z_*^\T Z_*\right)^{-1}W\right\|^2& =  {\rm tr}\left[\bbE\left(W^\T (Z_*^\T Z_*)^{-1}Z_*^\T Z_*(Z_*^\T Z_*)^{-1}W\right)\right]\\
        & =  {\rm tr}\left[ (Z_*^\T Z_*)^{-1}  \bbE \left(WW^\T\right)\right]\\
        &=\frac{1}{d}{\rm tr}\left[(Z_*^\T Z_*)^{-1}\right]\\
        &=\frac{1}{d}\sum_{i=1}^{d}\frac{1}{\lambda_i^2}. 
    \end{aligned}$$

    For the second term in (\ref{two parts}), we have

    $$\begin{aligned}
        &E\left\|\left[I_{n_*}-Z_*(Z_*^\T Z_*)^{-1}Z_*^\T -\left(I_{n_*}+\frac{d}{\sigma^2}Z_*Z_*^\T \right)^{-1/2}\right]\varepsilon_0\right\|^2\\
        =&E\left\|\left[U{\rm diag}\left\{{0}_{d},I_{n_*-d}\right\}U^\T -U\left(\begin{array}{cc}
             {\rm diag}\left\{\left(1+\frac{d}{\sigma^2}\lambda_1^2\right)^{-1/2},\cdots,\left(1+\frac{d}{\sigma^2}\lambda_d^2\right)^{-1/2}\right\}&{0}\\
            {0}&{I}_{n_*-d}
        \end{array}\right)U^\T \right]\varepsilon_0\right\|^2\\
      =& \sum_{i=1}^d\left(1+\frac{d}{\sigma^2}\lambda_i^2\right)^{-1}= \sum_{i=1}^d\frac{\sigma^2}{d\lambda_i^2+\sigma^2} \leq \sigma^2\frac{1}{d}\sum_{i=1}^{d}\frac{1}{\lambda_i^2}. 
    \end{aligned}$$

Combining the two bounds, we obtain
\begin{align*}
\bbE\left\|\widetilde{X}^*_{\rm css}-\widetilde{X}^*_{\rm acss}\right\|^2
\leq 4\sigma^2 \frac{1}{d}\sum_{i=1}^{d}\frac{1}{\lambda_i^2}. 
\end{align*}

Therefore, we only need to prove 
$$\frac{\sigma^2 n}{ d ~ n_*  }\sum_{i=1}^{d}\frac{1}{\lambda_i^2}=o(1).
$$

    Note that $\lambda_1^2/d,\cdots,\lambda_d^2/d$ is the eigenvalues of $Z^\T Z/d$. By the Marchenko-Pastur theorem \citep{marchenko1967distribution}, with probability 1, the empirical
    distribution of $\left\{\lambda_i^2/d : 1 \leq i \leq  d\right\}$ converges weakly to the distribution with probability
    density function
    $$p_\lambda(x)=\frac{1}{2\pi}\frac{\sqrt{(\lambda_+-x)(x-\lambda_-)}}{\lambda x}{1}\left(x\in [\lambda_-,\lambda_+]\right),$$
    where $$\lambda=\lim_{n\to\infty}\frac{d}{n_*}\in (0,1),\quad \lambda_-=(1-\sqrt{\lambda})^2,\quad \lambda_+=(1+\sqrt{\lambda})^2.$$
    We then have 
    $$\frac{1}{d}\sum_{i=1}^n\frac{d}{\lambda_i^2}{1}\left(\frac{\lambda_i^2}{d}\geq \frac{1}{2}\lambda_-\right)\stackrel{{\rm a.s.}}{\to} \int_{\lambda_-/2}^\infty \frac{1}{x}p_\lambda(x)\rmd x.$$
    In addition, as $n_*,d\to\infty$, we have $\min_{1\leq i\leq d}\lambda_i^2/d\geq {\lambda_-}/{2}$ almost surely \cite{10.1214/aop/1176989118}, so we actually have 
    $$\sum_{i=1}^n\frac{1}{\lambda_i^2}\stackrel{{\rm a.s.}}{\to}\int_{\lambda_-/2}^\infty \frac{1}{x}p_\lambda(x)\rmd x.$$

    By the assumption of the theorem, $\sigma^2=o\left( d \cdot n_*/{n} \right)$, and thus we have $\frac{\sigma^2 n}{ d ~ n_*  }\sum_{i=1}^{d}\frac{1}{\lambda_i^2}\to 0$.
\end{proof}
\subsection{Proof of Theorem \ref{power acss}}\label{pf:power acss}

\textbf{Step 1: Express the conditional distributions of test statistics.}

Recall that we sample copies $\widetilde{X}_{\rm acss}$ from 
$$\widetilde{X}_{\rm acss}\mid X_{\rm noise}, Y, Z\sim N\left(\frac{\sigma^2}{\sigma^2+1}Z\hat{\theta}+\frac{1}{\sigma^2+1}X_{\rm noise},\frac{\sigma^2}{\sigma^2+1}{I}_n\right),$$
and the test statistic derived from a copy satisfies 
$$\widetilde{T}=\frac{\widetilde{Y}^\T \widetilde{X}_{\rm acss}}{\|\widetilde{Y}\|} \mid X_{\rm noise}, Y, Z \sim N\left(\frac{\sigma^2}{\sigma^2+1}\frac{\widetilde{Y}^\T Z\hat{\theta}}{\|\widetilde{Y}\|}+\frac{1}{\sigma^2+1}\frac{\widetilde{Y}^\T X_{\rm noise}}{\|\widetilde{Y}\|},\frac{\sigma^2}{\sigma^2+1}\right).$$

Note that given $(X_{\rm noise}, Y, Z)$, the conditional density of $X$ is given by 
$$p(x\mid X_{\rm noise}, Y,Z)\propto \exp\left\{-\frac{1}{2\nu^2}\left\|Y-x\beta-Z\xi\right\|^2-\frac{1}{2}\left\|x-Z\theta_0\right\|^2-\frac{1}{2\sigma^2}\left\|x-X_{\rm noise}\right\|^2\right\},$$
which simplifies to 
$$X\mid X_{\rm noise},Y,Z\sim N\left(\left(\frac{1}{\sigma^2}+\frac{\beta^2}{\nu^2}+1\right)^{-1}\left[\frac{X_{\rm noise}}{\sigma^2}+\frac{\beta}{\nu^2}(Y-Z\xi)+Z\theta_0\right],\left(\frac{1}{\sigma^2}+\frac{\beta^2}{\nu^2}+1\right)^{-1}{I}_n\right).$$
The test statistic derived from $X$, i.e., $T=\frac{\widetilde{Y}^\T X}{\|\widetilde{Y}\|}$, satisfies 
$$T\mid X_{\rm noise},Y,Z
\sim N\left(\left(\frac{1}{\sigma^2}+\frac{\beta^2}{\nu^2}+1\right)^{-1}\left[\frac{\widetilde{Y}^\T Z\theta_0}{\|\widetilde{Y}\|}+\frac{\beta}{\nu^2}\frac{\widetilde{Y}^\T (Y-Z\xi)}{\|\widetilde{Y}\|}+\frac{1}{\sigma^2}\frac{\widetilde{Y}^\T X_{\rm noise}}{\|\widetilde{Y}\|}\right],\left(\frac{1}{\sigma^2}+\frac{\beta^2}{\nu^2}+1\right)^{-1}\right).$$

For simplicity, we denote the mean and standard deviation of the above conditional distribution of \(T\) given $(X_{\rm noise}, Y, Z)$ as follows:
 $$\begin{aligned}
    \mu_\beta(X_{\rm noise},\widetilde{Y},Z)&=\sigma^2_\beta\left[\frac{\widetilde{Y}^\T Z\theta_0}{\|\widetilde{Y}\|}+\frac{\beta}{\nu^2}\frac{\widetilde{Y}^\T (Y-Z\xi)}{\|\widetilde{Y}\|}+\frac{1}{\sigma^2}\frac{\widetilde{Y}^\T X_{\rm noise}}{\|\widetilde{Y}\|}\right],\\
    \sigma_\beta&=\left(\frac{1}{\sigma^2}+\frac{\beta^2}{\nu^2}+1\right)^{-1/2}=\sqrt{\frac{\sigma^2\nu^2}{(\sigma^2+1)\nu^2+\beta^2\sigma^2}}.\\
\end{aligned}$$
Similarly, for \(\widetilde{T}\) we introduce the following: 
 $$\begin{aligned}
    \mu^\prime(X_{\rm noise},\widetilde{Y},Z)&=\frac{\sigma^2}{\sigma^2+1}\frac{\widetilde{Y}^\T Z\hat{\theta}}{\|\widetilde{Y}\|}+\frac{1}{\sigma^2+1}\frac{\widetilde{Y}^\T X_{\rm noise}}{\|\widetilde{Y}\|}, \\
    \sigma_0&=\sqrt{\frac{\sigma^2}{\sigma^2+1}}.
\end{aligned}$$

\textbf{Step 2: Analyze the power.}

Here, we only consider the one-sided test: $$H_0:\beta=0\qquad {\rm versus}\qquad H_1:\beta>0.$$ The proof of two sided test follows similarly. Since we reject the null hypothesis when $T(X)$ is larger than $\mu^\prime(X_{\rm noise},\widetilde{Y},Z) + \sigma_0 z_{1-\alpha}$, the conditional power function $\operatorname{CP}(\beta)$ given $( X_{\rm noise},Y,Z)$ is given by  
\begin{equation}\label{eq:thm6.2-power}
\begin{aligned}
 \operatorname{CP}(\beta)  = &\text{pr}_\beta\left( T > \mu^\prime(X_{\rm noise},\widetilde{Y},Z) + \sigma_0 z_{1-\alpha}  \mid X_{\rm noise},Y,Z\right)\\
  = &\text{pr}_\beta\left(\mu_\beta(X_{\rm noise},\widetilde{Y},Z) + \sigma_{\beta} \xi_{X} > \mu^\prime(X_{\rm noise},\widetilde{Y},Z) + \sigma_0 z_{1-\al}\mid X_{\rm noise},Y,Z\right)\\
    =&\Phi\left\{\sigma_\beta^{-1}\left[\mu_\beta(X_{\rm noise},\widetilde{Y},Z)-\mu^\prime(X_{\rm noise},\widetilde{Y},Z)-\sigma_0 z_{1-\al}\right]    \right\},
\end{aligned}\end{equation}
where $\xi_{X}$ denotes a standard normal variable independent of other variables and in the last equation $X_{\rm noise},Y,Z$ are viewed as fixed. 

Since we only focus on the null and local alternative points, we have $\beta\to 0$,  which leads to 
\begin{equation}\label{eq:thm6.2-variance-ratio}
\sigma_0\sigma_\beta^{-1}=\left(\frac{1}{\sigma^2}+1\right)^{-1/2}\left(\frac{1}{\sigma^2}+\frac{\beta^2}{\nu^2}+1\right)^{1/2}\to 1.
\end{equation}
Denote the scaled difference in means appeared in the expression of $\text{power}(\beta)$ in \eqref{eq:thm6.2-power} as 
$$\begin{aligned}
    \Delta\mu(X_{\rm noise},\widetilde{Y},Z)& :=\sigma_\beta^{-1}\left[\mu_\beta(X_{\rm noise},\widetilde{Y},Z)-\mu^\prime(X_{\rm noise},\widetilde{Y},Z)\right]\\
    &=\sigma_\beta^{-1}\sigma_0^2\frac{\widetilde{Y}^\T Z(\theta_0-\hat{\theta})}{\|\widetilde{Y}\|}+\sigma_\beta^{-1}\frac{\sigma^2\beta}{(\sigma^2+1)\nu^2+\beta^2\sigma^2}\frac{\widetilde{Y}^\T (Y-Z\xi-Z\beta\theta_0)}{\|\widetilde{Y}\|}\\
    &\qquad -\sigma_\beta^{-1}\frac{\beta^2\sigma^2}{(\sigma^2+1)[(\sigma^2+1)\nu^2+\beta^2\sigma^2]}\frac{\widetilde{Y}^\T (X_{\rm noise}-Z\theta_0)}{\|\widetilde{Y}\|}, 
\end{aligned}$$
where we have omitted the straightforward manipulation using the expressions of $\sigma_\beta$ and $\sigma_0$. 

To bound the first item in $\Delta\mu(X_{\rm noise},\widetilde{Y},Z)$, we recall the conditions that  $\frac{\|Z^\T \widetilde{Y}\|_\infty}{\|\widetilde{Y}\|} \leq C_1 \sqrt{\log d}$ in Assumption~\ref{Test statistics}, $\|\hat{\theta} - \theta_0\|_1 = O_p\left(s\sqrt{{\log d}/{n}}\right)$ in Assumption~\ref{estimator}, and $\sigma = C_3 n^{-\gamma_n}$ in the premise of the theorem. 
By \eqref{eq:thm6.2-variance-ratio}, $\sigma_\beta^{-1}\leq 2\sigma_0$ for small enough $\beta$. 
Consequently, we have 
$$
\begin{aligned}
\left|\sigma_\beta^{-1}\sigma_0^2\frac{\widetilde{Y}^\T Z(\theta_0-\hat{\theta})}{\|\widetilde{Y}\|}\right|
&\leq 2\frac{\sigma}{\sqrt{\sigma^2+1}}\cdot \frac{\|\widetilde{Y}^\T Z\|_\infty}{\|\widetilde{Y}\|}\|\theta_0-\hat{\theta}\|_1 \\
&=O_p\left(n^{-\gamma_{n}}\sqrt{\log d}s\sqrt{\frac{\log d}{n}}\right)\\
&=O_p\left(n^{-1/2-\gamma_{n}}s \log d \right) \\
&=o_p(1),
\end{aligned}
$$
where the last equation is due to Assumption~\ref{sparsity}. 

Since the remaining terms in $\Delta\mu(X_{\rm noise},\widetilde{Y},Z)$ are exactly zero when $\beta=0$, we conclude that $$\operatorname{CP}(0)\to\Psi(-z_{1-\alpha})=\alpha,$$ which implies that the test is asymptotically valid. 

To investigate the power when $\beta\neq 0$, we show the second term in $\Delta\mu(X_{\rm noise},\widetilde{Y},Z)$ is non-vanishing while the third term is vanishing. 

\textbf{Analyze the second term:}

Recall from Assumption~\ref{Test statistics} that
$$
\frac{(Y - Z\xi - Z\beta \theta_0)^\T \widetilde{Y}}{\|\widetilde{Y}\|} \geq C_2 \sqrt{n}. 
$$
Since the premise of the theorem states that $\sigma = C_3 n^{-\gamma_{n}}$ and $\beta = h/n^{1/2-\gamma_{n}}$, we have 
$$
\begin{aligned}
\lim_{n\to\infty}\frac{\sigma_\beta^{-1}\sigma^2\beta}{(\sigma^2+1)\nu^2+\beta^2\sigma^2}\frac{\widetilde{Y}^\T (Y-Z\xi-Z\beta\theta_0)}{\|\widetilde{Y}\|}\geq  \lim_{n\to\infty}\frac{ C_3 h / n^{1/2} }{\nu^2} C_2 n^{1/2}
\geq \lim_{n\to\infty}\frac{ C_3 C_2h}{\nu^2} > 0,
\end{aligned}
$$
where we have use $\lim_{n}\sigma_\beta^{-1}\sigma=1$ in the first inequality. 

\textbf{Analyze the third term:}

Again, by the premise of the theorem, we have  $\sigma_\beta^{-1}\beta^2\sigma^2 \asymp h^2 / n^{1-\gamma_{n}} = o(n^{-1/2})$, where the last equation is due to Assumption~\ref{sparsity}. 
By the construction of $X_{\rm noise}$, we have $\|X_{\rm noise}-Z\theta_0\|=O_p(\sqrt{n})$. 
Therefore, we have 
$$
\begin{aligned}
\sigma_\beta^{-1}\beta^2\sigma^2\frac{\widetilde{Y}^\T (X_{\rm noise}-Z\theta_0)}{\|\widetilde{Y}\|}
= o_p(1), 
\end{aligned}
$$
which implies the third term in $\Delta\mu(X_{\rm noise},\widetilde{Y},Z)$ vanishes in probability. 

To sum up, for local alternatives with $\beta= h/n^{1/2-\gamma_{n}}$, both the first and the third terms in $\Delta\mu(X_{\rm noise},\widetilde{Y},Z)$ vanish while the second term is asymptotically lower bounded by a positive constant $ C_3 C_2 h/\nu^2$. We have 
$$
\operatorname{CP}(\beta)\geq \Phi\left(\frac{C_3C_2h}{\nu^2} - z_{1-\alpha} \right), 
$$
holds with probability tending to 1. 

\subsection{Proof of Theorem \ref{theorem:robust}}
\begin{proof}
    Define \begin{align*}
        \Gamma_{\rm SSOSP}&=\left\{(\theta,g)\in \Theta\times \bbR^d:\exists (x,w)\in {\cal X}\times \bbR^d \text{ such that}\right.\\
        &\qquad \qquad \left.{\theta}=\hat{\theta}(x,w) \text{ is an SSOSP, and } g=\hat{g}(x,w)\right\}. 
    \end{align*}

    Recall that $Q_{0}^*$ is the distribution of $(\hat{\theta}(X,W), \hat{g}(X,W))$ under $(X,W)\sim P_0^*$ while $Q_{\theta_0^*}$ is the analogous distribution under $(X,W)\sim P_{\theta_0}^*$. 
    Denote the density functions of $Q_0^*$ and $Q_{\theta_0}^*$ w.r.t. a common measure on $\Gamma_{\rm SSOSP}$ by $q_0(\theta,g)$ and $q_{\theta_0}(\theta,g)$, respectively. 
    
    We bound the expected total variation distribution between the true conditional distribution and the sampling distribution as follows: 
    \begin{align*}
        &\bbE_{Q_{0}^*}\left[d_{\rm TV}(p_{0}(\cdot\mid\hat{\theta},\hat{g}),p_{{\hat{\theta}}}(\cdot\mid\hat{\theta},\hat{g}))\right]\\
        =&\int_{\Gamma_{\rm SSOSP}} q_0(\hat{\theta},\hat{g})\left[d_{\rm TV}(p_{0}(\cdot\mid\hat{\theta},\hat{g}),p_{{\hat{\theta}}}(\cdot\mid\hat{\theta},\hat{g}))\right]\rmd\hat{\theta}\rmd\hat{g}\\
        =&\frac{1}{2}\int_{\Gamma_{\rm SSOSP}}\int_{\cal X}\left|q_0(\hat{\theta},\hat{g})p_{0}(x\mid\hat{\theta},\hat{g})-q_0(\hat{\theta},\hat{g})p_{\hat{\theta}}(x\mid\hat{\theta},\hat{g})\right|\rmd x\rmd\hat{\theta}\rmd\hat{g}\\
        \leq &\frac{1}{2}\int_{\Gamma_{\rm SSOSP}}\int_{\cal X}\left|q_0(\hat{\theta},\hat{g})p_{0}(x\mid\hat{\theta},\hat{g})-q_{\theta_0}(\hat{\theta},\hat{g})p_{\theta_0}(x\mid\hat{\theta},\hat{g})\right|\rmd x\rmd\hat{\theta}\rmd\hat{g}   &\quad&  \text{   (denoted as I)}\\
        &\qquad +\frac{1}{2}\int_{\Gamma_{\rm SSOSP}}\int_{\cal X}\left|q_{\theta_0}(\hat{\theta},\hat{g})p_{\theta_0}(x\mid\hat{\theta},\hat{g})-q_{\theta_0}(\hat{\theta},\hat{g})p_{\hat{\theta}}(x\mid\hat{\theta},\hat{g})\right|\rmd x\rmd\hat{\theta}\rmd\hat{g}   &\quad&  \text{   (denoted as II)}\\
        &\qquad\qquad +\frac{1}{2}\int_{\Gamma_{\rm SSOSP}}\int_{\cal X}\left|q_{\theta_0}(\hat{\theta},\hat{g})p_{\hat{\theta}}(x\mid\hat{\theta},\hat{g})-q_0(\hat{\theta},\hat{g})p_{\hat{\theta}}(x\mid\hat{\theta},\hat{g})\right|\rmd x\rmd\hat{\theta}\rmd\hat{g}    &\quad&  \text{   (denoted as III)}.
    \end{align*}
    For the term I, note that $q_{0}(\hat{\theta},\hat{g})p_{0}(x\mid\hat{\theta},\hat{g})$ is the joint density of $(X,\hat{\theta}(X,W),\hat{g}(X,W))$ when $X,W\sim P_0^*$, and $q_{\theta_0}(\hat{\theta},\hat{g})p_{\theta_0}(x\mid\hat{\theta},\hat{g})$ is the joint density of $(X,\hat{\theta}(X,W),\hat{g}(X,W))$ when $X,W\sim P_{\theta_0}^*$. Therefore, 
    $$\frac{1}{2}\int_{\Gamma_{\rm SSOSP}}\int_{\cal X}\left|q_0(\hat{\theta},\hat{g})p_{0}(x\mid\hat{\theta},\hat{g})-q_{\theta_0}(\hat{\theta},\hat{g})p_{\theta_0}(x\mid\hat{\theta},\hat{g})\right|\rmd x\rmd\hat{\theta}\rmd\hat{g}= {d}_{\rm TV}(P_0^*,P_{\theta_0}^*).$$
    
    For the term II, we have
    \begin{align*}
        &\frac{1}{2}\int_{\Gamma_{\rm SSOSP}}\int_{\cal X}\left|q_{\theta_0}(\hat{\theta},\hat{g})p_{\theta_0}(x\mid\hat{\theta},\hat{g})-q_{\theta_0}(\hat{\theta},\hat{g})p_{\hat{\theta}}(x\mid\hat{\theta},\hat{g})\right|\rmd x\rmd\hat{\theta}\rmd\hat{g}\\
        =\,&\bbE_{Q_{\theta_0}^*}\left[d_{\rm TV}(p_{\theta_0}(\cdot\mid\hat{\theta},\hat{g}),p_{\hat{\theta}}(\cdot\mid\hat{\theta},\hat{g}))\right].
    \end{align*}
    For the term III, we have
    $$\frac{1}{2}\int_{\Gamma_{\rm SSOSP}}\int_{\cal X}\left|q_{\theta_0}(\hat{\theta},\hat{g})p_{\hat{\theta}}(x\mid\hat{\theta},\hat{g})-q_0(\hat{\theta},\hat{g})p_{\hat{\theta}}(x\mid\hat{\theta},\hat{g})\right|\rmd x\rmd\hat{\theta}\rmd\hat{g}=d_{\rm TV}(Q_0^*,Q_{\theta_0}^*),$$
    which is bounded as $d_{\rm TV}(Q_0^*,Q_{\theta_0}^*)\leq d_{\rm TV}(P_0^*,P_{\theta_0}^*)$ since $\hat{\theta}(X,W)$ and $\hat{g}(X,W)$ is determined by $(X,W)$. 
    Combining the three terms, we complete the proof. 
    \end{proof}
\subsection{Proof of Theorem \ref{coupling}}\label{pf:coupling}
\begin{proof}
    Here, we prove the one-sided case and the proof of two-sided test is similar. 
    
    For the approximate validity, we only need to show 
    $$\lim_{n\to\infty}pr(T_0\geq Q_{1-\al}(\widetilde{T})\mid H_0)\leq \alpha + o(1).$$
    
    By Condition \ref{coupling accuracy}, $\widehat{T}$ is exchangeable with $T_0$ under $H_0$. Therefore, we only need to show 
    $$\lim_{n\to\infty}pr(\widehat{T}\geq Q_{1-\al}(\widetilde{T})\mid H_0)\leq \alpha + o(1).$$

    Using the union bound for probabilities, we have 
    $$\begin{aligned}
        &pr(\widehat{T}\geq Q_{1-\al}(\widetilde{T})\mid H_0)\\
        \leq& pr(\widehat{T}\geq Q_{1-\al}(\widetilde{T})+b_n\mid H_0)+ \hat{G}(\widetilde{T}+b_n)-\hat{G}(\widetilde{T})\\
        \leq &pr(\widetilde{T}\geq Q_{1-\al}(\widetilde{T})\mid H_0) + pr(|\widetilde{T}-T|\geq b_n)+ \hat{G}(\widetilde{T}+b_n)-\hat{G}(\widetilde{T})\\ 
\leq &pr(\widetilde{T}\geq Q_{1-\al}(\widetilde{T})\mid H_0) + o(1) + \sup_{t\in \mathbb{R}}\hat{G}(t+b_n)-\hat{G}(t)\\ 
        =&\alpha + o(1), 
    \end{aligned}$$
where the second last inequality is due to Condition \ref{coupling accuracy}, and the last equation is due to Condition \ref{dsitribution of hat W}. 

The argument for power analysis follows similarly by swapping the roles of $\widehat{T}$ and $\widetilde{T}$ in the above derivation.
\end{proof}

\section{Proof of conditional densities}\label{app: conditional density proof}

\subsection{Proof of Lemma \ref{lemma1}}\label{proof:Lemma1}
We begin by introducing some notations as follows: for a regular point $\theta_*\in \Theta$ and $\varepsilon>0$, we have
    $$\begin{aligned}
        { \Theta}_{\theta_*,\varepsilon}&=\left\{\theta\in \Theta: G_i(\theta)\leq 0,\forall i \in [r];G_i(\theta)=0,\forall i \in {\cal I}(\theta_*);\left\|\theta-\theta_*\right\|<\varepsilon\right\},\\
        \Omega_{\theta_*,\varepsilon}&=\left\{(x,w)\in\mathcal{X}\times \bbR^d: \hat{\theta}(x,w)\textrm{ is a}\right.\left.\textrm{SSOSP of (\ref{MLE}), and } \hat{\theta}(x,w)\in \Theta_{\theta_*,\varepsilon}\right\}\\
        \Psi_{\theta_*,\varepsilon}&=\left\{(x,\theta,g)\in\mathcal{X}\times  \Theta_{\theta_*,\varepsilon}\times \bbR^d:\exists\,\right.w\in \bbR^d \textrm{ such that }\\
        &\qquad \qquad\qquad\qquad\qquad\qquad\quad\left.\theta=\hat{\theta}(x,w) \text{ is an SSOSP of (\ref{MLE}), and } g=\hat{g}(x,w)
        \right\}
    \end{aligned}$$
    Moreover, we define 
    $$\begin{aligned}
        { \Theta}_{\theta_*,\infty}&=\left\{\theta\in \Theta: G_i(\theta)\leq 0,\forall i \in [r];G_i(\theta)=0,\forall i \in {\cal I}(\theta)\right\},\\
        \Omega_{\theta_*,\infty}&=\left\{(x,w)\in\mathcal{X}\times \bbR^d: \hat{\theta}(x,w)\textrm{ is a}\right.\left.\textrm{SSOSP of (\ref{MLE}), and } \hat{\theta}(x,w)\in \Theta_{\theta_*,\infty}\right\}\\
        \Psi_{\theta_*,\infty}&=\left\{(x,\theta,g)\in\mathcal{X}\times  \Theta_{\theta_*,\infty}\times \bbR^d:\exists\,\right.w\in \bbR^d \textrm{ such that }\\
        &\qquad \qquad\qquad\qquad\qquad\qquad\quad\left.\theta=\hat{\theta}(x,w) \text{ is an SSOSP of (\ref{MLE}), and } g=\hat{g}(x,w)
        \right\}
    \end{aligned}$$
    \begin{lemma}\label{lemma bijection}
        Define a map $\psi$ from $\Omega_{{\theta_*},\infty}$ to $\Psi_{{\theta_*},\infty}$ as
    $$\psi:(x,w)\to \left(x,\hat{\theta}(x,w),\hat{g}(x,w)\right).$$
For $\forall 0< \varepsilon\leq \infty$, when restricted to $\Omega_{\theta,\varepsilon}$, $\psi$ is a bijection between $\Omega_{{\theta_*},\varepsilon}$ and $\psi_{{\theta_*},\varepsilon}$ with inverse 
    $$\psi^{-1}:(x,\theta,g)\to \Biggl(x,\frac{ g-\nabla_\theta\mathcal{L}(\theta;x)}{\sigma}\Biggr).$$ 
    \end{lemma}
    While we are dealing with general constraints, Lemma \ref{lemma bijection} is essentially parallel to \citet[Lemma 4]{zhu2023approximate}, and the proof is similar and omitted. 

Before the proof of the Lemma \ref{lemma1}, we introduce the definition of minimal active set. 
\begin{definition}\label{minimal}
    For any regular point $\theta \in \Theta_0$, recall that ${\cal I}(\theta)$ is its set of active constraint indices in Definition~\ref{active} and $r(\theta)$ is the constant rank. 
    There exists a subset ${\cal M}(\theta)$ of ${\cal I}(\theta)$, such that $|{\cal M}(\theta)| =r(\theta)$ and $\left(\nabla_\theta G_i(\theta),i\in {\cal M}(\theta)\right)$ spans a space of rank $r(\theta)$. 
    We call ${\cal M}(\theta)$ as the minimal active set at $\theta$.
\end{definition}
The existence of ${\cal M}(\theta)$ is evident by linear algebra, but it may not be unique; to avoid ambiguity, we define ${\cal M}(\theta)$ to be the index set of the pivot-column basis of the matrix $\left[\nabla_\theta G_i(\theta)\right]_{i\in \mathcal{I}(\theta)}$ obtained from the the reduced row-echelon form. 

Proposition~\ref{prop:minimal-set} ensures that around a regular point, the minimal active set is locally determined, whose proof is deferred to Section~\ref{pf:prop:minimal-set}.
\begin{proposition}\label{prop:minimal-set}
Suppose Assumption \ref{assump:constraint} holds. 
For any regular point $\theta_1\in\Theta_0$, there exists $\varepsilon({\theta_1})>0$ such that 
$${\cal B}(\theta_1,\varepsilon({\theta_1}))\cap \left\{\theta: G_i(\theta)=0,\forall i\in {\cal M}(\theta_1)\right\}={\cal B}(\theta_1,\varepsilon({\theta_1}))\cap \left\{\theta: G_i(\theta)=0,\forall i\in {\cal I}(\theta_1)\right\}.$$
Furthermore, for $\forall \theta\in {\cal B}(\theta_1,\varepsilon({\theta_1}))\cap \left\{\theta: G_i(\theta)=0,\forall i\in {\cal M}(\theta_1)\right\}$, the set of active constraint indices at $\theta$ is ${\cal I}(\theta_1)$ and the minimal active set at $\theta$ is ${\cal M}(\theta_1)$. 
\end{proposition}

We now proceed to the proof of Lemma \ref{lemma1}.

    \begin{proof}[Proof of Lemma \ref{lemma1}]
    
Our proof is organized in 4 steps.

\textbf{Step 1: Express the conditional probability. }

    Consider the joint distribution $ (X,W)\sim P_{\theta_0}\times {N}\left(0,{I}_d/d\right)$. By assumption in the lemma, the event $(X,W)\in \Omega_{{\hat{\theta}},\infty}$ has positive probability. 
    Conditioning on the event the event that $\hat{\theta}(X,W)$ is an SSOSP of (\ref{MLE}) with ${\cal I}(\hat{\theta}(X,W))={\cal I}(\theta)$, the joint density of $(X,W)$ is proportional to the function 
    $$h_{\theta_0}(x,w)=f(x;\theta_0)\exp\left\{-\frac{d}{2}\|w\|^2\right\}{1}_{(x,w)\in \Omega_{{\hat{\theta}},\infty}}.$$
    
    By Lemma \ref{lemma bijection}, for any measurable set $I\subseteq \Psi_{{\hat{\theta}},\infty}$, define 
    $$\psi^{-1}(I)=\left\{(x,w)\in \Omega_{{\hat{\theta}},\infty}:\psi(x,w)\in I\right\}$$
    Then, we calculate
    \begin{align*}
            &pr\left\{(X,\hat{\theta}(X,W),\hat{g}(X,W))\in I\mid (X,\hat{\theta}(X,W),\hat{g}(X,W)\in \Psi_{{\hat{\theta}},\infty})\right\}\\
            =&pr\left\{(X,W)\in \psi^{-1}(I)\mid (X,W)\in \Omega_{{\hat{\theta}},\infty}\right\}\\
            = &\frac{ \int_{\psi^{-1}(I)}h_{\theta_0}(x,w)\rmd \nu_{{\cal X}}(x)\rmd w}{ \int_{{\cal X}\times \bbR^d}h_{\theta_0}(x^\prime,w^\prime)\rmd \nu_{\cal X}(x^\prime)\rmd w^\prime}\\
            = &\frac{ \int_{\psi^{-1}(I)}f(x,\theta_0)\exp\left\{-\frac{d}{2}\left\|w\right\|^2\right\}{1}_{(x,w)\in \Omega_{\hat{\theta},\infty}}\rmd \nu_{{\cal X}}(x)\rmd w}{ \int_{{\cal X}\times \bbR^d}h_{\theta_0}(x^\prime,w^\prime)\rmd \nu_{\cal X}(x^\prime)\rmd w^\prime}\\
            = &\frac{ \int_{\psi^{-1}(I)}f(x,\theta_0)\exp\left\{-\frac{d}{2\sigma^2}\left\|\hat{g}(x,w)-\nabla_\theta\mathcal{L}(\hat{\theta}(x,w);x)\right\|^2\right\}{1}_{(x,w)\in \Omega_{\hat{\theta},\infty}}\rmd \nu_{{\cal X}}(x)\rmd w}{ \int_{{\cal X}\times \bbR^d}h_{\theta_0}(x^\prime,w^\prime)\rmd \nu_{\cal X}(x^\prime)\rmd w^\prime}\\
            = &\frac{ \int_{\cal X}f(x,\theta)\int_{\bbR^d}\exp\left\{-\frac{d}{2\sigma^2}\left\|\hat{g}(x,w)-\nabla_\theta\mathcal{L}(\hat{\theta}(x,w);x)\right\|^2\right\}{1}_{(x,w)\in \psi^{-1}(I_{{}})}\rmd w\rmd \nu_{{\cal X}}(x)}{ \int_{{\cal X}\times \bbR^d}h_{\theta_0}(x^\prime,w^\prime)\rmd \nu_{\cal X}(x^\prime)\rmd w^\prime},
    \end{align*}
    where the second last step is due to the definitions of $\hat{g}$ and $\Omega_{\hat{\theta},\infty}$, and the last step holds because $\psi^{-1}(I)\subseteq \Omega_{\hat{\theta},\infty}$.

\textbf{Step 2: Reparameterize $\theta$ and $g$. }
For any given $(x,w)\in \psi^{-1}(I)$ and $(\theta, g)=(\hat{\theta}(x,w), \hat{g}(x,w))$, we find a reparameterization of $\theta$ and $g$. 

Recall the definition of the minimal active set ${\cal M}(\hat{\theta})$ at $\hat{\theta}$ in Definition \ref{minimal}. 
Without loss of generality (WLOG), we can assume that ${\cal M}(\hat{\theta})=[r(\hat{\theta})]$.
Since $\hat{\theta}$ is a regular point of $\Theta$, according to Proposition \ref{prop:minimal-set}, there exists $\varepsilon_1>0$, such that $\forall \theta\in \Theta_{\hat{\theta},\varepsilon_1}$, we have $${\rm rank}\left(\{\nabla_\theta G_i(\theta),1\leq i\leq r(\hat{\theta})\}\right)={r(\hat{\theta})} \text{ and $[r(\hat{\theta})]$ is the minimal active set at } \theta.$$

Let $U_{\hat{\theta}}\in\bbR^{d\times (d-{r(\hat{\theta})})}$ be an orthonormal matrix whose columns form a basis for the orthogonal complement to the space spanned by $\{\nabla_\theta G_i(\hat{\theta}),1\leq i\leq {r(\hat{\theta})} \}$. 
    Define a function $\widetilde{G}$ from $\bbR^d$ to $\bbR^d$ as $ \widetilde{G}(\theta)=\left(
       \begin{array}{cccc}
        \theta^\T  U_{\hat{\theta}}&G_1(\theta)&\ldots&G_{r(\hat{\theta})}(\theta)
       \end{array}
    \right)^\T $. The Jacobi matrix of $\widetilde{G}$ at $\hat{\theta}$ is 
    $$\mathbb{J}\,\widetilde{G}({\hat{\theta}})=\left(
        \begin{array}{cccc}
            U_{\hat{\theta}}&\nabla_\theta G_1(\hat{\theta})&\ldots&\nabla_\theta G_{r(\hat{\theta})}(\hat{\theta})
        \end{array}
    \right), 
    $$
    which is of full rank $d$.
Consequently, we can choose $\varepsilon_2$ (WLOG assumed $\leq \varepsilon_1$) sufficiently small such that there exists a function $H:\bbR^d\to\bbR^d$ such that 
    \begin{equation}\label{eq:H-G-inverse}
    H(\widetilde{G}(\theta))=\theta,\forall \theta\in {\cal B}(\hat{\theta},\varepsilon_2).
    \end{equation}

For any $\theta\in \Theta_{\hat{\theta},\varepsilon_2}$ and $g\in\bbR^d$, if $(x,\theta,g)\in \Psi_{\hat{\theta},\varepsilon_2}$, then by the KKT conditions and the results in Section~\ref{invariance-to-multipliers} (specifically, see \eqref{eq:ssosp-invariance-in-lambda2}), we must have a unique multiplier $\lambda=(\lambda_1,\ldots,\lambda_{r(\hat{\theta})})^\T \in\bbR^{r(\hat{\theta})}$ such that 
    $$g+\sum_{i=1}^{r(\hat{\theta})}\lambda_i\nabla_i G_i(\theta)=0,$$
which implies 
\begin{equation}\label{eq:g-in-lambda-theta}
{g}=-\sum_{i=1}^{r(\hat{\theta})}\lambda_i\nabla_i G_i(\theta).
\end{equation}

We can show that $\theta^\prime := U^\T _{\hat{\theta}}\theta\in \bbR^{d-{r(\hat{\theta})}}$ and the multiplier $\lambda$ from the KKT condition together form a reparameterization of $(\theta,g)$. 
Since $G_i(\theta)=0$ for $i\in [r(\hat{\theta})]$, \eqref{eq:H-G-inverse} implies that 
\begin{equation}\label{eq:constraint-density-theta-in-prime}
\theta=H\left(\left(\begin{array}{c}
        \theta^\prime\\
        0
    \end{array}\right) \right) =: T(\theta^\prime). 
\end{equation}
Furthermore, \eqref{eq:g-in-lambda-theta} implies that $g$ can be expressed as a function of $\lambda$ and $\theta^\prime$, denoted as $\tilde{g}(\lambda, \theta^\prime)$. 

\textbf{Step 3: Compute the Jacobian for changing variables.}

We first introduce define two subsets of $\bbR^{d-r(\hat{\theta})}$: 
\begin{align*}
   { \Theta}_{\hat{\theta},\varepsilon}^\prime&=\left\{ \theta^\prime\in \bbR^{d-r(\hat{\theta})}\mid \theta^\prime=U_{\hat{\theta}}^T\theta\text{ for some }\theta\in { \Theta}_{\hat{\theta},\varepsilon}\right\} ; \\
   { \Theta}_{\hat{\theta},\infty}^\prime&=\left\{ \theta^\prime\in \bbR^{d-r(\hat{\theta})}\mid \theta^\prime=U_{\hat{\theta}}^T\theta\text{ for some }\theta\in { \Theta}_{\hat{\theta},\infty}\right\}.
\end{align*}

    For $(x,w)\in \psi^{-1}(I)$ and $\psi(x,w)=(x,\theta,g)$, we can write 
    $$w=\frac{g-\nabla_\theta {\cal L}(\theta;x)}{\sigma}=-\frac{ \sum_{i=1}^{r(\hat{\theta})}\lambda_i\nabla_\theta G_i(T(\theta^\prime))+\nabla_\theta{\cal L}(T(\theta^\prime);x)}{\sigma}.$$

    Define $\phi_x(\theta^\prime,\lambda):=-\sigma^{-1}\left[\sum_{i=1}^{r(\hat{\theta})}\lambda_i\nabla_\theta G_i(T(\theta^\prime))+\nabla_\theta{\cal L}(T(\theta^\prime);x)\right]$ for $\theta^\prime \in { \Theta}_{\hat{\theta},\varepsilon_2}^\prime$ and $\lambda\in \mathbb{R}_{+}^{r(\hat{\theta})}$. 
    
    Let $\mathbb{J}\,T(\theta^\prime)\in \mathbb{R}^{(d-r(\hat{\theta}))\times d}$ be the Jacobian of $T(\theta^\prime)$ defined in \eqref{eq:constraint-density-theta-in-prime}. 
    A direct calculation yields the following:  
    $$\begin{aligned}
        \nabla_{\theta^\prime}\phi_x(\theta^\prime,\lambda)&=-\frac{1}{\sigma}\left(\mathbb{J}\,T(\theta^\prime)\right) \left[\sum_{i=1}^{r(\hat{\theta})}\lambda_i\nabla_\theta^2 G_i(T(\theta^\prime))+\nabla_\theta^2\mathcal{L}(T(\theta^\prime);x)\right]\\
        \nabla_{\lambda_i}\phi_x(\theta^\prime,\lambda)&=-\frac{1}{\sigma}\left(\nabla_\theta G_i(T(\theta^\prime))\right)^\T , \quad  i \in [r(\hat{\theta})].
    \end{aligned}$$
    
    Therefore,
    \begin{equation}\label{eq:det-grad-phi}
        \begin{aligned}
        \det \left(\nabla \phi_x(\theta^\prime,\lambda)\right)
        &= \det \left(
            \begin{array}{c}
                \nabla_{\theta^\prime}\phi_x(\theta^\prime,\lambda)\\
                \nabla_{\lambda_1}\phi_x(\theta^\prime,\lambda)\\
                \ldots\\
                \nabla_{\lambda_{r(\hat{\theta})}}\phi_x(\theta^\prime,\lambda)
            \end{array}
        \right)\\
        &=\left(-\frac{1}{\sigma}\right)^d\det\left(
            \begin{array}{c}
                 \left(\mathbb{J}\,T(\theta^\prime)\right) \left[\sum_{i=1}^{r(\hat{\theta})}\lambda_i\nabla_\theta^2 G_i(T(\theta^\prime))+\nabla_\theta^2\mathcal{L}(T(\theta^\prime);x)\right]\\
                \left(\nabla_\theta G_1(T(\theta^\prime))\right)^\T \\
                \ldots\\
                \left(\nabla_\theta G_{r(\hat{\theta})}(T(\theta^\prime))\right)^\T 
            \end{array}
        \right).
    \end{aligned}
    \end{equation}
    
For simplicity, we introduce the following shorthand notations: 
\begin{align*}
    G_{1:{r(\hat{\theta})}}(\theta)&=\left(
        \begin{array}{ccc}
            G_1(\theta)&
            \ldots&
            G_{r(\hat{\theta})}(\theta)
        \end{array}
    \right)^\T ,\quad \mathbb{J} G_{1:{r(\hat{\theta})}}(\theta)=\left(
        \begin{array}{ccc}
            \nabla_\theta G_1(\theta)&
            \ldots&
            \nabla_\theta G_{r(\hat{\theta})}(\theta)
        \end{array}
    \right)\\
    C(\theta^\prime,\lambda;x) & =\sum_{i=1}^{r(\hat{\theta})}\lambda_i\nabla_\theta^2 G_i(T(\theta^\prime))+\nabla_\theta^2\mathcal{L}(T(\theta^\prime);x), \\
    {\rm det}_{\theta^\prime,\lambda,x}&=(-\sigma)^d\det \left(\nabla \phi_x(\theta^\prime,\lambda)\right)=
    \det\left( \begin{array}{c}
                 \left(\mathbb{J}\,T(\theta^\prime)\right) C(\theta^\prime,\lambda;x) \\
                \left( \mathbb{J} G_{1:{r(\hat{\theta})}}( T(\theta^\prime) ) \right)^\T 
            \end{array} \right). 
\end{align*}
The identities that 
$$
\left\{
        \begin{array}{l}
            \theta^\prime =U_{\hat{\theta}}^\T \theta =U_{\hat{\theta}}^\T T(\theta^\prime),\\
            {0}=G_{1:{r(\hat{\theta})}}(\theta)=G_{1:{r(\hat{\theta})}}(T(\theta^\prime)),
        \end{array}
    \right.
$$
implies that 
\begin{equation}\label{equation2}
\left\{
        \begin{array}{l}
            \mathbb{J}\,T(\theta^\prime) U_{\hat{\theta}} =\mathbf{I}_{n-{r(\hat{\theta})}},  \\
            \mathbb{J}\,T(\theta^\prime)\mathbb{J}\,G_{1:{r(\hat{\theta})}}(T(\theta^\prime))={0}_{n-{r(\hat{\theta})},{r(\hat{\theta})}}.
        \end{array}
    \right. 
\end{equation}

In particular, consider the case where $\theta= \hat{\theta}$. In this case, we have $\theta^\prime=U_{\hat{\theta}}^\T \hat{\theta}$ and $T(\theta^\prime)=\hat{\theta}$. Furthermore, by the inverse function theorem, we can use elementary calculus to show that  $\mathbb{J}\,T(\theta^\prime)=U_{\hat{\theta}}^\T $.  
We refer to this fact as Fact 1. 
It then follows that the matrix 
$$\begin{aligned}
\left(\begin{array}{cc}
        \left(\mathbb{J}\,T(\theta^\prime)\right)^\T &\mathbb{J}\,G_{1:{r(\hat{\theta})}}(T(\theta^\prime))
    \end{array}\right)
    =\left(\begin{array}{cccc}
        U_{\hat{\theta}}&\nabla_\theta G_1(T(\theta^\prime))&\ldots&\nabla_\theta G_{r(\hat{\theta})}(T(\theta^\prime))
    \end{array}\right)
\end{aligned}$$
is non-singular. 
We can choose $\varepsilon_3$ (WLOG assumed $\leq \varepsilon_2$) small enough such that $\forall \theta^\prime\in \Theta^\prime_{\hat{\theta},\varepsilon_3}$, the matrix$$\left(\begin{array}{cc}
    \left(\mathbb{J}\,T(\theta^\prime)\right)^\T \;&\; \mathbb{J}\,G_{1:{r(\hat{\theta})}}(T(\theta^\prime))
\end{array}\right)$$
is non-singular due to the continuity of the involved Jacobian functions. 
Denoted as the determinant of this matrix as $D(\theta^\prime)$.

We can now compute 
\begin{equation}\label{eq:express-det-theta-lambda-x}
\begin{aligned}
    {\rm det}_{\theta^\prime,\lambda,x}=&\frac{1}{D(\theta^\prime)}\det\left(\left(
        \begin{array}{c}
             \mathbb{J}\,T(\theta^\prime)C(\theta^\prime,\lambda;x)\\
            \left[\mathbb{J}\,G_{1:{r(\hat{\theta})}}(T(\theta^\prime))\right]^\T 
        \end{array}
    \right)\cdot \left(\begin{array}{cc}
        \left(\mathbb{J}\,T(\theta^\prime)\right)^\T \;&\; 
 \mathbb{J}\,G_{1:{r(\hat{\theta})}}(T(\theta^\prime))
    \end{array}\right)\right)\\
        =&\frac{1}{D(\theta^\prime)}\det\left(
            \begin{array}{cc}
                \mathbb{J}\,T(\theta^\prime)C(\theta^\prime,\lambda;x)\left(\mathbb{J}\,T(\theta^\prime)\right)^\T &\mathbb{J}\,T(\theta^\prime)C(\theta^\prime,\lambda;x)\mathbb{J}\,G_{1:{r(\hat{\theta})}}(T(\theta^\prime))\\
                \left[\mathbb{J}\,G_{1:{r(\hat{\theta})}}(T(\theta^\prime))\right]^\T \left(\mathbb{J}\,T(\theta^\prime)\right)^\T &\left[\mathbb{J}\,G_{1:{r(\hat{\theta})}}(T(\theta^\prime))\right]^\T \mathbb{J}\,G_{1:{r(\hat{\theta})}}(T(\theta^\prime))
            \end{array}
        \right)\\
        =&\frac{1}{D(\theta^\prime)}\det\left(
            \begin{array}{cc}
                \mathbb{J}\,T(\theta^\prime)C(\theta^\prime,\lambda;x)\left(\mathbb{J}\,T(\theta^\prime)\right)^\T &\mathbb{J}\,T(\theta^\prime)C(\theta^\prime,\lambda;x)\mathbb{J}\,G_{1:{r(\hat{\theta})}}(T(\theta^\prime))\\
                {0}_{{r(\hat{\theta})},n-{r(\hat{\theta})}}&\left[\mathbb{J}\,G_{1:{r(\hat{\theta})}}(T(\theta^\prime))\right]^\T \mathbb{J}\,G_{1:{r(\hat{\theta})}}(T(\theta^\prime))
            \end{array}
        \right)\\
        =&\frac{1}{ D(\theta^\prime)}\det\left(\mathbb{J}\,T(\theta^\prime)C(\theta^\prime,\lambda;x)\left(\mathbb{J}\,T(\theta^\prime)\right)^\T \right)\det\left(\left[\mathbb{J}\,G_{1:{r(\hat{\theta})}}(T(\theta^\prime))\right]^\T \mathbb{J}\,G_{1:{r(\hat{\theta})}}(T(\theta^\prime))\right),
    \end{aligned}
\end{equation}
where the third equation follows from \eqref{equation2}.

In particular, consider $\theta^\prime=U_{\hat{\theta}}^\T \hat{\theta}$. 
Fact 1 ensures that  $T(\theta^\prime)=\hat{\theta},\mathbb{J}$ and  $T(\theta^\prime)=U_{\hat{\theta}}^\T $. 
Furthermore, since the minimal active set at $\hat{\theta}$ is $[r(\hat{\theta})]$, it holds that 
$$\mathbb{J}\,G_{1:{r(\hat{\theta})}}(T(\theta^\prime))=\left(
    \begin{array}{ccc}
        \nabla_\theta G_1(\hat{\theta})&\ldots&\nabla_\theta G_{r(\hat{\theta})}(\hat{\theta})
    \end{array}
\right)\in \bbR^{d\times {r(\hat{\theta})}}\textrm{ is of full rank }{r(\hat{\theta})}.$$
Therefore, the matrix $$\left(\left[\mathbb{J}\,G_{1:{r(\hat{\theta})}}(T(\theta^\prime))\right]^\T \mathbb{J}\,G_{1:{r(\hat{\theta})}}(T(\theta^\prime))\right)$$ is non-singular. 

The above reasoning allows us to choose $\varepsilon_4$  (WLOG assumed $\leq \varepsilon_3$) sufficiently small such that $\forall \theta^\prime\in \Theta^\prime_{\hat{\theta},\varepsilon_4}$, the matrix $\left(\left[\mathbb{J}\,G_{1:{r(\hat{\theta})}}(T(\theta^\prime))\right]^\T \mathbb{J}\,G_{1:{r(\hat{\theta})}}(T(\theta^\prime))\right)$ is non-singular.

\textbf{Step 4: Change of variables in the inner integral.}

From this point on, we follow a similar argument as that in \cite[Section B.4]{barber2022testing} to verify the validity of applying change-of-variables formula for integration. 
We choose $\varepsilon(\hat{\theta})$ to be $\varepsilon_4$ defined in Step 3. 
For any measurable set $I\subseteq \Psi_{\hat{\theta},\varepsilon(\hat{\theta})}$, we use \eqref{eq:det-grad-phi} to derive 

$$\begin{aligned}
    &\int_{\bbR^d}\exp\left\{-\frac{d}{2\sigma^2}\left\|\hat{g}(x,w)-\nabla_\theta\mathcal{L}(\hat{\theta}(x,w);x)\right\|^2\right\}{1}_{(x,w)\in \psi^{-1}(I_{{}})}\rmd w\\
    =&\sigma^{-d}\int_{\Theta^\prime_{\hat{\theta},\varepsilon(\hat{\theta})}\times \bbR^{r(\hat{\theta})}}\exp\left\{-\frac{d}{2\sigma^2}\left\|\tilde{g}(\lambda,\theta^\prime)-\nabla_\theta\mathcal{L}(T(\theta^\prime))\right\|^2\right\}\cdot \left|{\rm det}_{\theta^\prime,\lambda,x}\right|\cdot {1}_{(x,\phi_x(\theta^\prime,\lambda))\in \psi^{(-1)}(I)}\rmd \lambda\rmd \theta^\prime.
\end{aligned}$$

From the second order condition of SSOSP (Definition~\ref{SSOSP}), we have 
$$\det\left(\left(U_{T(\theta^\prime)}\right)^\T C(\theta^\prime,\lambda;x)U_{T(\theta^\prime)}\right)>0$$
where $U_{T(\theta^\prime)}$ denotes a matrix whose columns forms an orthonormal basis for the subspace orthogonal to span $\{\nabla_\theta G_i(T(\theta^\prime)),1\leq i\leq {r(\hat{\theta})}\}$. 
According to (\ref{equation2}), we have $\mathbb{J}\,T(\theta^\prime)$ is orthogonal to span $\{\nabla_\theta G_i(T(\theta^\prime)),1\leq i\leq {r(\hat{\theta})}\}$ and has rank $d-{r(\hat{\theta})}$. Therefore, there exists a non-singular matrix $S\in\bbR^{(d-{r(\hat{\theta})})\times (d-{r(\hat{\theta})})}$ such that 
\begin{equation}\label{eq:constrain-JT-in-SU}
\mathbb{J}\,T(\theta^\prime)=S U_{T(\theta^\prime)}^\T,
\end{equation}
which implies that 
$$
\det\left(\mathbb{J}\,T(\theta^\prime)C(\theta^\prime,\lambda;x)\left(\mathbb{J}\,T(\theta^\prime)\right)^\T \right)>0.
$$

Recall that in Step 3, we have showed $\mathbb{J}\,G_{1:{r(\hat{\theta})}}(T(\theta^\prime))$ has rank $r(\hat{\theta})$ for $\forall \theta^\prime\in \Theta_{\hat{\theta},\varepsilon(\hat{\theta})}$. 
From \eqref{eq:express-det-theta-lambda-x},  we have 
\begin{equation}\label{eq:constrain-det-in-JT}
|{\rm det}_{\theta^\prime,\lambda,x}|=\frac{1}{ D(\theta^\prime)}\det\left(\mathbb{J}\,T(\theta^\prime)C(\theta^\prime,\lambda;x)\left(\mathbb{J}\,T(\theta^\prime)\right)^\T \right)\det\left(\left[\mathbb{J}\,G_{1:{r(\hat{\theta})}}(T(\theta^\prime))\right]^\T \mathbb{J}\,G_{1:{r(\hat{\theta})}}(T(\theta^\prime))\right)>0 .
\end{equation}
We can also verify ${1}_{x,\phi_x(\theta^\prime,\lambda)\in\psi^{-1}(I)}={1}_{(x,T(\theta^\prime),\tilde{g}(\lambda,\theta^\prime))\in I}$ for any measurable set $I\subseteq \Psi_{\hat{\theta},\varepsilon(\hat{\theta})}\subseteq\Psi_{\hat{\theta},\infty}$. With this calculation in place, we have
$$\begin{aligned}
    &pr\left\{(X,\hat{\theta}(X,W),\hat{g}(\hat{\theta}(X,W);X,W)\in I\mid (X,\hat{\theta}(X,W),\hat{g}(\hat{\theta}(X,W);X,W))\in \Psi_{\hat{\theta},\infty})\right\}\\
    =& \frac{\int_{\cal X}f(x,\theta_0)\int_{\Theta^\prime_{\hat{\theta},\varepsilon(\hat{\theta})}\times \bbR^{r(\hat{\theta})}}\exp^{-\frac{d}{2\sigma^2}\left\|\tilde{g}(\lambda,\theta^\prime)-\nabla_\theta\mathcal{L}(T(\theta^\prime);x)\right\|^2}\cdot |{\rm det}_{\theta^\prime,\lambda,x}|\cdot {1}_{(x,T(\theta^\prime),\tilde{g}(\lambda,\theta^\prime))\in I}\rmd \lambda\rmd \theta^\prime\rmd \nu_{\cal X}(x)}{ \sigma^d\int_{{\cal X}\times \bbR^d}h_{\theta_0}(x^\prime,w^\prime)\rmd \nu_{\cal X}(x^\prime)\rmd w^\prime}.
\end{aligned}$$

In particular, this verifies that 

$$\frac{f(x,\theta_0)\exp^{-\frac{d}{2\sigma^2}\left\|\tilde{g}(\lambda,\theta^\prime)-\nabla_\theta\mathcal{L}(T(\theta^\prime);x)\right\|^2}\cdot |{\rm det}_{\theta^\prime,\lambda,x}|\cdot {1}_{(x,T(\theta^\prime),\tilde{g}(\lambda,\theta^\prime))\in \Psi_{\hat{\theta},\infty}}
}{ \sigma^d\int_{{\cal X}\times \bbR^d}h_{\theta_0}(x^\prime,w^\prime)\rmd \nu_{\cal X}(x^\prime)\rmd w^\prime}$$

is the joint density of $(X,\theta^\prime,\lambda)$ 
when restricted to the region $(X,T(\theta^\prime),\tilde{g}(\lambda,\theta^\prime))\in \Psi_{\hat{\theta},\varepsilon(\hat{\theta})}$, 
conditional on the event $(X,T(\theta^\prime),\tilde{g}(\lambda,\theta^\prime))\in \Psi_{\hat{\theta},\infty}$.
Therefore, the conditional density of $X\mid {\theta}^\prime,\lambda$ (again restricted to the same region and conditioning on the same event) can be written as

$$\propto f(x,\theta_0)\exp\left\{-\frac{d}{2\sigma^2}\left\|\tilde{g}(\lambda,\theta^\prime)-\nabla_\theta\mathcal{L}(T(\theta^\prime))\right\|^2\right\}\cdot |{\rm det}_{\theta^\prime,\lambda,x}|\cdot {1}_{(x,T(\theta^\prime),\tilde{g}(\lambda,\theta^\prime))\in \Psi_{\hat{\theta},\varepsilon(\hat{\theta})}},$$
where we only retain the terms that involve $x$ and view the other terms as constants. 

Recall Fact 1 in Step 3 ensures that when $\theta^\prime=U_{\hat{\theta}}^\T \hat{\theta}$, 
we have $\mathbb{J}\,T(\theta^\prime)=U_{\hat{\theta}}^\T $ and $T(\theta^\prime)=\hat{\theta}$. 
Consequently, \eqref{eq:constrain-det-in-JT} implies that 
$$\begin{aligned}
    |{\rm det}_{\theta^\prime,\lambda,x}|&=\frac{1}{|D(\theta^\prime)|}\det\left(\mathbb{J}\,T(\theta^\prime)C(\theta^\prime,\lambda;x)\left(\mathbb{J}\,T(\theta^\prime)\right)^\T \right)\det\left(\left[\mathbb{J}\,G_{1:{r(\hat{\theta})}}(T(\theta^\prime))\right]^\T \mathbb{J}\,G_{1:{r(\hat{\theta})}}(T(\theta^\prime))\right)\\
    & =\frac{1}{|D(\theta^\prime)|}\det\left( S U_{T(\theta^\prime)}^\T C(\theta^\prime,\lambda;x)\left( S U_{T(\theta^\prime)}^\T \right)^\T \right)\det\left(\left[\mathbb{J}\,G_{1:{r(\hat{\theta})}}(T(\theta^\prime))\right]^\T \mathbb{J}\,G_{1:{r(\hat{\theta})}}(T(\theta^\prime))\right)\\
     & =\frac{1}{|D(\theta^\prime)|}\det(S )^2 \det\left( U_{T(\theta^\prime)}^\T \left[\nabla_\theta^2\mathcal{L}(\hat{\theta};x)+\sum_{i=1}^r\lambda_i \nabla_\theta^2 G_i(\hat{\theta})\right] U_{T(\theta^\prime)} \right)\det\left(\left[\mathbb{J}\,G_{1:{r(\hat{\theta})}}(T(\theta^\prime))\right]^\T \mathbb{J}\,G_{1:{r(\hat{\theta})}}(T(\theta^\prime))\right)\\
    &\propto \det\left(U_{\hat{\theta}}^\T \left[\nabla_\theta^2\mathcal{L}(\hat{\theta};x)+\sum_{i=1}^r\lambda_i \nabla_\theta^2 G_i(\hat{\theta})\right]U_{\hat{\theta}}\right),
\end{aligned}$$
where the second equation is due to \eqref{eq:constrain-JT-in-SU}, the third equation is due to the definition of $C(\theta^\prime,\lambda;x)$ and the invariance of SSOSP w.r.t.  Lagrange multipliers (see Section~\ref{invariance-to-multipliers}), and in the last equation we have dropped the constants that do not involve $x$.

By definition, $(x,\hat{\theta},\hat{g})\in \Psi_{\hat{\theta},\varepsilon(\hat{\theta})}$ if and only if $x\in {\cal X}_{\hat{\theta},\hat{g}}$. Moreover, $(\theta^\prime,\lambda)$ uniquely determines $(\theta,g)$ as discussed before. Therefore, the conditional density of $X\mid \hat{\theta},\hat{g}$ is proportional to 
$$\begin{aligned}
    p_{\theta_0}(\cdot \mid \hat{\theta},\hat{g})\propto f(x;\theta_0)\cdot &\exp\Biggl\{-\frac{d}{2\sigma^2} \left\|\hat{g}-\nabla_\theta \mathcal{L}(\hat{\theta};x)\right\|^2\Biggr\}\\
    &\cdot\det\left(U_{\hat{\theta}}^\T \left[\nabla_\theta^2\mathcal{L}(\hat{\theta};x)+\sum_{i=1}^r\lambda_i \nabla_\theta^2 G_i(\hat{\theta})\right]U_{\hat{\theta}}\right)\cdot {1}_{x\in \mathcal{X}_{\hat{\theta},\hat{g}}}
\end{aligned}$$
\end{proof}
\subsection{Proof of Lemma \ref{lemma penalized}}\label{pf:lemma penalized}

Since ${\cal A}(\theta)$ takes values in the power set of $[d]$, which is finite, we can enumerate each subset $A\subset [d]$ and consider the event $\mathcal{E}_{A}$ where $\hat{\theta}(X,W)$ is an SSOSP of \eqref{grouppenalty} and ${\cal A}(\hat{\theta}(X,W)) = A$. 
Furthermore, we can focus on those subsets such that $\mathcal{E}_{A}$ has positive probability.

Given $A$,  we introduce the following notations:
$$\begin{aligned}
    {\Theta}_{A}&=\left\{\theta\in \Theta:{\cal A}(\theta)=A\right\}, \\
    \Omega_{A}&=\left\{(x,w)\in\mathcal{X}\times \bbR^d: \hat{\theta}(x,w)\textrm{ is an SSOSP of (\ref{grouppenalty}), and } {\cal A}(\hat{\theta}(x,w))=A\right\}, \\
    \Psi_{A}&=\left\{(x,\theta,g)\in\mathcal{X}\times  \Theta_{A}\times \bbR^d:\exists\,w\in \bbR^d \textrm{ such that }\right.\\
    &\qquad \qquad\qquad\qquad\qquad\qquad\quad\left.\theta=\hat{\theta}(x,w) \text{ is an SSOSP of (\ref{grouppenalty}), and } g=\hat{g}(\hat{\theta},x,w)
    \right\}. 
\end{aligned}$$

Suppose $\mathcal{E}_{A}$ happens. 
Our goal is to derive the conditional density of $X\mid \hat{\theta},\hat{g}$. 
We remark that $\mathcal{E}_{A}$ is in the sigma-field generated by $\hat{\theta}(X,W)$, so this conditional density is the same as the conditional density $X\mid \hat{\theta},\hat{g}$ under the probability measure defined by conditioning on $\mathcal{E}_{A}$. 

Consider the joint distribution of $ (X,W)\sim P_{\theta_0}\times N\left(0,I_d/d\right)$. 
Conditioning on the event $\mathcal{E}_{A}$, the joint density of $(X,W)$ is proportional to the function 
$$h_{\theta_0}(x,w)=f(x;\theta_0)\exp\left\{-\frac{d}{2}\|w\|^2\right\}{1}_{(x,w)\in \Omega_{A}}. $$

Define the map $\psi_{A}$ from $\Omega_{A}$ to $\Psi_{A}$ as
$$\psi_{A}:(x,w)\to \left(x,\hat{\theta}(x,w),\hat{g}\left(x,w\right)\right). $$
By a result analogous to Lemma \ref{lemma bijection}, we can see that $\psi_{A}$ is a bijection between $\Omega_{A}$ and $\Psi_{A}$ with inverse 
$$\psi_{A}^{-1}:(x,\theta,g)\to \left(x,\frac{g-\nabla_\theta\mathcal{L}(\theta;x)}{\sigma}\right).$$

For any measurable set $I\subseteq \Psi_{A}$, define 
$$\psi_{A}^{-1}(I)=\left\{(x,w)\in \Omega_{A}:\psi_{A}(x,w)\in I\right\}. $$

Following a similar calculation in Step 1 in Section \ref{proof:Lemma1}, we have
\begin{equation}\label{eq:lemma-conditional-density-penalize}
    \begin{aligned}
    &\text{pr}\left\{(X,\hat{\theta}(X,W),\hat{g}(X,W)\in I\mid (X,\hat{\theta}(X,W),\hat{g}(X,W))\in \Psi_{A})\right\}\\
    = &\frac{ \int_{\cal X}f(x,\theta)\int_{\bbR^d}\exp\left\{-\frac{d}{2\sigma^2}\left\|\hat{g}\left(x,w\right)-\nabla_\theta{\cal L}(\hat{\theta}(x,w);x)\right\|^2\right\}{1}_{(x,w)\in \psi_{A}^{-1}(I)}\rmd w\rmd \nu_{{\cal X}}(x)}{ \int_{{\cal X}\times \bbR^d}h_{\theta_0}(x^\prime,w^\prime)\rmd \nu_{\cal X}(x^\prime)\rmd w^\prime}.
\end{aligned}
\end{equation}

Next, we need to reparameterize $\hat{\theta}$ and $\hat{g}$. 
According to the grouping $G$, we partition the coordinates of the function $g(\theta)$ into $J$ groups as $\left(
    g_1(\theta)\quad\cdots\quad g_J(\theta)
   \right)^\T$. Note that for $\theta$ such that ${\cal A}(\theta)=A$, we have 
   $$\theta_{G_j}={0},\forall j\notin A ,\quad\quad  g_j(\theta)=\rho^\prime_{\lambda}\left(\left\|\theta_{G_j}\right\|\right)\frac{\theta_{G_j}}{\left\|\theta_{G_j}\right\|},\forall j\in A.$$ 
WLOG, by reordering the coordinates, we can assume for some $k\leq J$ and $m\leq d$, it holds that $A=[k]$ and the corresponding index set is $S(A)=[m]$. 
Recall the notation ${I}_{d,j}$ for the $d$-by-$d_{G_j}$ matrix formed by extracting the columns indexed by group $j$ from the $d$-by-$d$ identity matrix. 
Using this notation, we further define 
$${I}_{d,[k]}=\left(\begin{array}{ccc}
    {I}_{d,1}&\cdots&{I}_{d,k}
\end{array}\right)\in\bbR^{d\times m},{I}_{d,[J]\backslash [k]}=\left(\begin{array}{ccc}
    {I}_{d,k+1}&\cdots&{I}_{d,J}
\end{array}\right)\in\bbR^{d\times (d-m)}.$$
These notations enable the following expressions: 
for any $(x,\theta,g)\in \Psi_{A}$, we have
$$
\begin{aligned}
\theta & =\sum_{j=1}^k{I}_{d,j}\theta_{G_j} =: T(\theta_{[m]}),\\
g & =-\sum_{j=1}^k \frac{\rho^\prime_\lambda(\left\|\theta_{G_j}\right\|)}{\left\|\theta_{G_j}\right\|}{I}_{d,j}\theta_{G_j}+\sum_{j=k+1}^J {I}_{d,j} g_j =: \tilde{g}(\theta_{[m]},g_{[d]\backslash [m]}).
\end{aligned}
$$
These expressions show that $\theta^\prime:=\theta_{[m]}$ and $g^\prime:=g_{[d]\backslash [m]}$ parametrized $({\theta},{g})$ provided that $(x,\theta,g)\in \Psi_{A}$. 
Furthermore, since $w=\frac{g-\nabla_\theta\mathcal{L}(\theta;x)}{\sigma}$, we can write it as 
$$w=\frac{1}{\sigma}\left[\tilde{g}(\theta^\prime,g^\prime)-\nabla_\theta {\cal L}\left(T(\theta^\prime);x\right)\right] =: \phi_x(\theta^\prime,g^\prime).$$

Our goal is to derive the Jacobian for changing the variable $w$ to $(\theta^\prime,g^\prime)$. 
It is easy to see that 
$$
 \nabla_{g^\prime}\phi_x(\theta^\prime,g^\prime)=\frac{1}{\sigma}{I}_{d,[J]\backslash [k]}^\T. 
$$

In order to derive $\nabla_{\theta^\prime}\phi_x(\theta^\prime,g^\prime)$, we first derive 
$$
\nabla_{\theta_{G_j}}\left[\nabla_\theta {\cal L}\left(\sum_{j=1}^k{I}_{d,j}\theta_{G_j};x\right)\right]
=
{I}_{d,j}^\T 
    \nabla_\theta^2 \mathcal{L}\left(\sum_{j=1}^k{I}_{d,j}\theta_{G_j} ; x\right), 
    $$
and for any $1\leq j\leq k$, 
$$
\begin{aligned}
\nabla_{\theta_{G_j}}\left[\frac{\rho^\prime_{\lambda}\left(\left\|\theta_{G_j}\right\|\right)}{\left\|\theta_{G_j}\right\|}{I}_{d,j}\theta_{G_j}\right]
&=
\rho_\lambda^\prime \left(\left\|\theta_{G_j}\right\|\right) \left[\frac{{I}_{d_{G_j}} }{\left\|\theta_{G_j}\right\|}-\frac{\theta_{G_j}\left[\theta_{G_j}\right]^\T}{\left\|\theta_{G_j}\right\|^3}\right]{I}_{d,j}^\T+\rho''\left(\left\|\theta_{G_j}\right\|\right)\frac{\theta_{G_j}}{\left\|\theta_{G_j}\right\|}\left[\frac{\theta_{G_j}}{\left\|\theta_{G_j}\right\|}\right]^\T{I}_{d,j}^\T .
\end{aligned}$$
We can express $\nabla_{\theta^\prime}\phi_x(\theta^\prime,g^\prime)$ using the above two equations. 
For simplicity, we write 
$$\begin{aligned}
    G(\theta^\prime;x)=&
    \left[\nabla_\theta^2 \mathcal{L}\left(\sum_{j=1}^k{I}_{d,j}\theta_{G_j} ; x\right)\right]_{[m]}\\
    &+{\rm diag}\left\{\rho_\lambda^\prime\left(\left\|\theta_{G_j}\right\|\right)\left[\frac{{I}_{d_{G_j}}}{\left\|\theta_{G_j}\right\|}-\frac{\theta_{G_j}\left[\theta_{G_j}\right]^\T}{\left\|\theta_{G_j}\right\|^3}\right]+\rho''\left(\left\|\theta_{G_j}\right\|\right)\frac{\theta_{G_j}}{\left\|\theta_{G_j}\right\|}\left[\frac{\theta_{G_j}}{\left\|\theta_{G_j}\right\|}\right]^\T,j\in [k]\right\}.
\end{aligned}$$

Consequently, we have 
$$
    \nabla_{\theta^\prime}\phi_x(\theta^\prime,g^\prime)=-\frac{1}{\sigma}G(\theta^\prime;x){I}_{d,[k]}^\T. 
$$

Note that 
$$\left(\begin{array}{cc}
    {I}_{d,[k]}&{I}_{d,[J]\backslash [k]}
\end{array}\right)={I}_d\Rightarrow \left(\begin{array}{c}
    {I}_{d,[k]}^\T\\{I}_{d,[J]\backslash [k]}^\T
\end{array}\right)\left(\begin{array}{cc}
    {I}_{d,[k]}&{I}_{d,[J]\backslash [k]}
\end{array}\right)={I}_d.$$
Thus, we can obtain ${I}_{d,[k]}^\T{I}_{d,[k]}={I}_{m}$, ${I}_{d,[J]\backslash [k]}^\T{I}_{d,[J]\backslash [k]}={I}_{d-m}$, and ${I}_{d,[k]}^\T{I}_{d,[J]\backslash [k]}={0}$. 
Using these equations, we compute the factor for the change of variable from $w$ to $(\theta^\prime,g^\prime)$ as follows: 
\begin{align*}
        \det \left(\nabla \phi_x(\theta^\prime,g^\prime)\right)&= \det \left(
            \begin{array}{c}
                \nabla_{\theta^\prime}\phi_x(\theta^\prime,{g^\prime})\\
                \nabla_{g^\prime}\phi_x(\theta^\prime,{g^\prime})
            \end{array}
        \right)
    \\
        &=\det \left(\left(
            \begin{array}{c}
                \nabla_{\theta^\prime}\phi_x(\theta^\prime,{g^\prime})\\
                \nabla_{{g^\prime}}\phi_x(\theta^\prime,{g^\prime})
            \end{array}
        \right)\left(\begin{array}{cc}
            {I}_{d,[k]}&{I}_{d,[J]\backslash [k]}
        \end{array}\right)\right)\\
        &=\frac{1}{(-\sigma)^d}\det\left(
            \begin{array}{cc}
                G(\theta^\prime;x){I}_{d,[k]}^\T{I}_{d,[k]}&G(\theta^\prime;x){I}_{d,[k]}^\T{I}_{d,[J]\backslash [k]}\\
                {I}_{d,[J]\backslash [k]}^\T{I}_{d,[k]}& {I}_{d,[J]\backslash [k]}^\T{I}_{d,[J]\backslash [k]}
            \end{array}
        \right)\\
        &=\frac{1}{(-\sigma)^d}\det(G(\theta^\prime;x)). 
\end{align*}

From this point on, we follow a similar argument as that in \cite[Section B.4]{barber2022testing} to verify the validity of applying change-of-variables formula for integration. 
Using $w=\phi_x(\theta^\prime,g^\prime)$, 
we rewrite \eqref{eq:lemma-conditional-density-penalize} as follows: 
$$\begin{aligned}
&\int_{\bbR^d}\exp\left\{-\frac{d}{2\sigma^2}\left\|\hat{g}\left(x,w\right)+\nabla_\theta{\cal L}(\hat{\theta}(x,w);x)\right\|^2\right\}{1}_{(x,w)\in \psi_{A}^{-1}(I)}\rmd w\\
=&\int_{\bbR^{m}\times \bbR^{d-m}}\exp\left\{-\frac{d}{2\sigma^2}\left\|g_{s}(\theta^\prime,g^\prime)+\nabla_\theta{\cal L}(T(\theta^\prime,g^\prime);x)\right\|^2\right\}\\
&\qquad \qquad \qquad\qquad  \cdot \frac{1}{\sigma^d}\det(G(\theta^\prime;x))\cdot {1}_{(x,\phi_x(\theta^\prime,g^\prime))\in \psi^{(-1)}(I)}\rmd g^\prime\rmd \theta^\prime
\end{aligned},$$
where the positiveness of $\det(G(\theta^\prime;x))$ is guaranteed by  the second order condition of Assumption~\ref{definition for penalty}.

We can also verify ${1}_{(x,\phi_x(\theta^\prime,g^\prime))\in \psi^{(-1)}(I)}={1}_{(x,T(\theta^\prime),\tilde{g}(\theta^\prime,g^\prime))\in I}$. With this calculation in place, we then have 

$$\begin{aligned}
&\text{pr}\left\{(X,\hat{\theta}(X,W),\hat{g}(X,W)\in I\mid (X,\hat{\theta}(X,W),\hat{g}(X,W))\in \Psi_{A})\right\}\\
=& \frac{ \int_{\cal X}f(x,\theta_0)\int_{\bbR^m\times\bbR^{d-m}}\exp^{-\frac{d}{2\sigma^2}\left\| \tilde{g}(\theta^\prime,g^\prime) +\nabla_\theta{\cal L}(T(\theta^\prime);x)\right\|^2}\cdot \det(G(\theta^\prime;x))\cdot {1}_{(x,T(\theta^\prime),\tilde{g}(\theta^\prime,g^\prime))\in I}\rmd g^\prime\rmd \theta^\prime\rmd \nu_{\cal X}}{ \sigma^d\int_{{\cal X}\times \bbR^d}h_{\theta_0}(x^\prime,w^\prime)\rmd \nu_{\cal X}(x^\prime)\rmd w^\prime}.
\end{aligned}$$
In particular, we can verify that 
$$\frac{  f(x,\theta_0)\exp\left\{-\frac{d}{2\sigma^2}\left\| \tilde{g}(\theta^\prime,g^\prime)+\nabla_\theta{\cal L}(T(\theta^\prime);x)\right\|^2\right\}\cdot \det(G(\theta^\prime;x))\cdot {1}_{(x,T(\theta^\prime),\tilde{g}(\theta^\prime,g^\prime))\in I}}{ \sigma^d\int_{{\cal X}\times \bbR^d}h_{\theta_0}(x^\prime,w^\prime)\rmd \nu_{\cal X}(x^\prime)\rmd w^\prime}$$
is the joint density of $(X,\theta^\prime,g^\prime)=(X,{I}_{d,[k]}^\T\hat{\theta}(X,W),{I}_{d,[J]\backslash [k]}^\T\hat{g}(X,W))$, 
conditional on the event $\mathcal{E}_{A}$. 
Therefore, the conditional density of $X\mid \theta^\prime,g^\prime$ (again conditioning on this same event) can be written as 

$$\propto f(x;\theta_0)\exp\left\{-\frac{d}{2\sigma^2}\left\| \tilde{g}(\theta^\prime,g^\prime)+\nabla_\theta{\cal L}(T(\theta^\prime);x)\right\|^2\right\}\cdot \det(G(\theta^\prime;x))\cdot {1}_{(x,T(\theta^\prime),\tilde{g}(\theta^\prime,g^\prime))\in I}$$

Moreover, $(\theta^\prime,g^\prime)$ uniquely determines ${\theta}$ and ${g}$ on the event $\mathcal{E}_{A}$, so we can equivalently condition on $(\hat{\theta}(X,W),\hat{g}(X,W))$ and express the conditional density as 
$$f\left(x ; \theta_0\right)\cdot \exp \left\{-\frac{\left\|\hat{g}+\nabla_\theta \mathcal{L}(\hat{\theta} ; x)\right\|^2}{2 \sigma^2 / d}\right\} 
\cdot F_{\rm pen}(\hat{\theta};x)\cdot {1}_{x \in {\mathcal{X}}_{\hat{\theta}, \hat{g}}},$$
where $F_{\rm pen}$ is defined in \eqref{F_pen}.
\subsection{Proof of Lemma \ref{lem:constrain}}\label{pf:lem:constrain}

Recall the function $\hat{g}(x, w)=\nabla_\theta {\cal L}(\hat{\theta}(x,w);x,w)$. 
For $\theta^*\in \bbR^d$ and ${\cal T}\subseteq [d]$, we introduce the following notations:
\begin{align*}
    \Theta_{\theta^*}&=\left\{\theta\in \bbR^d:{\rm sign}(\theta)={\rm sign}(\theta^*),\left\|\theta\right\|_p^p=R_p\right\}, \\
    \Omega_{\theta^*}&=\left\{(x,w)\in\mathcal{X}\times \bbR^d: \hat{\theta}(x,w)\textrm{ is an }\textrm{SSOSP of (\ref{MLE}), and } \hat{\theta}(x,w)\in \Theta_{\theta^*}\right\}, \\
    \Psi_{\theta^*}&=\left\{(x,\theta,g)\in\mathcal{X}\times  \Theta_{\theta^*}\times \bbR^d:\exists\,\right.w\in \bbR^d \textrm{ such that }\\
        &\qquad \qquad\qquad\qquad\qquad\qquad\quad\left.\theta=\hat{\theta}(x,w) \text{ is an SSOSP of (\ref{l_0 problem}), and } g=\hat{g}(x,w)
        \right\},  \\
    \Theta_{\cal T}&=\left\{\theta\in \bbR^d:{\rm supp}(\theta)={\cal T},\left\|\theta\right\|_p^p=R_p\right\},\\
    \Omega_{\cal T}&=\left\{(x,w)\in\mathcal{X}\times \bbR^d: \hat{\theta}(x,w)\textrm{ is a}\textrm{SSOSP of (\ref{MLE}), and } \hat{\theta}(x,w)\in \Theta_{\cal T}\right\}, \\
    \Psi_{\cal T}&=\left\{(x,\theta,g)\in\mathcal{X}\times  \Theta_{\cal T}\times \bbR^d:\exists\,\right.w\in \bbR^d \textrm{ such that }\\
        &\qquad \qquad\qquad\qquad\qquad\qquad\quad\left.\theta=\hat{\theta}(x,w) \text{ is an SSOSP of (\ref{l_0 problem}), and } g=\hat{g}(x,w)
        \right\}.
\end{align*}

By a result analogous to Lemma \ref{lemma bijection}, there exists a bijection between \(\Omega_{\cal T}\) and \(\Psi_{\cal T}\), defined by the mapping \(\psi : (x, w) \mapsto (x, \hat{\theta}(x, w), \hat{g}(x, w))\) and its inverse given by 
\begin{equation}\label{eq:constrain-psi-inv}
\psi^{-1} : (x, \frac{g - \nabla_\theta \mathcal{L}(\theta; x)}{\sigma}).
\end{equation}

First, we consider the case where $p>0$ and $\|\hat{\theta}\|_p^p < R_p$. 
In this case, 
    the boundary constraints would not influence the optimization problem defined in ( \ref{l_0 problem}). Consequently, the problem can be reformulated as: 
    \begin{equation*}
        \hat{\theta}(X, W) = \argmin_{\|\theta\|_p^p < R_p} \mathcal{L}(\theta; X, W) = \argmin_{\|\theta\|_p^p < R_p} -\log f(X; \theta) + \sigma W^\T \theta.
    \end{equation*}
    Since the set \(\left\{\theta\in \bbR^d: \|\theta\|_p^p < R_p \right\}\) is open, we can directly apply  \cite[Lemma 1]{barber2022testing} to derive the conditional density  \(p_{\theta_0}(\cdot \mid \hat{\theta})\), which is the same as the density specified in ( \ref{true density l_0}) since  \(\hat{g} \equiv \mathbf{0}\) and \(\lambda = 0\).

The rest of the proof focuses on the case where   \(\|\hat{\theta}\|_p^p = R_p\). 
    Consider the joint distribution  $ (X,W)\sim P_{\theta_0}\times {N}\left(0,{I}_d/d\right)$. 
    Fix a subset $\mathcal{T}\subseteq [d]$, and let $\mathcal{E}_{\mathcal{T}}$ be the event that $(X,W)\in \Omega_{{\cal T}}$ and $\|\hat{\theta}\|^2=R_{p}$. 
    
    By assumption in the lemma, the event  $(X,W)\in \Omega_{{\cal T}}$ has positive probability.  
    The joint conditional density of  $(X,W)$ given that the event  $\mathcal{E}_{\mathcal{T}}$ happens is proportional to the function  
    $$h_{\theta_0}(x,w)=f(x;\theta_0)\exp\left\{-\frac{d}{2}\|w\|^2\right\}{1}_{(x,w)\in \Omega_{\cal T}}.$$

    Recall that $\psi$ is a bijection between \(\Omega_{\cal T}\) and \(\Psi_{\cal T}\). For any measurable set  $I\subseteq \Psi_{\cal T}$, define 
    $$\psi^{-1}(I)=\left\{(x,w)\in \Omega_{\cal T}:\psi(x,w)\in I\right\}.$$

    Then, we calculate 
    \begin{equation}\label{eq: integration-l_p}
\begin{aligned}
            &\text{pr}\left\{(X,\hat{\theta}(X,W),\hat{g}(X,W))\in I\mid (X,\hat{\theta}(X,W),\hat{g}(X,W)\in \Psi_{\cal T})\right\}\\
            =&\text{pr}\left\{(X,W)\in \psi^{-1}(I)\mid (X,W)\in \Omega_{\cal T}\right\}\\
            = &\frac{ \int_{\psi^{-1}(I)}h_{\theta_0}(x,w)\rmd \nu_{{\cal X}}(x)\rmd w}{ \int_{{\cal X}\times \bbR^d}h_{\theta_0}(x^\prime,w^\prime)\rmd \nu_{\cal X}(x^\prime)\rmd w^\prime}\\
            = &\frac{ \int_{\psi^{-1}(I)}f(x;\theta_0)\exp\left\{-\frac{d}{2}\left\|w\right\|^2\right\}{1}_{(x,w)\in \Omega_{{\cal T}}}\rmd \nu_{{\cal X}}(x)\rmd w}{ \int_{{\cal X}\times \bbR^d}h_{\theta_0}(x^\prime,w^\prime)\rmd \nu_{\cal X}(x^\prime)\rmd w^\prime}\\
            = &\frac{ \int_{\psi^{-1}(I)}f(x;\theta_0)\exp\left\{-\frac{d}{2\sigma^2}\left\|\hat{g}(x,w)-\nabla_\theta\mathcal{L}(\hat{\theta}(x,w);x)\right\|^2\right\}{1}_{(x,w)\in \Omega_{{\cal T}}}\rmd \nu_{{\cal X}}(x)\rmd w}{ \int_{{\cal X}\times \bbR^d}h_{\theta_0}(x^\prime,w^\prime)\rmd \nu_{\cal X}(x^\prime)\rmd w^\prime}\\
            = &\frac{ \int_{\cal X}f(x;\theta_0)\int_{\bbR^d}\exp\left\{-\frac{d}{2\sigma^2}\left\|\hat{g}(x,w)-\nabla_\theta\mathcal{L}(\hat{\theta}(x,w);x)\right\|^2\right\}{1}_{(x,w)\in \psi^{-1}(I)}\rmd w\rmd \nu_{{\cal X}}(x)}{ \int_{{\cal X}\times \bbR^d}h_{\theta_0}(x^\prime,w^\prime)\rmd \nu_{\cal X}(x^\prime)\rmd w^\prime},
\end{aligned}
\end{equation}
    where the last step holds since $\psi^{-1}(I)\subseteq \Omega_{{\cal T}}$.

As in Step 2 in Section~\ref{proof:Lemma1}, for any given $(x,w)\in \psi^{-1}(I)$ and $(\theta, g)=(\hat{\theta}(x,w), \hat{g}(x,w))$, we should find a reparameterization of $\theta$ and $g$. 
    
    WLOG, we can assume \(\mathcal{T} =\mathcal{T}(\hat{\theta})= [k]\) for some integer \(1 \leq k \leq d\) and \(\hat{\theta}_i > 0\) for all \(i = 1, \ldots, k\). 
    For convenience, we define \({I}_{d,k}\) to represent the \(d \times k\) submatrix formed by the first \(k\) columns of the identity matrix \({I}_d\), and \({I}_{d,d-k}\) to denote the \(d \times (d-k)\) submatrix formed by the remaining \(d-k\) columns of \({I}_d\). 

\textbf{Case 1:  $p>0$ and $k>1$. } \newline
For $(x,\theta,g)\in \Psi_{\cal T}$, according to the definitions of SSOSP and $\Psi_{\mathcal{T}}$, we have 
\begin{equation}\label{eq:lp-constrain-theta-g-simplified}
 \begin{aligned}
      \theta&={I}_{d,k}\theta_{[k]}, \\
       g&=-\lambda s(\theta)+{I}_{d,d-k}\cdot g_{[d]\backslash [k]}\\
       &=\begin{bmatrix}
            -\lambda s(\theta_{[k]})  \\
            g_{[d]\backslash [k]}
       \end{bmatrix}. 
 \end{aligned}  
\end{equation}
    where $s(\theta)\in \bbR^{d}$ with $$s(\theta)_i={\rm sign}(\theta_i)\cdot p\cdot |\theta_i|^{p-1}.$$
We emphasize that here $s(\theta)$ is a well-defined function where $s(\theta)_i=0$ if $\theta_i=0$; this should not be confused with the same as the vector $(s_j)$ defined in the KKT condition \eqref{l_p KKT}.

If $k>1$, due to the constraint $\left\|\theta\right\|_p^p = R_p$, we can further reduce $\theta_{[k]}$. 
Specifically, when restricting \((x, \theta, g)\) to \(\Psi_{\hat{\theta}} \subseteq \Psi_{\cal T}\), the constraint \(\left\|\theta\right\|_p^p = R_p\) and the simplifying assumption that $\theta_i>0,\forall 1\leq i\leq k$ together enable the expression of \(\theta_k\) in terms of \(\theta_{[k-1]}\) as:
\[
\theta_{k} = \left(R_p - \sum_{i=1}^{k-1} \theta_i^p\right)^{1/p}.
\]
For convenience, we define 
\begin{equation}\label{eq:constrain-F-def}
F(\theta_{[k-1]})=\left(R_p - \sum_{i=1}^{k-1} \theta_i^p\right)^{1/p}, \text{ for } \|\theta_{[k-1]}\|_{p}<R_p.
\end{equation}
Thus, when \((x, \theta, g) \in \Psi_{\hat{\theta}}\), the parameters \(\theta^\prime=\theta_{[k-1]}\in \mathbb{R}^{k-1}\), \(\lambda\in \mathbb{R}\), and \(g^\prime = \hat{g}_{[d] \backslash [k]}\in \bbR^{d-k}\) form a reparameterization of \(\theta\) and \(g\):
\begin{equation}\label{eq:constrain-theta-g-repar}
\begin{aligned}
    \theta &= {I}_{d,k}\left(\theta^\prime,F(\theta^\prime)\right)^\T =: T(\theta^\prime);\\
     g &=-\lambda s(T(\theta^\prime))+{I}_{d,d-k}\cdot g^\prime 
=: \tilde{g}\left(\theta^\prime,\lambda,g^\prime\right).
\end{aligned}
\end{equation}

Note that $g=\hat{g}(x,w)=\nabla_\theta {\cal L}(\theta;x,w)=\nabla_\theta {\cal L}(\theta;x)+\sigma w$, or equivalently, $w=(g-\nabla_\theta {\cal L}(\theta;x))/\sigma$. 
Combining with \eqref{eq:constrain-theta-g-repar}, we can write
$w=\phi_x(\theta^\prime,\lambda,g^\prime)$, where 
$$\phi_x(\theta^\prime,\lambda,g^\prime)=\frac{1}{\sigma}\left[\tilde{g}\left(\theta^\prime;\lambda;g^\prime\right)-\nabla_\theta {\cal L}(T(\theta^\prime);x)\right].$$

For $F$ defined in \eqref{eq:constrain-F-def}, we can calculate 
$$
\nabla_{\theta^\prime} F =
         -\left(\frac{\theta^\prime}{F(\theta^\prime)}\right)^{p-1},$$
         where $(\theta^\prime)^{p-1}:=( \theta_j^{p-1})_{j\in [k-1]}$ (recall that we have assumed $\theta_j>0$ for $j\leq k$). 
We can then use \eqref{eq:lp-constrain-theta-g-simplified} to derive
\begin{align*}
    \nabla_{\theta^\prime}\phi_x&=-\frac{1}{\sigma}\nabla_{\theta^\prime}T \cdot \left[\lambda p(p-1){\rm diag}\left\{\left(\theta^\prime\right)^{p-2},\left(F(\theta^\prime)\right)^{p-2},{0}\right\}+\nabla^2_\theta {\cal L}(T(\theta^\prime);x)\right]\\
    &=-\frac{1}{\sigma}\left(
        {I}_{k-1},\,- \left(\frac{\theta^\prime}{F(\theta^\prime)}\right)^{p-1}
    \right){I}_{d,k}^\T\left[\lambda p(p-1){\rm diag}\left\{\left(\theta^\prime\right)^{p-2},\left(F(\theta^\prime)\right)^{p-2},{0}\right\}+\nabla^2_\theta {\cal L}(T(\theta^\prime);x)\right],\\
    \nabla_\lambda \phi_x&=-\frac{1}{\sigma}s(T(\theta^\prime)),\\
    \nabla_{g^\prime}\phi_x&=-\frac{1}{\sigma}{I}_{d,d-k}^\T.
\end{align*}

Since ${I}_{d,d-k}^\T {I}_{d,k}=\mathbf{0}_{d-k,k}$, we have
\begin{align*}
    \det\left(\nabla \phi_x(\theta^\prime,\lambda,g^\prime)\right)&=\det\left(\left(\begin{array}{ccc}
        \nabla_{\theta^\prime}\phi_x\\
        \nabla_{\lambda}\phi_x\\
     -\frac{1}{\sigma}{I}_{d,d-k}^\T
    \end{array}\right)\right)\\
    &=\det\left(\left(\begin{array}{ccc}
        \nabla_{\theta^\prime}\phi_x\\
        \nabla_{\lambda}\phi_x\\
     -\frac{1}{\sigma}{I}_{d,d-k}^\T
    \end{array}\right)\left({I}_{d,k},{I}_{d,d-k}\right)\right)\\
    &=\det\left(
        \begin{array}{cc}
            \left(\left(\begin{array}{cc}
                \nabla_{\theta^\prime}\phi_x\\
                \nabla_{\lambda}\phi_x
            \end{array}\right){I}_{d,k}\right)& \left(\left(\begin{array}{cc}
                \nabla_{\theta^\prime}\phi_x\\
                \nabla_{\lambda}\phi_x
            \end{array}\right){I}_{d,d-k}\right)\\
            {0}& -\frac{1}{\sigma}{I}_{d-k}
        \end{array}
    \right)\\
    &=(-1)^{d-k}\sigma^{k-d}\det\left(\left(\begin{array}{cc}
        \nabla_{\theta^\prime}\phi_x\\
        \nabla_{\lambda}\phi_x
    \end{array}\right){I}_{d,k}\right). 
\end{align*}
To proceed with computing the determinant in the above display, we recall $\Lambda(\theta)={\rm diag}\left\{p(p-1)|\theta|^{p-2}\right\}$ defined in Definition~\ref{SSOSP l_0} and derive as follows:

\begin{align*}
    \nabla_{\theta^\prime}\phi_x{I}_{d,k}&=-\frac{1}{\sigma}\left(
        {I}_{k-1},\,- \left(\frac{\theta^\prime}{F(\theta^\prime)}\right)^{p-1}
    \right){I}_{d,k}^\T\left[\lambda p(p-1){\rm diag}\left\{\left(\theta^\prime\right)^{p-2},\left(F(\theta^\prime)\right)^{p-2},{0}\right\}+\nabla^2_\theta {\cal L}(T(\theta^\prime);x)\right]{I}_{d,k}\\
    &=-\frac{1}{\sigma}\left(
        {I}_{k-1},\,- \left(\frac{\theta^\prime}{F(\theta^\prime)}\right)^{p-1}\right)\left[\nabla_\theta^2{\cal L}({T(\theta^\prime)};X)+\lambda \Lambda(T(\theta^\prime))\right]_{[k]}; \\
    \nabla_\lambda \phi_x {I}_{d,k}&=-\frac{1}{\sigma}s(T(\theta^\prime)) {I}_{d,k}\\
    &=-\frac{p}{\sigma}\left( \left[(\theta')^{p-1}\right]^\T , \quad  [F(\theta^\prime)]^{p-1}\right). 
\end{align*}
So
\begin{align*}
    &\det\left(\nabla \phi_x(\theta^\prime,\lambda,g^\prime)\right)\\
    =&(-1)^{d} p \sigma^{-d}\det\left(\begin{array}{c}
        \left(
        {I}_{k-1},\,- \left(\frac{\theta^\prime}{F(\theta^\prime)}\right)^{p-1}
    \right)\left[\nabla_\theta^2\mathcal{L}(\theta ;X)+\lambda \Lambda(\theta)\right]_{[k]}\\
    \left(    \left[(\theta')^{p-1}\right]^\T, \quad  [F(\theta^\prime)]^{p-1}  \right)
    \end{array}\right)\\
    =&(-1)^d p \sigma^{-d}\det\left(\left(\begin{array}{c}
        \left(
        {I}_{k-1},\,- \left(\frac{\theta^\prime}{F(\theta^\prime)}\right)^{p-1}
    \right)\left[\nabla_\theta^2\mathcal{L}(\theta ;X)+\lambda \Lambda(\theta)\right]_{[k]}\\
    \left(    \left[(\theta')^{p-1}\right]^\T, \quad  [F(\theta^\prime)]^{p-1}  \right)
    \end{array}\right)\cdot\left(
        \left(
            {I}_{k-1},\,- \left(\frac{\theta^\prime}{F(\theta^\prime)}\right)^{p-1}\right)^\T , ~~ \left[\begin{array}{cc}
                {0}\\
                1
            \end{array}\right] \right)\right)\\
=&(-1)^d p \sigma^{-d}\det\left(
    \begin{array}{cc}
        \left(
        {I}_{k-1},\,- \left(\frac{\theta^\prime}{F(\theta^\prime)}\right)^{p-1}
    \right)\left[\nabla_\theta^2\mathcal{L}(\theta ;X)+\lambda \Lambda(\theta)\right]_{[k]}\left(
        {I}_{k-1},\,- \left(\frac{\theta^\prime}{F(\theta^\prime)}\right)^{p-1}
    \right)^\T&{\nu}\\
    {0}&\left[F(\theta^\prime)\right]^{p-1}
    \end{array}
\right)\\
=&(-1)^d p \sigma^{-d}\left[F(\theta^\prime)\right]^{p-1}\det\left(
        \left(
        {I}_{k-1},\,- \left(\frac{\theta^\prime}{F(\theta^\prime)}\right)^{p-1}
    \right)\left[\nabla_\theta^2\mathcal{L}(\theta ;X)+\lambda \Lambda(\theta)\right]_{[k]}\left(
        {I}_{d-1},\,- \left(\frac{\theta^\prime}{F(\theta^\prime)}\right)^{p-1}
    \right)^\T\right),
\end{align*}
where we have used the following shorthand notation in the third equation: 
$$\nu=\left(
    {I}_{k-1},\,- \left(\frac{\theta^\prime}{F(\theta^\prime)}\right)^{p-1}
\right)\left[\nabla_\theta^2\mathcal{L}(\theta ;X)+\lambda \Lambda(\theta)\right]_{[k]}\left(\begin{array}{cc}
    {0}\\
    1
\end{array}\right)\in \bbR^{k-1}.$$

\medskip

Denote $$V(\theta^\prime)=\left(
    {I}_{k-1},\,- \left(\frac{\theta^\prime}{F(\theta^\prime)}\right)^{p-1}
\right)^\T  \in \bbR^{k\times (k-1)}.$$
We can express the determinant as 
\begin{align*}
\det\left(\nabla \phi_x(\theta^\prime,\lambda,g^\prime)\right)
= & (-1)^d p \sigma^{-d}\left[F(\theta^\prime)\right]^{p-1}\det\left(
        V(\theta^\prime))^\T\left[\nabla_\theta^2\mathcal{L}(\theta ;X)+\lambda \Lambda(\theta)\right]_{[k]}V(\theta^\prime)\right).
\end{align*}
\medskip 

From this point on, we follow a similar argument as in \cite[Section B.4]{barber2022testing} to verify the validity of applying the change-of-variables formula for the integration in \eqref{eq: integration-l_p}. 
Concretely, we have 
\begin{align*}
    &\int_{\bbR^d}\exp\left\{-\frac{d}{2\sigma^2}\left\|\hat{g}(x,w)-\nabla_\theta\mathcal{L}(\hat{\theta}(x,w);x)\right\|^2\right\}{1}_{(x,w)\in \psi^{-1}(I)}\rmd w\\
    =&\int_{\bbR^{k-1}\times \bbR\times \bbR^{d-k}}\exp^{-\frac{d}{2\sigma^2}\left\|\tilde{g}\left(\theta^\prime;\lambda;g^\prime\right)-\nabla_\theta {\cal L}(T(\theta^\prime);x)\right\|^2}\cdot\det \psi_x(\theta^\prime,\lambda,g^\prime){1}_{(x,\phi_x(\theta^\prime,\lambda,g^\prime))\in \psi^{-1}(I)}\rmd g^\prime \rmd \lambda\rmd \theta^\prime\\
    =& p \sigma^{-d}\int_{\bbR^{k-1}\times \bbR\times \bbR^{d-k}}\exp^{-\frac{d}{2\sigma^2}\left\|\tilde{g}\left(\theta^\prime;\lambda;g^\prime\right)-\nabla_\theta {\cal L}(T(\theta^\prime);x)\right\|^2}  \times \\
    &\qquad \qquad \qquad \det\left(
        V(\theta^\prime)^\T\left[\nabla_\theta^2\mathcal{L}(\theta ;X)+\lambda \Lambda(\theta)\right]_{[k]}V(\theta^\prime)\right){1}_{(x,\phi_x(\theta^\prime,\lambda,g^\prime))\in \psi^{-1}(I)}\rmd g^\prime \rmd \lambda\rmd \theta^\prime.
\end{align*}
From the definition of $Z_{T(\theta^\prime)}$ in (\ref{proj matrix}) and the SSOSP conditions, we can find a non-singular matrix $S\in \bbR^{(k-1)\times (k-1)}$, which is irrelevant with $x$ given $\theta^\prime$, such that 

$$V(\theta^\prime)= Z_{T(\theta^\prime)} S\Rightarrow \det\left(
        V(\theta^\prime)^\T\left[\nabla_\theta^2\mathcal{L}(\theta ;X)+\lambda \Lambda(\theta)\right]_{[k]}V(\theta^\prime)\right)>0.$$
The existence of this non-singular matrix $S$ comes from 
the definition of $Z_{T(\theta^\prime)}$ in (\ref{proj matrix}) and 
the fact that $$V\left({\rm sign}(\theta)\odot \theta^{p-1}\right)_{{\rm supp}(\theta)}=0,\quad \text{and}\quad {\rm rank}\bigl(V(\theta^\prime)\bigr)={\rm rank}\bigl(Z_{T(\theta^\prime)}\bigr)=k-1.$$

We can also verify from our definition that ${1}_{(x,\phi_x(\theta^\prime,\lambda,g^\prime))\in \psi^{-1}(I)}={1}_{(x,T(\theta^\prime),\tilde{g}(\theta^\prime,\lambda,g^\prime)\in I)}$, which leads to 
\begin{align*}
    &\text{pr}\left\{(X,\hat{\theta}(X,W),\hat{g}(X,W)\in I\mid (X,\hat{\theta}(X,W)))\in \Psi_{\cal T}\right\}\\
    =&\frac{ \int_{\cal X}f(x;\theta_0)\int_{\bbR^{k-1}\times \bbR\times \bbR^{d-k}}\exp^{-\frac{d}{2\sigma^2}\left\|\tilde{g}\left(\theta^\prime;\lambda;g^\prime\right)-\nabla_\theta {\cal L}(T(\theta^\prime);x)\right\|^2}\det \psi_x(\theta^\prime,\lambda,g^\prime){1}_{(x,T(\theta^\prime),\tilde{g}(\theta^\prime,\lambda,g^\prime)\in I)}\rmd g^\prime \rmd \lambda\rmd \theta^\prime}{\int_{{\cal X}\times \bbR^d}h_{\theta_0}(x^\prime,w^\prime)\rmd \nu_{\cal X}(x^\prime)\rmd w^\prime}.
\end{align*}
for any measurable set $I\subseteq \Psi(\hat{\theta})$. In particular, this verifies that when restricted to $\Psi_{\hat{\theta}}$
$$\frac{f(x;\theta_0)\exp^{-\frac{d}{2\sigma^2}\left\|\tilde{g}\left(\theta^\prime;\lambda;g^\prime\right)-\nabla_\theta {\cal L}(T(\theta^\prime);x)\right\|^2}\det \psi_x(\theta^\prime,\lambda,g^\prime){1}_{(x,T(\theta^\prime),\tilde{g}(\theta^\prime,\lambda,g^\prime)\in I)}}{\int_{{\cal X}\times \bbR^d}h_{\theta_0}(x^\prime,w^\prime)\rmd \nu_{\cal X}(x^\prime)\rmd w^\prime}$$
is the joint density of $(X,\hat{\theta}(X,W),\hat{g}(X,W))$, conditional on the event $(X,\hat{\theta}(X,W),\hat{g}(X,W))\in \Psi_{\cal T}$. Therefore, the conditional density of $X\mid \theta^\prime,\lambda,g^\prime$ (again restricted to the same region and conditioning on the same event) can be written as 
\begin{align*}
    &\propto f(x;\theta_0)\exp^{-\frac{d}{2\sigma^2}\left\|\tilde{g}\left(\theta^\prime;\lambda;g^\prime\right)-\nabla_\theta {\cal L}(T(\theta^\prime);x)\right\|^2}\det \psi_x(\theta^\prime,\lambda,g^\prime){1}_{(x,T(\theta^\prime),\tilde{g}(\theta^\prime,\lambda,g^\prime)\in I)}\\
    &\propto f(x;\theta_0)\exp^{-\frac{d}{2\sigma^2}\left\|\tilde{g}\left(\theta^\prime;\lambda;g^\prime\right)-\nabla_\theta {\cal L}(T(\theta^\prime);x)\right\|^2}\cdot \det \left[Z_\theta^\T\left(\nabla_\theta^2\mathcal{L}(\theta ;X)+\lambda \Lambda(\theta)\right)_{{\cal T}(\theta)}Z_\theta\right]{1}_{(x,T(\theta^\prime),\tilde{g}(\theta^\prime,\lambda,g^\prime)\in I)}.
\end{align*}
By definition, $(x,\hat{\theta},\hat{g})\in \Psi_{\hat{\theta}}$ if and only if $x\in {\cal X}_{\hat{\theta},\hat{g}}$. Moreover, $({\theta}^\prime,\lambda,{g}^\prime)$ uniquely decides $(\theta,g)$ as described before, so we can equivalently condition on $(\theta,g)$ and obtain 
$$ p_{\theta_0}(\cdot\mid \hat{\theta},\hat{g})\propto f\left(x ; \theta_0\right) \cdot\exp\left\{-\frac{\left\|\hat{g}-\nabla_\theta \mathcal{L}(\hat{\theta};x)\right\|^2}{2\sigma^2/d}\right\}
\cdot \det
   \left(Z_\theta^\T\left[\nabla_\theta^2\mathcal{L}(\theta ;X)+\lambda \Lambda(\theta)\right]_{{\cal T}(\theta)}Z_\theta\right)\cdot {1}_{x \in {\mathcal{X}}_{\hat{\theta}, \hat{g}}}.$$
   
\textbf{Case 2: $p=0$ or $p>0$ while $k=1$.} \newline

This case is considerably simpler than \textbf{Case 1}, and hence we omit detailed derivations and focus on focus on the differences caused by the degeneracies at specific values of $p$ or $k$.

For the case \( p = 0 \), since we always have \( \theta_{[d] \setminus [k]} = 0 \) and \( g_{[k]} = 0 \), the vectors \( \theta_{[k]} \) and \( g_{[d] \setminus [k]} \) serve as a reparameterization of \( \theta \) and \( g \), given by:
\[
\theta = I_{d,k} \theta_{[k]}, \quad g = I_{d,d-k} g_{[d] \setminus [k]}.
\]
The remaining steps proceed analogously to the previous derivations but with significantly simplified calculations. This simplification stems from the absence of constraints on the \( \ell_p\)-norm of \( \theta \); once the support is specified, the \( \ell_0 \)-sparsity condition is automatically satisfied. As a result, the degrees of freedom for \( \theta \) effectively increase by one.

For the case \( p > 0 \) and \( k = 1 \), the parameter \( \theta \) becomes locally fixed, since its support being restricted to
\[
\left\{ \theta \in \mathbb{R}^d : \theta_1 = \pm R_p, \ \theta_j = 0 \text{ for all } j \geq 2 \right\}.
\]
Thus, reparameterization is unnecessary, and the computations simplify considerably compared to \textbf{Case 1}.

\subsection{Proof of Lemmas \ref{lemma conditional density trimmed}, \ref{lemma conditional density quantile} and \ref{condition density GAM}}
To conserve space, detailed proofs of these conditional densities are not presented. Instead, we focus on discussing the differences between these proofs and the previous cases.

The proof of Lemma \ref{lemma conditional density trimmed} closely resembles the case with MLE. The primary difference lies in MTLE's selection process. The conditional density depends on this selection, allowing us to limit our consideration to the selected samples for computations.

The proof of Lemma \ref{lemma conditional density quantile} is straightforward. Here, only \(\hat{g}\) connects the random noise \(W\) with the original sample \(X\), and it can be adequately addressed through a simple change of variables.

The proof of Lemma \ref{condition density GAM} is analogous to that with penalized MLE. There are two main differences: the distribution of \(f\) is not parametrically represented by \(\theta\) but rather by \(\mu_0\), and the distribution of the noise \(W\) also differs. However, these variations minimally impact the derivation of the conditional density; only the corresponding parts need to be adjusted.

\section{Additional proofs}\label{app: add proof}
\subsection{Verifying that (\ref{density2}) defines a density function}
\label{verify}
To ensure that our procedure is well-defined, it is crucial to verify that the plug-in version of the conditional density, defined as \( p_{\hat{\theta}}(\cdot \mid \hat{\theta}, \hat{g}) \propto p^{\rm un}_{\hat{\theta},\hat{g}} \) for some unnormalized function $p^{\rm un}_{\hat{\theta},\hat{g}}$, constitutes a valid density with respect to \( \nu_{\cal X} \). 
We focus on the constrained case, but the arguments and results extend analogously to other cases.

For the constrained case, we have 
\[
p^{\rm un}_{{\theta}, {g}}(x)\stackrel{\triangle}{=}f(x; {\theta}) \cdot \exp\left\{-\frac{\|{g} - \nabla_\theta \mathcal{L}({\theta}; x)\|^2}{2\sigma^2/d}\right\} \cdot \det\left(U_{\theta}^\T\bigl[\nabla_\theta^2\mathcal{L}(\theta ;X)+\sum_{i=1}^r\lambda_i \nabla_\theta^2 G_i(\theta)\bigr]U_{\theta}\right) \cdot {1}_{x \in \mathcal{X}_{{\theta}, {g}}}.
\]
The following lemma confirms that this $p^{\rm un}_{{\theta},g}$ integrates to a finite and positive value. The analogous results for the original aCSS and the regularized aCSS appear in \cite[Section B.3]{barber2022testing} and \cite[Section B.1]{zhu2023approximate}).

\begin{lemma}
    If Assumptions 1 and 3 hold, then for any regular point \(\theta \in \Theta_0\) and any \(g \in \mathbb{R}^d\), the unnormalized density \(p_{\theta,g}^{\text{un}}(x)\) is nonnegative and integrable with respect to \(\nu_{\cal X}\). Furthermore, if the event \(\hat{\theta} = \hat{\theta}(X, W)\) is an SSOSP has positive probability, then conditional on this event, \(\int_X p_{\hat{\theta},\hat{g}}^{\text{un}}(x) dv_x(x) > 0\) holds almost surely.
    \end{lemma}
    
    \begin{proof}
    We first check nonnegativity. For any regular point \(\theta \in \Theta_0\) and any \(x\), we have \(f(x; \theta) > 0\). Furthermore, if \(x \in \mathcal{X}_{\theta, g}\), then \(\det (U_{\theta}^\T \nabla^2_\theta {\cal L}(\theta; x) U_{\theta}) > 0\) by the definition of \(\mathcal{X}_{\theta, g}\) and the SSOSP conditions. This verifies the nonnegativity for \(p_{\theta,g}^{un}(x)\) for any \((\theta, g, x)\). 
    
    Next, we check integrability. 
    As proved in Section \ref{invariance-to-multipliers}, the matrix $$U_\theta^\T\left(\sum_{i=1}^r\lambda_i \nabla_\theta^2 G_i(\theta)\right) U_\theta$$ is the same for all the $\lambda_i,i=1,\cdots,r$ satisfying 
    $$g+\sum_{i=1}^r\lambda_i\nabla_\theta G_i(\theta)=0.$$
    Here we treat $\theta$ and $g$ as fixed, and therefore, $$\lambda_{\rm max}\left(U_\theta^\T\left(\sum_{i=1}^r\lambda_i \nabla_\theta^2 G_i(\theta)\right) U_\theta\right)_+$$ is also fixed, denoted as $\lambda_{\max}(\theta,g)$. Therefore, we have
    \begin{align*}
        &\int_{\cal X} p_{\theta,g}^{\text{un}}(x) \rmd\nu_{\cal X}(x) \\
        \leq &\int_{\cal X} f(x; \theta) \cdot \det\left(U_{\theta}^\T\bigl[\nabla_\theta^2\mathcal{L}(\theta ;X)+\sum_{i=1}^r\lambda_i \nabla_\theta^2 G_i(\theta)\bigr]U_{\theta}\right) \cdot {1}_{U_{\theta}^\T\bigl[\nabla_\theta^2\mathcal{L}(\theta ;X)+\sum_{i=1}^r\lambda_i \nabla_\theta^2 G_i(\theta)\bigr]U_{\theta} \succ 0} \rmd \nu_{\cal X}(x)\\
        \leq &\int_{\cal X} f(x; \theta) \cdot \left[\lambda_{\text{max}}(\nabla^2_\theta {\cal L}(\theta; x))_+ +\lambda_{\max}(\theta,g)\right]^d \cdot \rmd \nu_{\cal X}(x)\\
        \leq & 2^{d-1}\left[\int_{\cal X} f(x; \theta) \cdot \left[\lambda_{\text{max}}(\nabla^2_\theta {\cal L}(\theta; x))_+ \right]^d \rmd \nu_{\cal X}(x)+\int_{\cal X}f(x; \theta) \cdot \lambda^d_{\max}(\theta,g) \rmd \nu_{\cal X}(x)\right],
    \end{align*}
    where the last step holds since $[(a+b)/2]^d\leq (a^d+b^d)/2$ for $d\geq 1$ and $a,b\geq 0$. Since $$\int_{\cal X}f(x; \theta) \cdot \lambda^d_{\max}(\theta,g) \rmd \nu_{\cal X}(x)=\lambda^d_{\max}(\theta,g),$$ 
    we only need to bound
    $$\int_{\cal X} f(x; \theta) \cdot \left[\lambda_{\text{max}}(\nabla^2_\theta {\cal L}(\theta; x))_+ \right]^d \rmd \nu_{\cal X}(x).$$

    Following similar steps in \cite{zhu2023approximate}, we have 
    \begin{align*}
        &\int_{\cal X} f(x; \theta) \cdot \left[\lambda_{\text{max}}(\nabla^2_\theta {\cal L}(\theta; x))_+ \right]^d \rmd \nu_{\cal X}(x)\\
        \leq &\frac{d!}{r(\theta)^{2d}} \int_{\cal X} f(x; \theta) \exp \bigl\{ r(\theta)^2 (\lambda_{\text{max}} (H(\theta, x) - H(\theta)))_+ + r(\theta)^2 (\lambda_{\text{max}} (H(\theta) - \nabla^2_\theta{\cal R}(\theta)))_+
        \bigr\} \rmd \nu_{\cal X}(x)\\
        =& \frac{d!}{r(\theta)^{2d}} \exp \bigl\{ r(\theta)^2 (\lambda_{\text{max}} (H(\theta) - \nabla^2_\theta{\cal R}(\theta)))_+\bigr\}\cdot\bbE_{\theta}\left[\exp \{ r(\theta)^2 (\lambda_{\text{max}} (H(\theta, x) - H(\theta)))_+ \} \right]\\
        \leq &\frac{d!}{r(\theta)^{2d}} \exp \left\{r(\theta)^2 (\lambda_{\text{max}} (H(\theta) - \nabla^2_\theta{\cal R}(\theta)))_++\varepsilon(\theta)\right\} ,
    \end{align*}
    where the first step holds since \(t^d \leq d!e^d\) for any \(t \geq 0\), and the last step is justified by applying Assumption \ref{assumption:3}. This establishes that \(\int_{\mathcal{X}} p_{\theta,g}^{\text{un}}(x) \, \mathrm{d}\nu_{\mathcal{X}}(x)\) is finite.

    Lastly, we show that \(\int_{\mathcal{X}} p_{\theta_0,g}^{\text{un}}(x) \, \mathrm{d}\nu_{\mathcal{X}}(x) > 0\) almost surely. Given any \(x\), the ratio \(\frac{f(x,\theta_0)}{f(x,\hat{\theta})} > 0\) holds. 
    Combined with the proven nonnegativity of \(p_{\hat{\theta},\hat{g}}^{\text{un}}(x)\), it is equivalent to demonstrate that \(\int_{\mathcal{X}} \left[\frac{f(x,\theta_0)}{f(x,\hat{\theta})} p_{\theta}^{\text{un}}(x)  \right] \, \mathrm{d}\nu_{\mathcal{X}}(x) > 0\), which is implied by the fact that \(p_{\theta_0}(\cdot \mid \hat{\theta}, \hat{g}) \propto \frac{f(x,\theta_0)}{f(x,\hat{\theta})} p_{\hat{\theta},\hat{g}}^{\text{un}}(x)\) is the conditional density of \(X \mid \hat{\theta}, \hat{g}\).

    \end{proof}
    
\subsection{Proof of Lemma \ref{lemma_weighted}}\label{pf:lemma_weighted}
\begin{proof}
    The first part of this proof closely aligns with \cite[Appendix A.1.1]{zhu2023approximate}, while the subsequent part draws parallels with \cite[Appendix A.2]{harrison2012conservative}.

We should assume $\al<1$ since the case with \(\alpha = 1\) is obvious due to the non-negativity of the total variation distance. 
In the following we use $\widetilde{Y}^{(j)}$ to denote hypothetical copies in some conceived distributions and we should distinguish them from $\widetilde{X}^{(j)}$ the actual copies sampled in aCSS. 

Consider the following joint distribution (a)
    \[
    \text{Distrib. (a)}\quad\begin{cases}
    (X,W) \sim P^*_{\theta_0}, \\
    \hat{\theta}=\hat{\theta}(X,W),\hat{g}=\hat{g}(X,W)=\nabla {\cal L}(\hat{\theta};X,W), \\
    \widetilde{X}^{(1)}, \dots, \widetilde{X}^{(M)} \mid X, \hat{\theta}, \hat{g} \sim \widetilde{Q}_M(\cdot; X, \hat{\theta}, \hat{g}),
    \end{cases}
    \]
    which is equivalent to the weighted aCSS procedure  conditional on the event $\hat{\theta}(X,W)$ is an SSOSP. 
    Note that if $\hat{\theta}(X,W)$ is not an SSOSP, then by construction, we have $\widetilde{X}^{(1)} = \dots = \widetilde{X}^{(M)} = X$, and Definition \ref{weighted p-values} implies that 
    ${\rm pval}(X,\widetilde{X}^{(1)},\cdots,\widetilde{X}^{(M)})=1$. 
 Therefore, \({\rm pval}_w(X, \widetilde{X}^{(1)}, \ldots, \widetilde{X}^{(M)}) \leq \alpha<1\) only happens when \(\hat{\theta}(X,W)\) is an SSOSP. Consequently, we have
    \[
    \text{pr}\left({\rm pval}_{w}(X,\widetilde{X}^{(1)},\cdots,\widetilde{X}^{(M)})\leq \al\right) \leq \text{pr}\left({\rm pval}_{w}(X,\widetilde{X}^{(1)},\cdots,\widetilde{X}^{(M)})\leq \al\text{ under Distrib. (a) }\right).
    \]
    Define the distribution (b) as 
    \[
        \text{Distrib. (b)}\quad\begin{cases}
    (\hat{\theta}, \hat{g}) \sim Q^*_{\theta_0}, \\
    X \mid \hat{\theta}, \hat{g} \sim p_{\theta_0}(\cdot \mid \hat{\theta}, \hat{g}), \\
    \widetilde{X}^{(1)}, \dots, \widetilde{X}^{(M)} \mid X, \hat{\theta}, \hat{g} \sim \widetilde{Q}_M(\cdot; X, \hat{\theta}, \hat{g}).
    \end{cases}
    \]
     By definition of \(p_{\theta_0}(\cdot \mid \hat{\theta}, \hat{g})\), it is clear that Distrib. (b) is equivalent to Distrib. (a). and then
    \[
    \text{pr}\left({\rm pval}_{w}(X,\widetilde{X}^{(1)},\cdots,\widetilde{X}^{(M)})\leq \al\right) \leq \text{pr}\left({\rm pval}_{w}(X,\widetilde{X}^{(1)},\cdots,\widetilde{X}^{(M)})\leq \al\text{ under Distrib. (b) }\right).
    \]
     Define the distribution (c) as 
     \[
         \text{Distrib. (c)}\quad\begin{cases}
     (\hat{\theta}, \hat{g}) \sim Q^*_{\theta_0}, \\
     X \mid \hat{\theta}, \hat{g} \sim p_{ \hat{\theta} }(\cdot \mid \hat{\theta}, \hat{g}), \\
     \widetilde{X}^{(1)}, \dots, \widetilde{X}^{(M)} \mid X, \hat{\theta}, \hat{g} \sim \widetilde{Q}_M(\cdot; X, \hat{\theta}, \hat{g}).
     \end{cases}
     \]
    By the property of the total variation distance, we have 
    $$\begin{aligned}
            &\text{pr}\left({\rm pval}_{w}(X,\widetilde{Y}^{(1)},\cdots,\widetilde{Y}^{(M)})\leq \al \right)\\
      \leq  \,   &\text{pr}\left({\rm pval}_{w}(X,\widetilde{Y}^{(1)},\cdots,\widetilde{Y}^{(M)})\leq \al\text{ under Distrib. (b) }\right)\\
        \leq \,&\text{pr}\left({\rm pval}_{w}(X,\widetilde{Y}^{(1)},\cdots,\widetilde{Y}^{(M)})\leq \al\text{ under Distrib. (c) }\right)+\bbE_{Q_{\theta_0}^*}\left(d_{\rm TV}(p_{\theta_0}(\cdot\mid\hat{\theta},\hat{g}), ~~   p_{\hat{\theta}}(\cdot\mid\hat{\theta},\hat{g})\right).
    \end{aligned}$$
 Therefore, it suffices to prove 
 $$\text{pr}\left({\rm pval}_{w}(X,\widetilde{Y}^{(1)},\cdots,\widetilde{Y}^{(M)})
        \leq \al\text{ under Distrib. (c) }\right)\leq \al. $$
Define the distribution (d) as 
\[
    \text{Distrib. (d)}\quad\begin{cases}
(\hat{\theta}, \hat{g}) \sim Q^*_{\theta_0}, \\
\widetilde{Y}^{(0)} \mid \hat{\theta}, \hat{g} \sim Q(\cdot \mid \hat{\theta}, \hat{g}), \\
\widetilde{Y}^{(1)}, \dots, \widetilde{Y}^{(M)} \mid \widetilde{Q}_{M}(\widetilde{Y}^{(0)} ; \hat{\theta}, \hat{g}).
\end{cases}
\]
Note that beyond the change in notation, the distribution (d) differs from the distribution (c) in the (conditional) sampling distribution of $\widetilde{Y}^{(0)}$, which is $Q(\cdot \mid \hat{\theta}, \hat{g})$ rather than the distribution with density $p_{\hat{\theta}}(\cdot \mid \hat{\theta}, \hat{g})$. 

By definition of $\widetilde{Q}_{M}(\cdot; X, \hat{\theta}, \hat{g})$ in \eqref{weighted sampling}, we see that under the distribution (d), $(\widetilde{Y}^{(0)}, \widetilde{Y}^{(1)}, \dots, \widetilde{Y}^{(M)})$ is exchangeable (conditional on $(\hat{\theta}, \hat{g})$, which is viewed as fixed in the sampling step).

Recall that \(\mu(\cdot; \hat{\theta}, \hat{g})\) denotes the Radon-Nikodym derivative of $p_{\hat{\theta}}(\cdot \mid \hat{\theta}, \hat{g})$ w.r.t. $Q(\cdot ; \hat{\theta}, \hat{g})$. 
The properties of Radon-Nikodym derivatives imply that for any nonnegative, measurable function $f$ defined on ${\cal X}^{M+1}$, we have 
\begin{equation}\label{eq:RN-c-to-d}
    \begin{aligned}
&\bbE\left(f(X,\widetilde{Y}^{(1)},\cdots,\widetilde{Y}^{(M)})\text{ under Distrib. (c) }\right)\\
        = \,&\bbE\left(\mu(\widetilde{Y}^{(0)} ;\hat{\theta}, \hat{g}) f(\widetilde{Y}^{(0)},\widetilde{Y}^{(1)},\cdots,\widetilde{Y}^{(M)})\text{ under Distrib. (d) }\right).
\end{aligned}
\end{equation}

To proceed, we introduce some notations for analyzing the weighted exchangeability. 
Let ${\cal Q}$ denotes the set of all $(M + 1)!$ permutations $\pi = (\pi_0, \cdots , \pi_M)$ of
$(0, \cdots, {M})$. 
For any $\pi \in \mathcal{Q}$, there is a unique inverse permutation $\pi^{-1} \in \mathcal{Q}$ with $\pi_{\pi_i}^{-1}=\pi_{\pi_i^{-1}}=i$, for each $i=0, \ldots, M$. 
For \( \pi \in {\cal Q} \) and \( z = (z_0, \ldots, z_M) \in \mathbb{R}^{1+M}\), we define \( z^\pi \triangleq (z_{\pi_0}, \ldots, z_{\pi_n}) \). 

Let $Z=\left(X,\widetilde{Y}^{(1)},\cdots,\widetilde{Y}^{(M)}\right)$ be the random vector under the distribution (c) and 
let $Y=\left(\widetilde{Y}^{(0)},\widetilde{Y}^{(1)},\cdots,\widetilde{Y}^{(M)}\right)$ be the random vector under the distribution (d). 
For any fixed $\pi \in {\cal{Q}}$, we use \eqref{eq:RN-c-to-d} to obtain
\begin{equation}\label{eq:Zpi-RN-Y}
\begin{aligned}
\bbE\left(f(Z^{\pi}) \right)
 = \,& \bbE\left(\mu(\widetilde{Y}^{(0)} ;\hat{\theta}, \hat{g}) f(Y^{\pi})  \right)\\
 = \,&  \bbE \left(\mu(\widetilde{Y}^{(\pi_0^{-1})};\hat{\theta},\hat{g})f(Y)   \right),
\end{aligned}
\end{equation}
where the second equation is implied by the exchangeability of  $(\widetilde{Y}^{(0)}, \widetilde{Y}^{(1)}, \dots, \widetilde{Y}^{(M)})$ under the distribution (d). 

In the following, let $\pi$ be uniformly distributed on ${\cal Q}$ and be independent of the other random variables. 
We have $$
\text{pr}(\pi_0^{-1}=k)=\frac{1}{M+1},\qquad \text{ for every integer }k,0\leq k\leq M. 
$$
Let $U=Z^{\pi}$. We use \eqref{eq:Zpi-RN-Y} to obtain
\begin{equation}\label{eq:expect-f(U)}
\begin{aligned}
\bbE(f(U)) & = \sum_{k=0}^{M}\text{pr}(\pi_0^{-1}=k)
\bbE\left(\mu(\widetilde{Y}^{(k)};\hat{\theta},\hat{g})f(Y)\mid \pi_0^{-1}=k\right)\\
& = \frac{1}{M+1}\sum_{k=0}^{M}
\bbE\left(\mu(\widetilde{Y}^{(k)};\hat{\theta},\hat{g})f(Y)\right).
\end{aligned}
\end{equation}

In particular, for a mapping $H:\mathbb{R}^{1+m}\mapsto[0,1]$ to be determined later, 
we consider the following form of the function $f$:
$$f(u):=\frac{(M+1)H(u)\mu(u_0;\hat{\theta},\hat{g})}{\sum_{k=0}^{M}\mu(u_k;\hat{\theta},\hat{g})}, u\in \mathbb{R}^{1+M}.$$
Using \eqref{eq:expect-f(U)}, we have
\begin{equation}\label{eq:expect-H(Z)}
    \begin{aligned}
\bbE(f(U))
&=\frac{1}{M+1}\sum_{k=0}^{M}\bbE\left(\mu(\widetilde{Y}^{(k)};\hat{\theta},\hat{g})f(Y) \right) \\
& = \sum_{k=0}^{M}\bbE\left(\mu(\widetilde{Y}^{(k)};\hat{\theta},\hat{g})H(Y)\mu(\tilde{Y}^{(0)}; \hat{\theta}, \hat{g}) \left(\sum_{k=0}^{M}\mu(\tilde{Y}^{(k)};\hat{\theta},\hat{g})  \right)^{-1} \right) \\
& = \bbE\left( \sum_{k=0}^{M}\mu(\widetilde{Y}^{(k)};\hat{\theta},\hat{g}) \left(\sum_{k=0}^{M}\mu(\tilde{Y}^{(k)};\hat{\theta},\hat{g})  \right)^{-1}  H(Y)\mu(\tilde{Y}^{(0)}; \hat{\theta}, \hat{g}) \right) \\
&= 
\bbE\left(\mu(\widetilde{Y}^{(0)} ;\hat{\theta}, \hat{g}) H(\widetilde{Y}^{(0)},\widetilde{Y}^{(1)},\cdots,\widetilde{Y}^{(M)})\right) \\
& = \bbE\left( H(Z)\right),
\end{aligned}
\end{equation}

where the second equation is by definition of $f$, the thrid equation is by the linearity of expectation, and the last equation is due to \eqref{eq:RN-c-to-d}. 

For $u=(u_0,\cdots,u_{M})$, we use $u^{(k)}$ to denote
denotes the vector by switching the original $u_0$ with $u_k$, i.e, 
$$u^{(k)}=\left(u_{k},u_1,\cdots,u_{k-1},u_0,u_{k+1},\cdots,u_M\right).
$$
Since $\pi$ is uniformly distributed, it holds that $U\stackrel{d.}{=}U^{(k)}$. 
We use \eqref{eq:expect-H(Z)} to derive as follows: 
\begin{equation}\label{eq:H(Z)-in-Uk}
\begin{aligned}
    \bbE(H(Z))=\bbE(f(U)) = & \bbE\left(\frac{(M+1)\mu(U_0;\hat{\theta},\hat{g})}{\sum_{k=0}^{M}\mu(U_k;\hat{\theta},\hat{g})}H(U)\right)\\
    =&\frac{1}{M+1}\sum_{k=0}^{M}\bbE\left(\frac{(M+1)\mu(U_k;\hat{\theta},\hat{g})}{\sum_{i=0}^{M}\mu(U_i;\hat{\theta},\hat{g})}H(U^{(k)})\right)\\
    =&\sum_{k=0}^M\bbE\left( \frac{\mu(U_k;\hat{\theta},\hat{g})}{\sum_{i=0}^{M}\mu(U_{i};\hat{\theta},\hat{g})}H(U^{(k)})\right)
\end{aligned}
\end{equation}
Now we consider a particular $H(\cdot)$ that links the weighted p-value with the hypothesis test. 
Recall that \(T\) is the chosen test statistic in Algorithm~\ref{alg: aCSS implementation}. 
Define the weighted p-value function as 
$$
{\rm pval}(z_0, z_1,\cdots,z_M)=\sum_{i=0}^{M}\omega_i{1}(t_i\geq t_0), 
$$
where  $\omega_i := \mu(z_i;\hat{\theta},\hat{g}) / \sum_{\ell=0}^M\mu(z_\ell;\hat{\theta},\hat{g}) $ and $t_i:=T(z_i)$ for  $i = 0, 1,\ldots, M$.

Given $\alpha\in [0,1]$, we set $H(z) = \mathbf{1}\left[{\rm pval}(z_0, z_1,\cdots,z_M)\leq \al\right]$.  
The following is implied by \eqref{eq:H(Z)-in-Uk}: 
$$
\text{pr}\left( 
{\rm pval}(Z)\leq \al  \right) = \sum_{k=0}^M\bbE\left(  
\omega_k H(U^{(k)})\right)= \bbE\left(  \sum_{k=0}^M
\omega_k H(U^{(k)})\right). 
$$
Following \cite[Lemma 1]{harrison2012conservative}, the integrand inside the expectation of the last display is always not greater than $\alpha$, and the proof is thereby finished.

\end{proof}
\subsection{Proof of Proposition \ref{prop:minimal-set}}\label{pf:prop:minimal-set}
\begin{proof} 
    The second statement in Proposition \ref{prop:minimal-set} follows from the first one due to the continuity of \( G_i \) and the fact that \( \varepsilon(\theta_1) \) can be made arbitrarily small. We only need to prove the first statement.

    Given $\theta_1$, we use ${r_1}$ as a shorthand for $r(\theta_1)$. 
    For simplicity, suppose there exists a integer $k\geq {r_1}$ such that ${\cal I}(\theta_1) = [k]$ and $[{r_1}]$ is the minimal active set at $\theta_1$.
    Let 
    $$\tilde{G}_1(\theta)=\left(\begin{array}{c}
        G_1(\theta)\\
        \cdots\\
        G_{k}(\theta)
    \end{array}\right)\quad\text{ and }\quad\tilde{G}_2=\left(
        \begin{array}{c}
        G_1(\theta)\\
        \cdots\\
        G_{{r_1}}(\theta)
        \end{array}
    \right).$$
To prove the first statement of the proposition, it suffices to prove that there exists $\varepsilon(\theta_1)>0$, such that when $\theta\in {\cal B}(\theta_1,\varepsilon(\theta_1))$, we have 
$$\tilde{G}_1(\theta)={0}\Leftrightarrow \tilde{G}_2(\theta)={0}.$$

Since $k\geq {r_1}$, we only need to prove 
when $\theta\in {\cal B}(\theta_1,\varepsilon(\theta_1))$, the following relationship holds: 
$$
\tilde{G}_2(\theta)={0}\Rightarrow \tilde{G}_1(\theta)={0}.
$$

Since $\theta_1$ is a regular point, by Definition \ref{regular}, there exists $\varepsilon_1(\theta_1)>0$ such that $\nabla_\theta G_i(\theta),i \in {\cal I}(\theta_1)$ remains constant rank ${r_1}$ for all $\theta\in {\cal B}(\theta_1,\varepsilon_1(\theta_1))$. 
Let $A =\nabla_\theta \tilde{G}_1(\theta_1)\in \bbR^{d\times k}$. 
By the constant rank theorem, we know $\tilde{G_1}(\theta)$ is locally determined by its projection onto the range of $A$. 
As a results, there exists $\varepsilon_2(\theta_1)$, such that when $\theta\in {\cal B}(\theta_1,\varepsilon_2(\theta_1))$, we have
$$\tilde{G}_1(\theta)={0}\Leftrightarrow A\tilde{G}_1(\theta)={0}.$$
Let $B=\nabla_\theta \tilde{G}_2(\theta_1)\in \bbR^{d\times {r_1}}$. 
By Definition \ref{minimal}, we know that ${\rm rank}(A)={\rm rank}(B)={r_1}$. 
WLOG, we assume that 
the first ${r_1}$-th rows of $B$ has rank ${r_1}$. 
In other words, $B_{[{r_1}],[{r_1}]}$ is of full rank. 
Let $C=A_{[{r_1}],[k]}$ be the first ${r_1}$-th rows of $A$. 
By the definitions of $\tilde{G}_1$ and $\tilde{G}_2$, we have $C_{[{r_1}],[{r_1}]}=B_{[{r_1}],[{r_1}]}$. Therefore, the rank of $C$ is ${r_1}$, and thus ${\rm rank}(A)={\rm rank}(C)={r_1}$.

Therefore, there exists a matrix $D\in \bbR^{d\times {r_1}}$ 
such that $A=DC$. 
Since all the matrices $A,B,C$, and $D$ are determined by $\theta_1$, we have 
\begin{equation}\label{eq:from-G1-to_G3}
\tilde{G}_1(\theta)={0}\Leftrightarrow A\tilde{G}_1(\theta)={0}\Leftrightarrow C\tilde{G}_1(\theta)={0}.
\end{equation}
Consider the function $\tilde{G}_3(\theta)=C\tilde{G}_1(\theta)$ from $\bbR^d$ to $\bbR^{{r_1}}$. 
Note that
$$\nabla_{\theta}\tilde{G}_3(\theta_1)=\nabla_{\theta}\tilde{G}_1(\theta_1)C^\T =AC^\T =DCC^\T \Rightarrow \nabla_{\theta_{[{r_1}]}}\tilde{G}_3(\theta_1)=D_{[{r_1}]}CC^\T .$$
Note that 
${\rm rank}(A_{[{r_1}]})={\rm rank}(D_{[{r_1}]}C)={r_1}$, so $${r_1}={\rm rank}(D_{[{r_1}]}CC^\T D^\T _{[{r_1}]})\leq {\rm rank}(D_{[{r_1}]}CC^\T )\leq {r_1}.$$
Therefore, ${\rm rank}(\nabla_{\theta_{[{r_1}]}}\tilde{G}_3(\theta_1))={r_1}$. 
Together with the fact that $\tilde{G}_3(\theta_1)=0$, we apply the implicit function theorem to conclude that $\tilde{G}_3(\theta)=0$ locally defines a unique function $\tilde{f}_1$ from $\theta_{[d]\backslash [{r_1}]}\in \bbR^{d-{r_1}}$ to $\theta_{[{r_1}]}\in \bbR^{{r_1}}$, such that $\theta_1\in\{ (\tilde{f}_1(\theta_{[d]\backslash [{r_1}]}), \theta_{[d]\backslash [{r_1}]})  \}$ and it holds locally that $\tilde{G}_3(\tilde{f}_1(\theta_{[d]\backslash [{r_1}]}), \theta_{[d]\backslash [{r_1}]})=0$. In view of \eqref{eq:from-G1-to_G3}, it also holds locally that $\tilde{G}_1(\tilde{f}_1(\theta_{[d]\backslash [{r_1}]}), \theta_{[d]\backslash [{r_1}]})=0$. 

We can make a similar arguments for $\tilde{G}_2$. 
Note that $B_{[{r_1}],[{r_1}]}=\nabla_{\theta_{[{r_1}]}}\tilde{G}_2(\theta_1)$. 
Since $\tilde{G}_2(\theta_1)=0$, by the implicit function theorem, $\tilde{G}_2(\theta)=0$ locally defines another unique function $\tilde{f}_2$ from $\theta_{[d]\backslash [{r_1}]}\in \bbR^{d-{r_1}}$ to $\theta_{[{r_1}]}\in \bbR^{{r_1}}$ so that $\tilde{G}_2(\tilde{f}_2(\theta_{[d]\backslash [{r_1}]}), \theta_{[d]\backslash [{r_1}]})=0$ holds locally near $\theta_1$. 

Since $\tilde{G}_1=0\Rightarrow \tilde{G}_2=0$, it is clear that $\tilde{f}_1\equiv \tilde{f}_2$ due to the uniqueness of $\tilde{f}_1$ and that of $\tilde{f}_2$. 
By choosing $\varepsilon(\theta_1)$ sufficiently small, we conclude that for $\theta\in {\cal B}(\theta_1,\varepsilon(\theta_1))$, if $\tilde{G}_2(\theta)={0}$ then $\tilde{G}_1(\theta)={0}$. 
Therefore, we complete the proof.
\end{proof}

\subsection{Invariance of SSOSP with Respect to Lagrange multipliers}\label{invariance-to-multipliers}

In this section, we will verify the invariance of the strict second-order stationary point (SSOSP) definition with respect to the choice of the Lagrange multipliers $\lambda_i$. Suppose we have a Lagrange multiplier vector $\lambda^{(1)} = (\lambda^{(1)}_1, \ldots, \lambda^{(1)}_r)^\T \in \mathbb{R}^r$ satisfying

\begin{equation}\label{eq:ssosp-invariance-in-lambda1}
\nabla_\theta {\cal L}(\theta;X,W)+\sum_{i=1}^{r}\lambda_i^{(1)}\nabla_\theta G_i(\theta)=0\quad {and}\quad \lambda_i^{(1)}\geq 0,\forall 1\leq i\leq r,\lambda_i^{(1)}=0,\forall i\in [r]\backslash {\cal I}(\theta).
\end{equation}

For simplicity of writing, we assume that ${\cal M}(\theta)=[r(\theta)]=\left\{1,\cdots,r(\theta)\right\}$. According to Definition \ref{minimal} and Proposition \ref{prop:minimal-set}, there exists $\varepsilon_0>0$, such that $\forall \theta^\prime\in{\cal B}(\theta,\varepsilon_0)$, both matrices 
$$\left(\nabla_\theta G_i(\theta^\prime),i\in {\cal I}(\theta)\right)\quad \text{and}\quad \left(\nabla_\theta G_i(\theta^\prime),i\in [r(\theta)]\right)$$
are of rank $r(\theta)\leq r$. 
Therefore, we have the following: 
\begin{enumerate}
    \item We can find a unique vector $\lambda^{(2)} = (\lambda^{(2)}_1, \ldots, \lambda^{(2)}_r)\in \bbR^r$ such that 
\begin{equation}\label{eq:ssosp-invariance-in-lambda2}
\nabla_\theta {\cal L}(\theta;X,W)+\sum_{i=1}^r\lambda_i^{(2)}\nabla_\theta G_i(\theta)=0  \quad \text{and}\quad \lambda_i^{(2)}\geq 0,\forall 1\leq i\leq r,\lambda_i^{(2)}=0,\forall r(\theta)<i\leq r.
\end{equation}
\item For any $i\in {\cal I}(\theta)$ and any $\theta^\prime\in {\cal B}(\theta,\varepsilon_0)$, we can find $\Lambda_{i,j}(\theta^\prime)\in\bbR$ for $1\leq j\leq r(\theta)$ such that 
\begin{equation}\label{eq:G_I-to-G_r(theta)}
\nabla_\theta G_i(\theta^\prime)=\sum_{j=1}^{r(\theta)}\Lambda_{i,j}(\theta^\prime)\nabla_\theta G_j(\theta^\prime).
\end{equation}
\end{enumerate}

We can prove the following two conclusions: 
\begin{enumerate}
    \item[]\textbf{Conclusion 1.} For any $j\leq r(\theta)$, $\lambda^{(2)}_j=\sum_{i\in {\cal I}(\theta)}\lambda_i^{(1)}\Lambda_{i,j}(\theta)$. 
    \item[]\textbf{Conclusion 2.} The functions $\Lambda_{i,j}(\theta)\in\bbR,1\leq i\leq r,1\leq j\leq r(\theta)$ are differentiable at $\theta$. 
\end{enumerate}
\textbf{Proof of Conclusion 1: }

\eqref{eq:ssosp-invariance-in-lambda1} and \eqref{eq:ssosp-invariance-in-lambda2} together shows that 
$$
\sum_{i=1}^{r}\lambda_i^{(1)}\nabla_\theta G_i(\theta)=\sum_{j=1}^{r(\theta)}\lambda_j^{(2)}\nabla_\theta G_j(\theta)=0. 
$$
Taking $\theta^\prime=\theta$ in \eqref{eq:G_I-to-G_r(theta)} and multiplied with $\lambda_i^{(1)}$ for all $i\in {\cal I}(\theta)$, we have 
$$
\begin{aligned}
\sum_{j=1}^{r(\theta)}\lambda_j^{(2)}\nabla_\theta G_j(\theta) & =\sum_{i=1}^{r}\lambda_i^{(1)}\nabla_\theta G_i(\theta)\\
& = \sum_{i=1}^{r}\lambda_i^{(1)}\sum_{j=1}^{r(\theta)}\Lambda_{i,j}(\theta)\nabla_\theta G_j(\theta)\\
& = \sum_{j=1}^{r(\theta)} \left(\sum_{i=1}^{r}\lambda_i^{(1)}\Lambda_{i,j}(\theta) \right) \nabla_\theta G_j(\theta)
\end{aligned}
$$
Since $\left(\nabla_\theta G_i(\theta^\prime),i\in [r(\theta)]\right)$ has rank $r(\theta)$, we conclude that for any $j\leq r(\theta)$, it holds that 
\begin{equation}\label{eq:ssosp-invariance-lambda2-in-lambda1}
\lambda^{(2)}_j=\sum_{i\in {\cal I}(\theta)}\lambda_i^{(1)}\Lambda_{i,j}(\theta).
\end{equation}
\qedsymbol

\textbf{Proof of Conclusion 2: }
We will make use of the smoothness of $G_i(\theta)$ and the implicit function theorem. 

For $\theta^\prime\in \mathcal{B}(\theta, \varepsilon_0)$, define a matrix function $A(\theta^\prime)=\left[ \nabla_\theta G_1(\theta^\prime), \ldots, \nabla_\theta G_{r(\theta)}(\theta^\prime) \right]^\T \in \mathbb{R}^{r(\theta)\times d}$. Note that the rank of the matrix $A(\theta^\prime)$ is $r(\theta)$. 

For any $i\in {\cal I}(\theta)$, define $F_i(\theta^\prime, c)=A(\theta^\prime)\left(\sum_{j=1}^{r(\theta)}c_i\nabla_\theta G_j(\theta^\prime)-\nabla_\theta G_i(\theta^\prime)\right)$ for $\theta^\prime\in \mathcal{B}(\theta, \varepsilon_0)$ and $c\in \mathbb{R}^{r(\theta)}$. 

Note that the Jacobian matrix of $F_i$ w.r.t. $c$ is given by
$$
\nabla_{c}F_i=\left[ \nabla_\theta G_1(\theta^\prime), \ldots, \nabla_\theta G_{r(\theta)}(\theta^\prime) \right]^\T A(\theta^\prime)^\T = A(\theta^\prime)A(\theta^\prime)^\T,  
$$
which is non-singular since the rank is $r(\theta)$. 

Since $G_i(\theta)$ are twice continuously differentiable, by the implicit function theorem, we conclude that the vector function $c(\theta^\prime)=(\Lambda_{i,1}(\theta^\prime), \ldots, \Lambda_{i,r(\theta)}(\theta^\prime))^\T$ is the unique solution to $F_i(\theta^\prime, c)=0$ and is continuously differentiable. 
\qedsymbol

By the product rule of differentiation,  we have 
$$\nabla_\theta^2 G_i(\theta)=\sum_{j=1}^{r(\theta)}\nabla^2_\theta G_j(\theta)\Lambda_{i,j}(\theta)+\sum_{j=1}^{r(\theta)}\nabla_\theta G_j(\theta)\left[\nabla_\theta\Lambda_{i,j}(\theta)\right]^\T,\forall 1\leq i\leq r.$$
From the definition of $U_{\theta}^\T$ in (\ref{proj}), we have $U_{\theta}^\T\nabla_\theta G_j(\theta)={0},\forall 1\leq j\leq r(\theta)$. Therefore, we can obtain
$$\begin{aligned}
    &U_{\theta}^\T\left[\sum_{i\in{\cal I}(\theta)}\lambda_i^{(1)} \nabla_\theta^2 G_i(\theta)\right]U_{\theta}\\=
    &U_{\theta}^\T\left[\sum_{j=1}^{r(\theta)}\nabla^2_\theta G_j(\theta)\left(\sum_{i\in {\cal I}(\theta)}\lambda_i^{(1)}\Lambda_{i,j}(\theta)\right)+\sum_{j=1}^{r(\theta)}\nabla_\theta G_j(\theta)\left(\sum_{i\in {\cal I}(\theta)}\lambda_i^{(1)}\left[\nabla_\theta\Lambda_{i,j}(\theta)\right]^\T\right)\right]U_{\theta}\\
    =&U_{\theta}^\T\left[\sum_{j=1}^{r(\theta)}\lambda_j^{(2)}\nabla_\theta^2G_j(\theta)\right]U_{\theta},
\end{aligned}$$
where the last equation is due to \eqref{eq:ssosp-invariance-lambda2-in-lambda1}. 
In other words, we have shown
$$U_{\theta}^\T\left[\sum_{i\in{\cal I}(\theta)}\lambda_i^{(1)} \nabla_\theta^2 G_i(\theta)\right]U_{\theta}=U_{\theta}^\T\left[\sum_{j=1}^{r(\theta)}\lambda_j^{(2)} \nabla_\theta^2 G_j(\theta)\right]U_{\theta},$$
where $(\lambda_j^{(2)} : j = 1, \dots, r(\theta))$ is uniquely determined, and thus the proof is complete.

\subsection{Proof of Lemma \ref{program lemma}}\label{pf:program lemma}
\begin{proof}
    Denote the feasible region of the program \eqref{program} as 
    $$C=\left\{x\in \bbR^n:\left\|Z^\T x\right\|_\infty\leq \lambda_n,\left\|x\right\|\leq 1\right\}$$
    Use $\widehat{Y}$ to denote the unique maximizer of the program \eqref{program}. 
    Using the orthogonal decomposition w.r.t. the linear span of $Y$, we can write 
    $$e_1:=\frac{Y-Z(\xi+\beta\theta_0)}{\left\|Y-Z(\xi+\beta\theta_0)\right\|}=k_1Y+k_2e_2,$$
    where $e_2$ be a unit vector such that $Y^\T e_2=0$, and $$k_1:=\frac{Y^\T (Y-Z(\xi+\beta\theta_0))}{\left\|Y\right\|^2\left\|Y-Z(\xi+\beta\theta_0)\right\|}$$
    By the assumption about $\lambda_n$, we have $e_1\in C$. Since $\widehat{Y}$ is the maximizer of the program \eqref{program}, 
    we have $Y^\T \widehat{Y}\geq {Y}^\T e_1$, which leads to 
    \begin{equation}\label{eq:program-Yhat-norm-large}
    \|\widehat{Y}\|\geq \frac{\widehat{Y}^TY}{\|Y\|}\geq \frac{e_1^TY}{\|Y\|}=\frac{Y^T(Y-Z(\xi+\beta\theta_0)}{\left\|Y\right\|\left\|Y-Z(\xi+\beta\theta_0\right\|}
    \end{equation}
    and
    \begin{equation}\label{eq:program-Yhat-noise}
    \begin{aligned}
        \widehat{Y}^\T (Y-Z\xi-Z\beta\theta_0)
        &=\left\|Y-Z(\xi+\beta\theta_0)\right\|\widehat{Y}^\T e_1\\
        &=\left\|Y-Z(\xi+\beta\theta_0)\right\|\left[\widehat{Y}^\T \left(k_1 Y + k_2 e_2\right) \right]\\
        &\geq \left\|Y-Z(\xi+\beta\theta_0)\right\|\left[k_1Y^\T e_1+k_2\widehat{Y}^\T e_2\right]\\
        &=\frac{\left[Y^\T (Y-Z(\xi+\beta\theta_0))\right]^2}{\left\|Y\right\|^2\left\|Y-Z(\xi+\beta\theta_0)\right\|}+\left\|Y-Z(\xi+\beta\theta_0)\right\|k_2\widehat{Y}^\T e_2,
    \end{aligned}
    \end{equation}
    where the last equation is due to the definitions of $e_1$ and $k_1$.

    We only need to prove $k_2\widehat{Y}^\T e_2\geq 0$. Consider the decomposition of $\widehat{Y}$ as follows:
    $$\widehat{Y}=\lambda_1Y+\lambda_2e_2+\lambda_3 e_3,$$
    where $Y^\T e_3=0,e_2^\T e_3=0$.

    We only need to show $\lambda_2k_2\geq 0$. 
    Otherwise, if $\lambda_2k_2< 0$, for some $t\in (0,1)$ to be determined, define 
    $$\widehat{Y}^\prime=t\widehat{Y}+(1-t)\frac{\lambda_1}{k_1}e_1=\lambda_1  Y+\left[t\lambda_2+(1-t)\frac{\lambda_1}{k_1}k_2\right]e_2+t\lambda_3e_3.$$
    Since $\widehat{Y}^\prime-\widehat{Y}\in \text{span}(e_2,e_3)$, we have $Y^\T \widehat{Y}=Y^\T \widehat{Y}^\prime$. 
    We will show below that $\widehat{Y}^\prime\in C$ and $\widehat{Y}^\prime\neq \widehat{Y}$, and thus leads to a contradiction to the uniqueness of the maximizer of \eqref{program}. 

Let's first show $\widehat{Y}^\prime\neq \widehat{Y}$. 
Using \eqref{eq:program-Yhat-norm-large} and \eqref{eq:program-Yhat-noise}, we have 
    $$\lambda_1=\frac{\widehat{Y}^TY}{\left\|Y\right\|_2^2}\geq \frac{e_1^TY}{\left\|Y\right\|_2^2}=k_1>0,$$
which implies $\lambda_1/k_1>0$. 
When $t\in (0,1)$, since $\lambda_2k_2<0$, we can choose $t$ sufficiently close to 1 such that $$|t\lambda_2+(1-t)\frac{\lambda_1}{k_1}k_2| < |\lambda_2|,$$ which leads to $\|\widehat{Y}^\prime\|< \|\widehat{Y}\|\leq 1$ by comparing their coefficients of $Y$, $e_2$, and $e_3$, respectively. 
    
Next, we show $\|Z^\T \widehat{Y}^\prime\|\leq \lambda_n$. 
The constraints of the program~\eqref{program} implies that $\left\|Z^\T\widehat{Y}\right\|_\infty\leq \lambda_n$. 
Furthermore, using the definitions of $e_1$ and $k_1$ and the equation that $\lambda_1=\widehat{Y}^T Y / \left\|Y\right\|_2^2$, we have  
$$
\begin{aligned}
& \left\|Z^\T \frac{\lambda_1}{k_1}e_1\right\|_\infty\\
= & \left\|
\frac{\widehat{Y}^TY}{\left\|Y\right\|_2^2} 
\frac{\|Y\|^2 \left\|Y-Z(\xi+\beta\theta_0)\right\|}{Y^\T\left(Y-Z(\xi+\beta\theta_0)\right)}
\frac{Z^\T  \left(Y-Z(\xi+\beta\theta_0) \right)}{\left\|Y-Z(\xi+\beta\theta_0)\right\|}
\right\|_\infty\\
= &\frac{\widehat{Y}^TY}{Y^\T\left(Y-Z(\xi+\beta\theta_0)\right)}  \left\|
Z^\T  \left(Y-Z(\xi+\beta\theta_0) \right)
\right\|_\infty\\
\leq & \frac{\|Y\|}{Y^\T (Y-Z(\xi+\beta\theta_0))}\|Z^\T (Y-Z(\xi+\beta\theta_0))\|_\infty \\
\leq & \lambda_n,
\end{aligned}
$$
where the first inequality is because $\|\widehat{Y}^T\|\leq $ and $\widehat{Y}^TY\leq \|\widehat{Y}^T\|\|Y\|\leq \|Y\|$, and the second inequality is due to the assumption of the lemma. 

Since $\|\cdot\|$ is convex and $\widehat{Y}^\prime$ is a convex combination of $\hat{Y}\in C$ and $\frac{\lambda_1}{k_1}e_1 \in C$, we have $\|Z^\T \widehat{Y}^\prime\|\leq \lambda_n$. This completes the proof. 
\end{proof}

\subsection{Proof of Lemma \ref{quantile lemma}}\label{pf:quantile lemma}
\begin{proof}
    Recall that the two vectors $\widetilde{X}$ and $X\in \bbR^n$ satisfy that $s_{\varepsilon}(X;\theta)=s_{\varepsilon}(\widetilde{X};\hat{\theta})$ and $X_{-J_{\varepsilon}(X;\theta)}=\widetilde{X}_{-J_{\varepsilon}(\widetilde{X};\hat{\theta})}$. 
    We can choose $\delta>0$ sufficiently small (depending on $X, \widetilde{X}, \hat{\theta}, \varepsilon, Z$) such that $\forall \theta\in {\cal B}(\hat{\theta},\delta)$, the following holds: 
    \begin{align*}
        \text{ If } X_i-Z_i\hat{\theta} \text{ and }\widetilde{X}_i-Z_i\hat{\theta}>0, \text{ we have }X_i-Z_i{\theta} \text{ and } \widetilde{X}_i-Z_i{\theta}>0&\Rightarrow \ell_\tau(\widetilde{X}_i-{Z}_i^\T\theta)=\ell_\tau({X}_i-{Z}_i^\T\theta);\\
        \text{ If } X_i-Z_i\hat{\theta}  \text{ and } \widetilde{X}_i-Z_i\hat{\theta}<0, \text{ we have }X_i-Z_i{\theta}  \text{ and } \widetilde{X}_i-Z_i{\theta}<0&\Rightarrow \ell_\tau(\widetilde{X}_i-{Z}_i^\T\theta)=\ell_\tau({X}_i-{Z}_i^\T\theta);\\
        \text{ If } X_i-Z_i\hat{\theta}=\widetilde{X}_i-Z_i\hat{\theta}=0, \text{ then }X_i=\widetilde{X}_i&\Rightarrow \ell_\tau(\widetilde{X}_i-{Z}_i^\T\theta)=\ell_\tau({X}_i-{Z}_i^\T\theta).
    \end{align*}
    Therefore, we have $F_{\rm quantile}(\theta;X)=F_{\rm quantile}(\theta;\widetilde{X})$ for all $\theta\in {\cal B}(\hat{\theta},\delta)$. 
    This shows that the local optimality of $F_{\rm quantile}(\theta;X)$ at $\hat{\theta}$ implies the local optimality of $F_{\rm quantile}(\theta;\widetilde{X})$ at $\hat{\theta}$. 
    This finishes the proof.
\end{proof}
\section{Simulation details}\label{app: simulate}
In this section, we are going about something more about the simulations.
\subsection{Example 1: Behrens--Fisher problem with contamination}\label{app: BF}
Since
\[
X_1^{(0)}, \ldots, X_{n^{(0)}}^{(0)} \overset{\text{i.i.d.}}{\sim} {N}(\mu^{(0)}, \gamma^{(0)}), \quad
X_1^{(1)}, \ldots, X_{n^{(1)}}^{(1)} \overset{\text{i.i.d.}}{\sim} {N}(\mu^{(1)}, \gamma^{(1)}),
\]
with the two samples drawn independently. Under the null hypothesis 
$
H_0 : \mu^{(0)} = \mu^{(1)}$, 
the distributions can be parameterized by 
$\theta = (\mu, \gamma^{(0)}, \gamma^{(1)}) \in \Theta = \mathbb{R} \times \mathbb{R}_+ \times \mathbb{R}_+$, and the family \(\{P_\theta : \theta \in \Theta\}\) has density
\[
f(x; \theta) = f(x; (\mu, \gamma^{(0)}, \gamma^{(1)})) = \prod_{i=1}^{n^{(0)}} \frac{1}{\sqrt{2\pi \gamma^{(0)}}} e^{-(X_i^{(0)} - \mu)^2 / 2\gamma^{(0)}} \cdot
\prod_{i=1}^{n^{(1)}} \frac{1}{\sqrt{2\pi \gamma^{(1)}}} e^{-(X_i^{(1)} - \mu)^2 / 2\gamma^{(1)}}
\]
with respect to the Lebesgue measure on \(\mathcal{X} = \mathbb{R}^{n^{(0)} + n^{(1)}}\) and the gradient of negative log likelihood is
$$\left(\frac{1}{\gamma^{(0)}}\sum_{i=1}^{n^{(0)}}(\mu-X_i^{(0)})+\frac{1}{\gamma^{(1)}}\sum_{i=1}^{n^{(1)}}(\mu-X_i^{(1)})\quad \frac{1}{2\gamma^{(0)}}+\sum_{i=1}^{n^{(0)}}\frac{(X_i^{(0)}-\mu)^2}{2(\gamma^{(0)})^2}\quad\frac{1}{2\gamma^{(1)}}+\sum_{i=1}^{n^{(1)}}\frac{(X_i^{(1)}-\mu)^2}{2(\gamma^{(1)})^2}\right)^\T.$$

To generate the data, we take $n^{(0)}=n^{(1)}=50, \mu^{(0)}=0, \gamma^{(0)}=1, \gamma^{(1)}=2$, and $\mu^{(1)} \in$ $\{0,0.1,0.2, \ldots, 1\}$ (with $\mu^{(1)}=0$ corresponding to the case where the null hypothesis holds).
The test statistic $T$ (used both for aCSS and for the oracle) is given by the absolute difference in sample means between the two halves of the data. The first \( m = 5 \) data points are generated as \( 3 + |t_1| \), where \( t_1 \) denotes a \( t \)-distribution with 1 degree of freedom.

Since \cite{barber2022testing} has studied this example before, we don't need to check the assumptions in detail and generally the steps are the same except only for the selected samples. In this case, only the first dataset is contaminated and therefore the adjustment due to MTLE is only restricted to the first dataset The sampling distribution can be written as 
$$p_{\hat{\theta}}(\cdot\mid\hat{\theta})\propto f(x;\hat{\theta})\exp\left\{-\frac{d}{2\sigma^2}\left\|\nabla_\theta \log f(x;\theta)\right\|^2\right\} I(x),$$
where $I(x)$ denotes the corresponding check function.

Here we use the weighted sampling method to generate copies. One simple idea is to directly sample copies \( \widetilde{X}_j, j=1,\dots,n^{(k)} \) i.i.d. from \( f(x;\hat{\theta}) \) like the parametric bootstrap. However, in this case, the weight  
\[
w(\widetilde{X}) = \exp\left\{-\frac{d}{2\sigma^2}\left\|\nabla_\theta \log f(\widetilde{X};\theta)\right\|^2\right\}I(\widetilde{X})
\]
will be small, leading to lower power.  To improve power, instead of independent sampling, we sample copies \( \widetilde{X}^{(i)} \) from \( f(x;\hat{\theta}) \) while ensuring that  
\[
\sum_{i=1}^{n^{(k)}}\left(\widetilde{X}_i^{(k)}\right)^2 = \sum_{i=1}^{n^{(k)}}\left(\widetilde{X}_i^{(0)}\right)^2, \quad k=0,1.
\]
Then, we still use \( w(\widetilde{X}) \) as the weighting function. This construction satisfies the requirement in \eqref{weighted sampling} and enhances the power of the test. 

We also studied the performance of the aCSS method with MTLE for different values of \( m \) and \( h \). Specifically, we employed two versions of the aCSS method with MTLE: one using \( h = 50 - m \), denoted as aCSS(MTLE\_o), and the other using a fixed \( h = 45 \), denoted as aCSS(MTLE). We report the results for \( m = 1, 2, 3, 4 \) in the Figure \ref{fig: diff h}.

\begin{figure}
\centering
\includegraphics[width=0.45\textwidth]{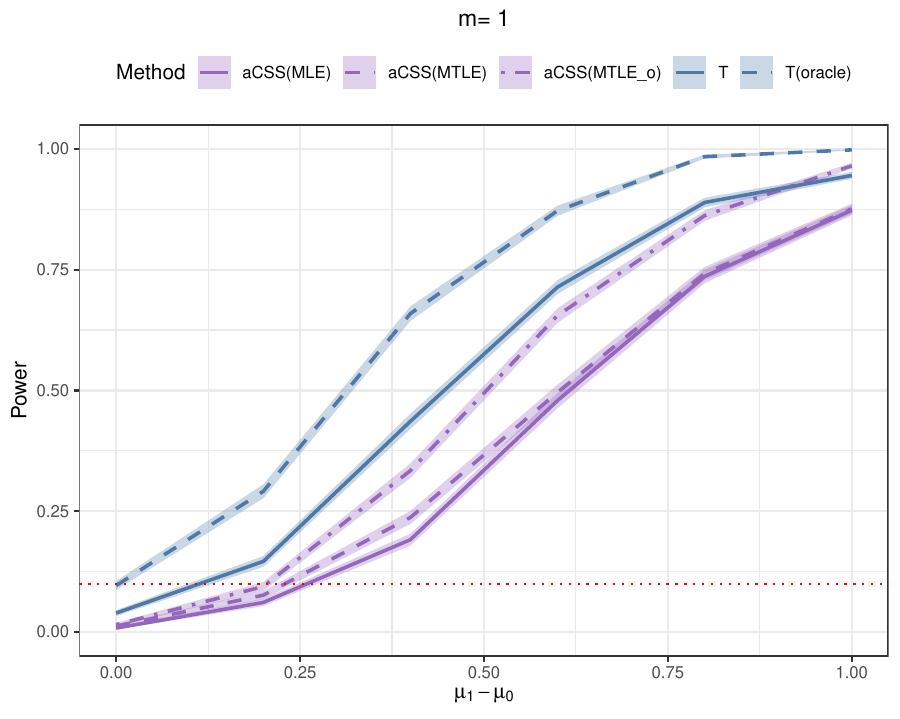}
\includegraphics[width=0.45\textwidth]{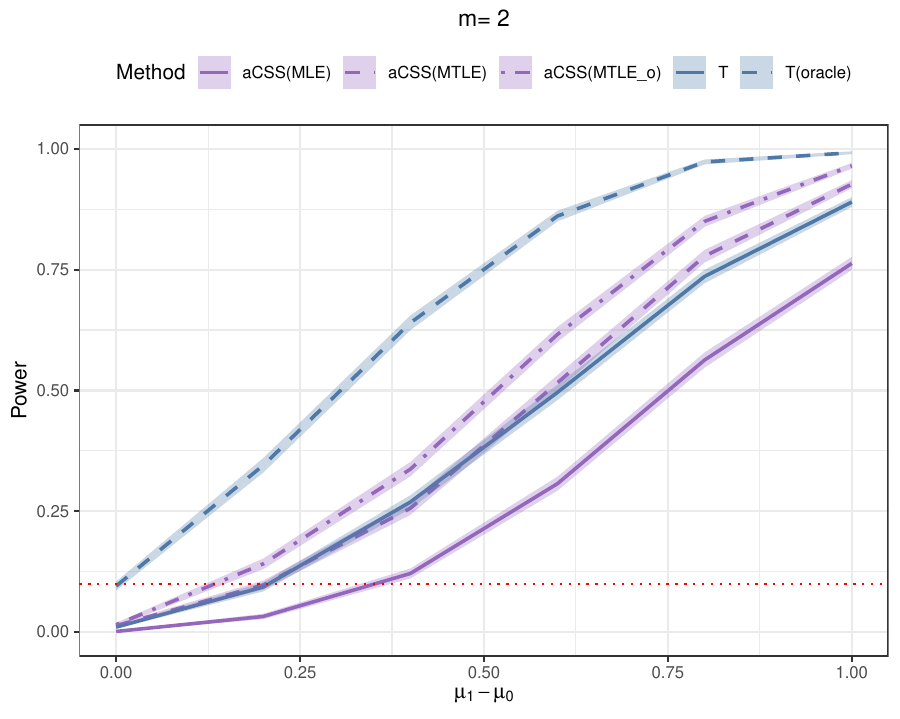}
\includegraphics[width=0.45\textwidth]{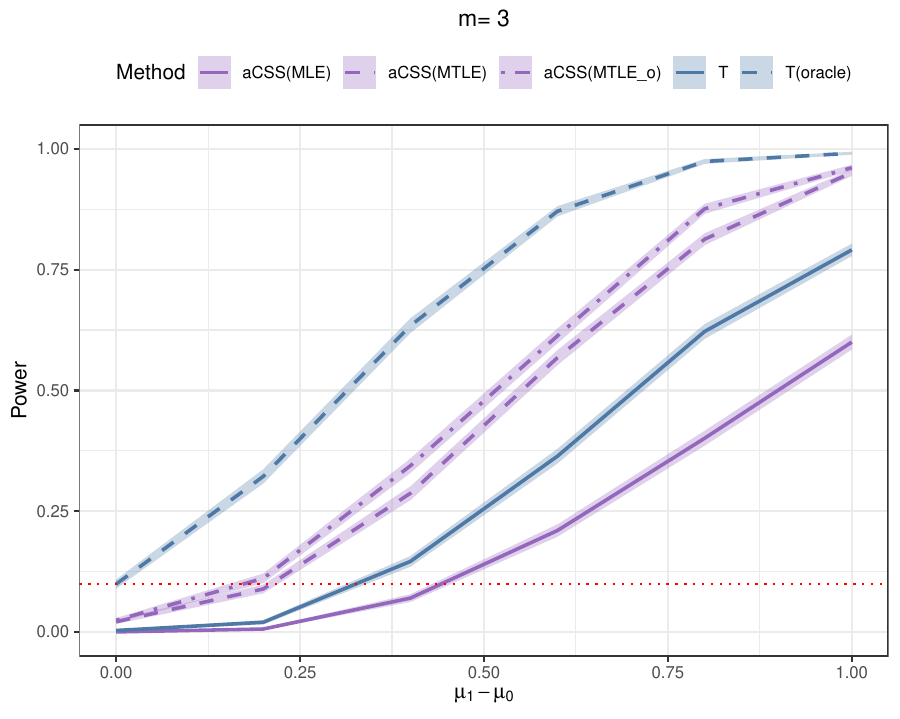}
\includegraphics[width=0.45\textwidth]{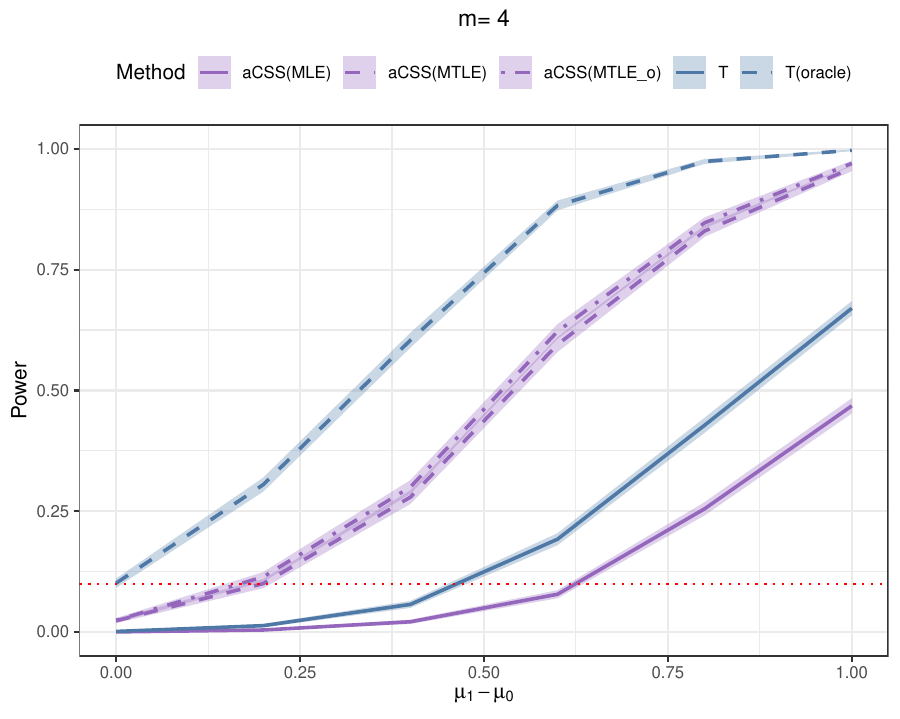}
\caption{Power of aCSS methods for differnt values of $m$}
\label{fig: diff h}
\end{figure}

\subsection{Example 2: aCSS CRT: conditional independence test}
We set \( n = 50, d = 200 , \nu = 1 \), and \( \theta_0 = (1.5,1.5,1.5,1.5,1.5, \dots, 0) \). The  
covariate matrix \( Z \in \mathbb{R}^{n \times d} \) is generated with i.i.d. \( N(0,1) \) entries, and we draw \( X \mid Z \sim N(Z\theta_0, \nu^2 I_n) \). The random vector \( Y \in \mathbb{R}^n \) is then generated with each entry \( Y_i \) drawn as  
\[
    Y_i \mid X_i, Z_i \sim N \left( \beta_0 X_i + \sum_{j=1}^{5} 0.2 Z_{i,j}, 1 \right).
\]
We consider \( \beta_0 \in \{0, 0.2, 0.4, \dots, 1\} \) with \( \beta_0 \) corresponding to the setting where \( Y \perp X \mid Z \).  
Formally, our null hypothesis is given by assuming that \( X \mid Y,Z \sim N(Z\theta, \nu^2 I_n) \) for some  
\( \theta \in \Theta = \mathbb{R}^d \). If \( \beta_0 \neq 0 \), then this null does not hold.

The debiased lasso or aCSS method with lasso performs poorly in this setting, likely due to the high estimation error in lasso for both estimating \( \xi \) and \( \theta_0 \). The aCSS method is implemented with $\sigma=0.7$. We also perform the aCSS method with different values of $\sigma$ and plot the Type-I error in the Figure \ref{fig: Type-I error}. Generally, the Type-I error of the aCSS methods increases as \( \sigma \) increases, which aligns with the earlier discussion.

\begin{figure}[ht]
\centering
\includegraphics[width=0.75\textwidth]{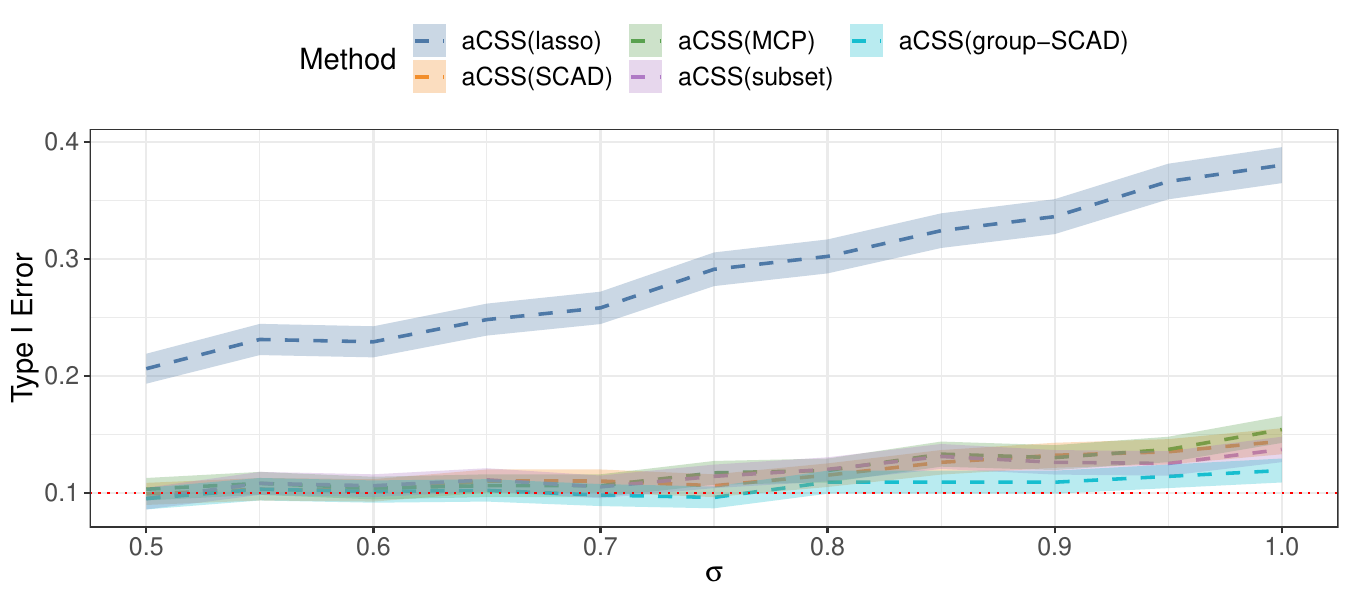}
\caption{Type I error rate of aCSS}
\label{fig: Type-I error}
\end{figure}
\end{document}